\documentclass{lmcs}
\pdfoutput=1
\usepackage[utf8]{inputenc}

\usepackage{lastpage}
\lmcsdoi{19}{3}{3}
\lmcsheading{}{\pageref{LastPage}}{}{}%
{Apr.~19,~2022}{Jul.~12,~2023}{}

\pdfoutput=1

\usepackage{namespc}
\usepackage{etoc}
\let\qualified\:
\let\:\medspace
\usepackage{amsmath}
\usepackage{amssymb}
\usepackage{thmtools}
\usepackage{hyperref}
\usepackage{cleveref}
\usepackage{cancel}
\usepackage{float}
\let\evil\cancel \undef\cancel

\crefname{thm}{Theorem}{Theorems}
\Crefname{thm}{Theorem}{Theorems}
\crefname{lem}{Lemma}{Lemmas}
\Crefname{lem}{Lemma}{Lemmas}
\crefname{defi}{Definition}{Definitions}
\Crefname{defi}{Definition}{Definitions}
\crefname{prop}{Proposition}{Propositions}
\Crefname{prop}{Proposition}{Propositions}

\usepackage{todonotes,regexpatch}
\makeatletter
\xpatchcmd{\@todo}{\setkeys{todonotes}{#1}}{\setkeys{todonotes}{inline,#1}}{}{}
\makeatother
\definecolor{ao}{rgb}{0.0,0.5,0.0}
\definecolor{amethyst}{rgb}{0.6,0.4,0.8}
\definecolor{arsenic}{rgb}{0.23,0.27,0.29}

\newcommand{\deprecated}[1]{\todo[textcolor=black,backgroundcolor=lipicsGray]{\textbf{Deprecated}: #1}}

\newcommand{\secref}[1]{\S\ref{#1}}

\usepackage{xspace}
\usepackage[strict]{changepage}
\usepackage{subcaption}
\usepackage{multicol}
\usepackage[all]{foreign}

\usepackage[Symbol]{upgreek}
\usepackage{amsfonts}
\usepackage{centernot}
\usepackage{stmaryrd}
\usepackage{physics}
\usepackage{cmll}
\usepackage{csquotes}
\usepackage{mathtools}
\usepackage{tikz-cd}
\usetikzlibrary{arrows,fit,positioning,automata,shapes,shapes.geometric, matrix}
\newcommand{\button}[1]{\raisebox{.5pt}{\textcircled{\raisebox{-.9pt} {#1}}}}

\newcommand\fgHGV[0]{$\text{HGV}*$\xspace}

\newcommand{\tmcolor}[0]{\color[HTML]{a40038}}
\newcommand{\tycolor}[0]{\color[HTML]{00007a}}
\newcommand{\tm}[1]{\ensuremath{{\tmcolor#1}}}
\newcommand{\ty}[1]{\ensuremath{{\tycolor#1}}}
\newcommand{\tmty}[2]{\ensuremath{\tm{#1}:\ty{#2}}}

\newcommand\bnfdef{\Coloneqq}
\newcommand\sep{\;\mid\;}
\newcommand\subst[4][]{\ifstrempty{#1}{\ensuremath{#2\{#3/#4\}}}{\ensuremath{#2(\{#3/#4\}\cup#1)}}}
\newcommand\plug[2]{\ensuremath{#1[#2]}}

\newcommand\defeq[0]{\triangleq}
\newcommand\hole[0]{\ensuremath{\square}}
\DeclareMathOperator{\fv}{fv}
\DeclareMathOperator{\fn}{fn}
\DeclareMathOperator{\bn}{bn}
\DeclareMathOperator{\cn}{cn}

\newcommand\conf[1]{\mathcal{#1}}
\newcommand\rel[1]{\mathcal{#1}}
\newcommand\config[1]{\mathcal{#1}}
\newcommand\taux{\mathcal{A}}
\newcommand\tany{\mathcal{T}}
\newcommand\hyper[1]{\mathcal{#1}}
\newcommand\variable[1]{\mathit{#1}}
\newcommand\deriv[1]{\mathbf{#1}}
\newcommand\mkwd[1]{\ensuremath{\mathsf{#1}}}
\newcommand\envs[1]{\mkwd{envs}(#1)} \newcommand\size{\abs}
\newcommand\graphvar{\hyper}

\newenvironment{case}[1]{\flushleft{\bfseries Case}~(#1).}{\endflushleft}
\newenvironment{subcase}[1]{\flushleft{\bfseries Subcase}~(#1).}{\endflushleft}

\usepackage{mathpartir}
\newcommand{\rulename}[1]{\LabTirName{#1}}

\newcommand\emptyenv{\cdot}
\newcommand\hypersep{\parallel}
\newcommand\emptyhyperenv{\varnothing}

\newcommand\fgcbv[1]{\llparenthesis{#1}\rrparenthesis}
\newcommand\gvcp[2]{\llbracket{#1}\rrbracket_{\tm{#2}}}
\newcommand\gvcplab{\mathsf}

\newcommand\gvcpval[2]{\gvcp{#1}{#2}^{\!\:{\gvcplab{v}}}} \newcommand\gvcpcom[2]{\gvcp{#1}{#2}^{\!\:{\gvcplab{m}}}} \newcommand\gvcpcnf[2]{\gvcp{#1}{#2}^{\!\:{\gvcplab{c}}}} \newcommand\gvcpevc[2]{\gvcp{#1}{#2}^{\!\:{\gvcplab{f}}}} \DeclarePairedDelimiter\gvcpup{\llceil}{\rrceil}
\DeclarePairedDelimiter\gvcpdown{\llfloor}{\rrfloor}

\newcommand{\flatten}[1]{\mathop{\downarrow}{#1}}

\newcommand{\ta}[0]{\alpha}
\newcommand{\tb}[0]{\beta}

\newcommand\bl{\begin{array}[t]{@{}l@{}}}
\newcommand\blx{\begin{array}{@{}l@{}}}
\newcommand\el{\end{array}}

\newcommand\header[2][]{\begin{flushleft}\textbf{#2}\hfill{#1}\end{flushleft}}
\newcommand\headerarg[2]{\header[#2]{#1}}
\newcommand\headersig[2]{\headerarg{#1}{\framebox{#2}}}

\newcommand{\shade}[1]{\setlength{\fboxsep}{0pt}\colorbox{lightgray}{\ensuremath{#1}}}

\namespace*{hcp}{
\providecommand{\tyone}{\ensuremath{\mathbf{1}}}
  \providecommand{\tynil}{\ensuremath{\mathbf{0}}}
  \providecommand{\tytop}{\ensuremath{\top}}
  \providecommand{\tybot}{\ensuremath{\bot}}
  \providecommand{\typlusop}{\ensuremath{\oplus}}
  \providecommand{\tywithop}{\ensuremath{\with}}
  \providecommand{\tytensop}{\ensuremath{\otimes}}
  \providecommand{\typarrop}{\ensuremath{\parr}}
  \providecommand{\typlus}[2]{\ensuremath{{#1}\mathbin{\typlusop}{#2}}}
  \providecommand{\tywith}[2]{\ensuremath{{#1}\mathbin{\tywithop}{#2}}}
  \providecommand{\tytens}[2]{\ensuremath{{#1}\mathbin{\tytensop}{#2}}}
  \providecommand{\typarr}[2]{\ensuremath{{#1}\mathbin{\typarrop}{#2}}}
  \providecommand{\co}[1]{\ensuremath{#1^\bot}}
\providecommand{\res}[4][]{\ensuremath{(\nu{#2}^{#1}{#3}){#4}}}
  \providecommand{\ppar}[2]{\ensuremath{#1\parallel#2}}
  \providecommand{\halt}[0]{\ensuremath{\mathbf{0}}}
  \providecommand{\link}[3][]{\lablink[#1]{#2}{#3}}
  \providecommand{\send}[3]{\ensuremath{\labsend{#1}{#2}.#3}}
  \providecommand{\usend}[3]{\ensuremath{#1\langle{#2}\rangle.#3}}
  \providecommand{\recv}[3]{\ensuremath{\labrecv{#1}{#2}.#3}}
  \providecommand{\close}[2]{\ensuremath{\labclose{#1}.#2}}
  \providecommand{\wait}[2]{\ensuremath{\labwait{#1}.#2}}
  \providecommand{\inl}[2]{\ensuremath{\labselinl{#1}.{#2}}}
  \providecommand{\inr}[2]{\ensuremath{\labselinr{#1}.{#2}}}
  \providecommand{\offer}[3]{\ensuremath{{#1}\triangleright\{\text{inl}:#2;\text{inr}:#3\}}}
  \providecommand{\absurd}[1]{\ensuremath{{#1}\triangleright\{\}}}
  \providecommand{\ping}[2]{\ensuremath{\labping{#1}.#2}}
  \providecommand{\pong}[2]{\ensuremath{\labpong{#1}.#2}}

\providecommand{\seq}[2]{\ensuremath{\tm{#1}\vdash{#2}}}

  \providecommand{\sbis}[0]{\sim}
  \providecommand{\bis}[0]{\approx}

  \providecommand{\slto}[2][]{\ensuremath{\overset{#2}\Longrightarrow_{#1}}}

  \providecommand{\stto}[1][]{\slto[#1]{\tau}}
  \providecommand{\lto}[1]{\ensuremath{\overset{\tm{#1}}\longrightarrow}}
  \providecommand{\ato}[0]{\lto{\ta}}
  \providecommand{\bto}[0]{\lto{\tb}}
  \providecommand{\tto}[0]{\lto{\tau}}
  
  \providecommand{\lablink}[3][]{\ensuremath{{#2}{\leftrightarrow}^{#1}{#3}}}
  \providecommand{\labsend}[2]{\ensuremath{#1[#2]}}
  \providecommand{\labrecv}[2]{\ensuremath{#1(#2)}}
  \providecommand{\labclose}[1]{\labsend{#1}{}}
  \providecommand{\labwait}[1]{\labrecv{#1}{}}
  \providecommand{\labinl}[0]{\text{inl}}
  \providecommand{\labinr}[0]{\text{inr}}
  \providecommand{\labselinl}[1]{\ensuremath{{#1}\triangleleft\labinl}}
  \providecommand{\labselinr}[1]{\ensuremath{{#1}\triangleleft\labinr}}
  \providecommand{\laboffinl}[1]{\ensuremath{{#1}\triangleright\labinl}}
  \providecommand{\laboffinr}[1]{\ensuremath{{#1}\triangleright\labinr}}
  \providecommand{\labping}[1]{\ensuremath{\bar{#1}}}
  \providecommand{\labpong}[1]{\ensuremath{#1}}
}{}
\newcommand{\hcp}[1]{\namespace*{hcp}{}{#1}}

\namespace*{hgv}{
  
  \providecommand{\app}{\;}
    
  \providecommand{\seq}{\overrightarrow}

    \providecommand{\tyunit}[0]{\ensuremath{\mathbf{1}}}
  \providecommand{\tyvoid}[0]{\ensuremath{\mathbf{0}}}
  
  \providecommand{\typrod}[2]{\ensuremath{{#1}\mathbin{\times}{#2}}}
  \providecommand{\tysum}[2]{\ensuremath{{#1}\mathbin{+}{#2}}}
  \providecommand{\tylolli}[2]{\ensuremath{{#1}\mathbin{\multimap}{#2}}}

  \providecommand{\co}[1]{\ensuremath{\overline{#1}}}
\providecommand{\tysend}[2]{\ensuremath{!{#1}.{#2}}}
  \providecommand{\tyrecv}[2]{\ensuremath{?{#1}.{#2}}}
  \providecommand{\tyend}[0]{\ensuremath{\mathbf{end}}}
  \providecommand{\tyends}[0]{\ensuremath{\mathbf{end}_{!}}}
  \providecommand{\tyendr}[0]{\ensuremath{\mathbf{end}_{?}}}
  \providecommand{\tyselect}[2]{\ensuremath{{#1}\mathbin{\oplus}{#2}}}
  \providecommand{\tyoffer}[2]{\ensuremath{{#1}\mathbin{\with}{#2}}}
  \providecommand{\tyselectemp}[0]{\ensuremath{{\oplus}\{\}}}
  \providecommand{\tyofferemp}[0]{\ensuremath{{\with}\{\}}}
\providecommand{\ret}[1]{#1} \providecommand{\lablet}[0]{\ensuremath{\mathbf{let}}}
  \providecommand{\labin}[0]{\ensuremath{\mathbf{in}}}
  \providecommand{\letbind}[3]{\ensuremath{\lablet\;#1=#2\;\labin\;#3}}
  \providecommand{\letbindtwo}[2]{\letbind{#1}{#2}{}} \providecommand{\pair}[2]{\ensuremath{(#1,#2)}}
  \providecommand{\letpair}[4]{\ensuremath{\letbind{\pair{#1}{#2}}{#3}{#4}}}
  \providecommand{\labinl}[0]{\ensuremath{\mathbf{inl}}}
  \providecommand{\labinr}[0]{\ensuremath{\mathbf{inr}}}
  \providecommand{\inl}[1]{\ensuremath{\labinl\;#1}}
  \providecommand{\inr}[1]{\ensuremath{\labinr\;#1}}
  \providecommand{\casesum}[5]{\ensuremath{\mathbf{case}\;#1\left\{\inl{#2}\mapsto{#3};\;\inr{#4}\mapsto{#5}\right\}}}
  \providecommand{\unit}[0]{\ensuremath{()}}
  \providecommand{\andthen}[2]{\ensuremath{#1;#2}}
  \providecommand{\letunit}[2]{\letbind{\unit}{#1}{#2}}
  \providecommand{\absurd}[1]{\ensuremath{\mathbf{absurd}\;#1}}
  \providecommand{\link}[0]{\ensuremath{\mathbf{link}}}
  \providecommand{\gvlink}[2]{\link\;(#1, #2)}
  \providecommand{\calcwd}[1]{\ensuremath{\mathbf{#1}}}
  \providecommand{\new}[0]{\ensuremath{\mathbf{new}}}
  
  \providecommand{\halt}[0]{\ensuremath{\mathbf{halt}}}
  \providecommand{\spawn}[0]{\ensuremath{\mathbf{spawn}}}
  \providecommand{\send}[0]{\ensuremath{\mathbf{send}}}
  \providecommand{\gvsend}[2]{\send\;(#1, #2)}
  \providecommand{\recv}[0]{\ensuremath{\mathbf{recv}}}
  \providecommand{\fork}[0]{\ensuremath{\mathbf{fork}}}
  \providecommand{\halt}[0]{\ensuremath{\mathbf{halt}}}
  \providecommand{\cancel}[0]{\ensuremath{\mathbf{cancel}}}
  \providecommand{\zap}[1]{\ensuremath{\lightning #1}}
  \providecommand{\raiseexn}[0]{\ensuremath{\mathbf{raise}}}
  \providecommand{\tryasinotherwise}[4]{\ensuremath{\mathbf{try}\;#1\;\mathbf{as}\;#2\;\mathbf{in}\;#3\;\mathbf{otherwise}\;#4}}
\providecommand{\wait}[0]{\ensuremath{\mathbf{wait}}}
  \providecommand{\close}[0]{\ensuremath{\mathbf{close}}}
  \providecommand{\labselect}[0]{\ensuremath{\mathbf{select}}}
  \providecommand{\select}[1]{\ensuremath{\labselect\;#1}}
  \providecommand{\laboffer}[0]{\ensuremath{\mathbf{offer}}}
  \providecommand{\offer}[5]{\ensuremath{\laboffer\;#1\;\{\inl{#2}\mapsto{#3};\inr{#4}\mapsto{#5}\}}}
  \providecommand{\offeremp}[1]{\ensuremath{\mathbf{offer}\;#1\;\{\}}}
  \providecommand{\isect}[0]{\sqcap}
\providecommand{\ppar}[2]{\ensuremath{#1\parallel#2}}
  \providecommand{\res}[3]{\ensuremath{(\nu#1#2)#3}}
  
\providecommand{\main}[0]{\ensuremath{\bullet}}
  \providecommand{\child}[0]{\ensuremath{\circ}}
  \providecommand{\tychild}[0]{\ensuremath{\child}}
  \providecommand{\tymain}[1]{\ensuremath{\main\;#1}}
\providecommand{\tseq}[3]{\ensuremath{#1\vdash\tmty{#2}{#3}}}
  \providecommand{\cseq}[4][]{\ensuremath{#2\vdash^{#1}\tmty{#3}{#4}}}
  
\providecommand{\tred}[0]{\ensuremath{\longrightarrow_{\mathsf{M}}}}
  \providecommand{\cred}[0]{\ensuremath{\longrightarrow}}

  \providecommand{\blocked}[2]{\ensuremath{\mathsf{blocked}(\tm{#1},\tm{#2})}}
  \providecommand{\lock}[2]{\ensuremath{\langle{#1,#2}\rangle}}
  \providecommand{\lockty}[1]{\ensuremath{#1^\sharp}}
  \providecommand{\linkconfig}[3]{\ensuremath{#2{\overset{#1}{\leftrightarrow}}#3}}
\providecommand{\mix}[0]{\textsc{Mix}\xspace}
  \providecommand{\hyperg}[0]{\hyper{G}}
  
  \providecommand{\mixconf}[1][\phi]{\ensuremath{\mathcal{M}^\phi}\xspace}
  \providecommand{\mixrule}{\LabTirName{TC-Mix}\xspace}
  
  \providecommand{\eat}[1]{\ensuremath{\mathbf{eat}\;#1}}
  \providecommand{\hungry}[0]{\ensuremath{\mathbf{hungry}}}
  \providecommand{\letnew}[3]{\letbind{\langle#1,#2\rangle}{\new}{#3}}
  \providecommand{\letspawn}[2]{\letbind{\langle\rangle}{\spawn\;#1}{#2}}
}{}
\newcommand{\hgv}[1]{\namespace*{hgv}{}{#1}}

\namespace*{cp}{\providecommand{\co}[1]{\ensuremath{#1^\bot}}
  \providecommand{\cut}[3]{\ensuremath{(\nu#1)(#2\parallel#3)}}
  \providecommand{\send}[4]{\ensuremath{#1[#2].(#3\parallel#4)}}
}{}
\let\cp\undefined \newcommand{\cp}[1]{\namespace*{cp}{}{#1}}

\namespace*{gv}{\providecommand{\tyunit}[0]{\ensuremath{\mathbf{1}}}
  \providecommand{\tyvoid}[0]{\ensuremath{\mathbf{0}}}
  \providecommand{\typrod}[2]{\ensuremath{{#1}\mathbin{\times}{#2}}}
  \providecommand{\tysum}[2]{\ensuremath{{#1}\mathbin{+}{#2}}}
  \providecommand{\tylolli}[2]{\ensuremath{{#1}\mathbin{\multimap}{#2}}}
  \providecommand{\co}[1]{\ensuremath{\overline{#1}}}
  \providecommand{\tysend}[2]{\ensuremath{!{#1}.{#2}}}
  \providecommand{\tyrecv}[2]{\ensuremath{?{#1}.{#2}}}
  \providecommand{\tyends}[0]{\ensuremath{\mathbf{end}_{!}}}
  \providecommand{\tyendr}[0]{\ensuremath{\mathbf{end}_{?}}}
  \providecommand{\andthen}[2]{\ensuremath{#1;#2}}
  \providecommand{\letbind}[3]{\ensuremath{\mathbf{let}\;#1\mathbin{=}#2\;\mathbf{in}\;#3}}
  \providecommand{\pair}[2]{\ensuremath{(#1,#2)}}
  \providecommand{\letpair}[4]{\ensuremath{\letbind{\pair{#1}{#2}}{#3}{#4}}}
  \providecommand{\labinl}[0]{\ensuremath{\mathbf{inl}}}
  \providecommand{\labinr}[0]{\ensuremath{\mathbf{inr}}}
  \providecommand{\unit}[0]{\ensuremath{()}}
  \providecommand{\letunit}[2]{\ensuremath{\letbind{\unit}{#1}{#2}}}
  \providecommand{\absurd}[1]{\ensuremath{\mathbf{absurd}\;#1}}
  \providecommand{\new}[0]{\ensuremath{\mathbf{new}}}
  \providecommand{\spawn}[0]{\ensuremath{\mathbf{spawn}}}
  \providecommand{\send}[0]{\ensuremath{\mathbf{send}}}
  \providecommand{\recv}[0]{\ensuremath{\mathbf{recv}}}
  \providecommand{\fork}[0]{\ensuremath{\mathbf{fork}}}
  \providecommand{\wait}[0]{\ensuremath{\mathbf{wait}}}
  \providecommand{\close}[0]{\ensuremath{\mathbf{close}}}
  \providecommand{\main}[0]{\ensuremath{\bullet}}
  \providecommand{\child}[0]{\ensuremath{\circ}}
  \providecommand{\tychild}[0]{\ensuremath{\child}}
  \providecommand{\tymain}[1]{\ensuremath{\main\;#1}}
  \providecommand{\ppar}[3][]{\ensuremath{#2\mathbin{\parallel_{#1}}#3}}
  \providecommand{\parannot}[2]{\parallel_{\langle {#1}, {#2} \rangle}}
  \providecommand{\resd}[3]{\ensuremath{(\nu#1#2)#3}}
  \providecommand{\res}[2]{\ensuremath{(\nu#1)#2}}
  \providecommand{\lock}[2]{\ensuremath{\langle{#1,#2}\rangle}}
  \providecommand{\lockty}[1]{\ensuremath{#1^\sharp}}
  \providecommand{\cseqold}[3][]{\ensuremath{#2\vdash^{#1}\tm{#3}}}
  \providecommand{\gvseq}[4][]{\ensuremath{#2\vdash_{\text{\textsf{GV}}}^{#1}\tmty{#3}{#4}}}
  \providecommand{\nseq}[4][]{\ensuremath{#2\not\vdash_{\textsf{GV}}^{#1}\tmty{#3}{#4}}}
  \providecommand{\linkconfig}[3]{\ensuremath{#2{\overset{#1}{\leftrightarrow}}#3}}
  \providecommand{\isect}[0]{\sqcap}
}{}
\newcommand{\gv}[1]{\namespace*{gv}{}{#1}}

\usepackage{xspace}
\usepackage{lscape}

\newcommand{\kleeneiff}{\stackrel{\bumpeq}{\Longleftrightarrow}}
\newcommand{\treeprop}[2]{\mkwd{Tree}(#1, #2)}
\newcommand{\nameset}{\mathcal{N}}
\newcommand{\lblset}{\mathcal{L}}
\newcommand{\lbl}{l}
\newcommand{\translbl}{\ell}

\newcommand{\tcf}{\mathcal{F}}
\newcommand{\hr}{\mathsf{hr}} 

\usepackage{ragged2e} % justified text in case environment
\let\origcase\case
\def\case#1{\origcase{#1}\justifying~}
\let\origsubcase\subcase
\def\subcase#1{\origsubcase{#1}\justifying~}

\begin{document}
\title{Separating Sessions Smoothly}

\author[S.~Fowler]{Simon Fowler\lmcsorcid{0000-0001-5143-5475}}[a]

\author[W.~Kokke]{Wen Kokke\lmcsorcid{0000-0002-1360-4714}}[b]

\author[O.~Dardha]{Ornela Dardha\lmcsorcid{0000-0001-9927-7875}}[a]

\author[S.~Lindley]{Sam Lindley\lmcsorcid{0000-0002-1360-4714}}[c]

\author[J.~G.~Morris]{J.\ Garrett Morris\lmcsorcid{0000-0002-3992-1080}}[d]

\address{University of Glasgow, UK}\email{simon.fowler@glasgow.ac.uk, ornela.dardha@glasgow.ac.uk}

\address{University of Strathclyde, UK}\email{wen.kokke@strath.ac.uk}

\address{The University of Edinburgh, UK}\email{sam.lindley@ed.ac.uk}

\address{The University of Iowa, USA}\email{garrett-morris@uiowa.edu}

\etocdepthtag.toc{mtchapter}
%\etocsettagdepth{mtchapter}{subsection} % This was not necessary
%\etocsettagdepth{mtappendix}{none}

\begin{abstract}
This paper introduces Hypersequent GV (HGV), a modular and extensible
core calculus for functional programming with session types that
enjoys deadlock freedom, confluence, and strong normalisation.
HGV exploits hyper-environments, which are collections of type environments, to
ensure that structural congruence is type preserving.
As a consequence we obtain an operational correspondence between HGV
and HCP---a process calculus based on hypersequents and in a
propositions-as-types correspondence with classical linear logic
(CLL).  Our translations from HGV to HCP and vice-versa both preserve
and reflect reduction.
HGV scales smoothly to support Girard's Mix rule, a crucial ingredient
for channel forwarding and exceptions.
\end{abstract}

\maketitle

{\section{Introduction}
\label{sec:introduction}

Session types~\cite{honda93,takeuchihk94,hondavk98} are types used to
model and verify communication protocols in concurrent and distributed systems:
just as data types rule out dividing an integer by a string, session
types rule out sending an unexpected message.  Session types
originated in process calculi, but there is a gap between process
calculi, which model the evolving state of concurrent systems, and the
descriptions of these systems in mainstream programming languages.  This
paper addresses two foundations for session types: (1) a session-typed
concurrent lambda calculus called GV~\cite{lindleym15:semantics},
intended to be a modular and extensible basis for functional
programming languages with session types; and, (2) a session-typed
process calculus called CP~\cite{wadler14:sessions}, with a
propositions-as-types correspondence to classical linear logic
(CLL)~\cite{girard87:ll}.

Processes in CP correspond exactly to proofs in CLL and deadlock
freedom follows from cut-elimination for CLL.  However, while CP is
strongly tied to CLL, at the same time it departs from the
$\pi$-calculus. Independent $\pi$-calculus features can only appear in
combination in CP: CP combines name restriction with parallel
composition ($\cp{\tm{\cut{x}{P}{Q}}}$), corresponding to CLL's cut
rule, and combines sending (of bound names only) with parallel
composition ($\cp{\tm{\send{x}{y}{P}{Q}}}$), corresponding to CLL's
tensor rule.  This results in a proliferation of process constructors
and prevents the use of standard techniques from concurrency theory,
such as labelled-transition semantics and bisimulation, since the expected
transitions give rise to ill-typed terms. For example, we cannot write the expected transition
rule for output,
$\cp{\tm{\send{x}{y}{P}{Q}}} \xrightarrow{\hcp{\tm{\labsend{x}{y}}}} \tm{P \parallel
Q}$,
since $\tm{P \parallel Q}$ is not a valid CP process. A similar issue
arises when attempting to design a synchronisation transition rule for bound
output (see~\cite{kokkemp19:tlla} for a detailed discussion).
Inspired by Carbone~\emph{et al.}~\cite{carbonems18} who use
hypersequents~\cite{avron91:hypersequents} to give a logical grounding to
choreographic programming languages~\cite{montesi13},
Hypersequent CP
(HCP)~\cite{kokkemp19:popl,kokkemp19:tlla,montesip18:ct} restores the
independence of these features by factoring out parallel composition
into a standalone construct while retaining the close correspondence
with CLL proofs. HCP typing reasons about collections of processes
using collections of type environments (or \emph{hyper-environments}).

GV extends linear $\lambda$-calculus with constants for session-typed
communication. Following Gay and Vasconcelos~\cite{gayv10:last},
Lindley and Morris~\cite{lindleym15:semantics} describe GV's semantics
by combining a reduction relation on single terms, following standard
$\lambda$-calculus rules, and a reduction relation on concurrent
configurations of terms, following standard $\pi$-calculus rules.
They give a semantic characterisation of deadlocked processes, an
extrinsic~\cite{Reynolds00} type system for configurations, and show
that well-typed configurations are deadlock-free.  There is, however,
a large fly in this otherwise smooth ointment: GV's process equivalence
does not preserve typing.  As a result, it is not enough for Lindley
and Morris to show progress and preservation for well-typed
configurations; instead, they must show progress and preservation for
\emph{all} configurations \emph{equivalent to} well-typed
configurations. This not only complicates the metatheory of GV, but
the burden is inherited by any effort to build on GV's account of
concurrency~\cite{fowlerlmd19:stwt}.

In this paper, we show that using hyper-environments in the typing of
configurations enables a metatheory for GV that, compared to that of
Lindley and Morris, is simpler, is more general, and as a result is
easier to use and easier to extend.  Hypersequent GV (HGV) repairs the
treatment of process equivalence---equivalent configurations are
equivalently typeable---and avoids the need for formal gimmickry
connecting name restriction and parallel composition.  HGV admits
standard semantic techniques for concurrent programs: we use
bisimulation to show that our translations both preserve \emph{and} reflect
reduction, whereas Lindley and Morris resort to weak explicit
substitutions~\cite{LevyM99} and only show that their translations between GV
and CP preserve reduction.
HGV is also more easily
extensible: we outline three examples, including showing that HGV
naturally extends to disconnected sets of communication processes,
without any change to the proof of deadlock freedom, and that it
serves as a simpler foundation for existing work on exceptions in
GV~\cite{fowlerlmd19:stwt}.

\paragraph{Contributions}

The paper contributes the following:
\begin{itemize}
\item \Cref{sec:hgv} introduces Hypersequent GV (HGV), a modular and
  extensible core calculus for functional programming with session
  types which uses hyper-environments to ensure that structural
  congruence is type preserving.
\item \Cref{sec:relation-to-gv} shows that every well-typed GV
  configuration is also a well-typed HGV configuration, and every
  tree-structured HGV configuration is equivalent to a well-typed GV
  configuration.
\item \Cref{sec:relation-to-cp} gives an operational
  correspondences between HGV and HCP via translations in both
  directions that preserve and reflect reduction.
\item \Cref{sec:extensions} demonstrates the extensibility of HGV
  through: (1) unconnected processes, (2) a simplified treatment of
  forwarding, and (3) an improved foundation for exceptions.
\end{itemize}
\Cref{sec:problem-with-gv} reviews GV and its metatheory,
\Cref{sec:hgv-plus} discusses why it is difficult to apply hyper-environments to
term typing,
\Cref{sec:related-work} discusses related work, and
\Cref{sec:conclusion} concludes and discusses future work.

This paper is an improved and extended version of a paper published at
CONCUR 2021~\cite{FKDLM21}. Additional highlights include:
\begin{itemize}
\item
  a more detailed account of process structures;
\item
  a more detailed account of extensions;
\item
  a more detailed account of the metatheory for HCP; and
\item
  a modified formulation of HCP's labelled transition system and the
  translation of $\tm{\hgv\fork}$ in Section~\ref{sec:relation-to-cp}
  fixing errors in the operational correspondence result from the
  CONCUR 2021 paper.
\end{itemize}
Proofs of all of the technical results are included in the paper.

 }    
{\section{The Equivalence Embroglio}
\label{sec:problem-with-gv}
GV programs are deadlock free, which GV ensures by restricting process
structures to trees. A \emph{process structure} is an undirected
graph where nodes represent processes and edges represent channels
shared between the connected nodes.  Session-typed programs with an
acyclic process structure are deadlock-free by construction.
We illustrate this with a session-typed vending machine example
written in GV.
\begin{exa}
Consider the session type of a vending machine below, which sells
  chocolate bars and lollipops. If the vending machine is free, the
  customer can press \button1 to receive a chocolate bar or \button2 to
  receive a lollipop. If the vending machine is busy, the session
  ends.

\begingroup
\usingnamespace{hgv}
\setlength{\arraycolsep}{.2777em}
{\small
\begin{align*}
  &\ty{\mkwd{VendingMachine}}
  &&\defeq\ty{{\oplus}
     \left\{
     \begin{array}{lcl}
       \mkwd{Free} &:& {\&}
                       \left\{
                       \button1\;{:}\;\tysend{\mkwd{ChocolateBar}}{\tyends},
                       \button2\;{:}\;\tysend{\mkwd{Lollipop}}{\tyends}
                       \right\}
       \\
       \mkwd{Busy} &:& {\tyends}
     \end{array}
     \right\}}
\end{align*}
}
\endgroup

\noindent
The customer's session type is \emph{dual}: where the vending machine sends a $\mkwd{ChocolateBar}$, the customer receives a $\mkwd{ChocolateBar}$, and so forth. \Cref{fig:example1} shows the vending machine and customer as a GV program with its process structure.
\begin{figure}[h!]
  \centering
  \small
  \usingnamespace{hgv}
  \def\vendingMachine{\ensuremath{\tm{\begin{array}{l}
        \letbind{s}{\select{\mkwd{Free}}\;{s}}{}
        \\
        \lablet\;{s} = \laboffer\;{s}\;\left\{
        \begin{array}{lcl}
          \button1 & \mapsto & \send\;\mkwd{chocolateBar}
          \\
          \button2 & \mapsto & \send\;\mkwd{lollipop}
        \end{array}
        \right\}
        \\
        \close\;{s}
      \end{array}
    }}}
  \def\customer{\ensuremath{\tm{\begin{array}{l}
        \laboffer\;{s}\;\left\{
        \begin{array}{lcl}
          \mkwd{Free} & \mapsto & \letbind{s}{\select{\button1}\;{s}}{}
          \\
                      &         & \letpair{\mkwd{cb}}{s}{\recv\;{s}}{}
          \\
                      &         & \wait\;{s}; \eat{cb}
          \\
          \mkwd{Busy} & \mapsto & \wait\;{s}; \hungry
        \end{array}
        \right\}
      \end{array}
    }}}
  \begin{minipage}[b]{0.6\linewidth}
      \small
\centering
    \(\tm{\setlength{\arraycolsep}{2pt}
    \begin{array}{l}
      \lablet\;\mkwd{vendingMachine} = \lambda{s}.
      \\
      \quad\vendingMachine
      \\
      \labin\;\lablet\;\mkwd{customer} = \lambda{s}.
      \\
      \quad\customer
      \\
      \labin\;\lablet\;{s} = {\fork\;(\lambda{s}.\mkwd{vendingMachine}\;{s})}
      \\
      \labin\;\mkwd{customer}\;{s}
    \end{array}
    }\)
    \subcaption{Vending machine and customer as a GV program.}
    \label{fig:example1a}
  \end{minipage}\begin{minipage}[b]{0.4\linewidth}
    \centering
    \begin{tikzpicture}
      \node[draw] (vendingMachine) {$\tm{\mkwd{vendingMachine}}$};
      \node[draw, below=2cm of vendingMachine] (customer) {$\tm{\mkwd{customer}}$};
      \path (vendingMachine) edge node[pos=0.2, right] {$\tm{s}$} node[pos=0.8, right] {$\tm{s}$} (customer);
    \end{tikzpicture}
    \vspace*{1cm}
    \subcaption{Process structure of \Cref{fig:example1a}.}
    \label{fig:example1b}
  \end{minipage}
  \caption{Example program with acyclic process structure.}
  \label{fig:example1}
\end{figure}
\end{exa}

GV establishes the restriction to tree-structured processes by restricting the primitive for spawning processes. In GV, $\gv{\tm{\fork}}$ has type $\gv{\ty{\tylolli{(\tylolli{S}{\tyends})}{\co{S}}}}$. It takes a closure of type $\gv{\ty{\tylolli{S}{\tyends}}}$ as an argument, creates a channel with endpoints of dual types $\gv{\ty{S}}$ and $\gv{\ty{\co{S}}}$, spawns the closure as a new process by supplying one of the endpoints as an argument, and then returns the other endpoint. In essence, $\gv{\tm{\fork}}$ is a branching operation on the process structure: it creates a new node connected to the current node by a single edge. Linearity guarantees that the tree structure is preserved, even in the presence of higher-order channels.

Lindley and Morris~\cite{lindleym15:semantics} introduce a semantics for GV, which evaluates programs embedded in process configurations, consisting of embedded programs, flagged as main ($\gv{\tm{\main}}$) or child ($\gv{\tm{\child}}$) threads, $\nu$-binders to create new channels, and parallel compositions:
\[
  \usingnamespace{gv}
  \tm{\conf{C}}, \tm{\conf{D}} \Coloneqq
  \tm{\main\;M} \sep \tm{\child\;M} \sep \tm{\res{x}{\conf{C}}} \sep \tm{(\ppar{\conf{C}}{\conf{D}})}
\]
They introduce these process configurations together with a standard structural congruence, which allows, amongst other things, the reordering of processes using commutativity ($\gv{\tm{{\ppar{\conf{C}}{\conf{C'}}}\equiv{\ppar{\conf{C'}}{\conf{C}}}}}$), associativity ($\gv{\tm{{\ppar{\conf{C}}{(\ppar{\conf{C'}}{\conf{C''}})}}\equiv{\ppar{(\ppar{\conf{C}}{\conf{C'}})}{\conf{C''}}}}}$), and scope extrusion
($\gv{\tm{{\ppar{\conf{C}}{\res{x}{\conf{C'}}}}\equiv{\res{x}{(\ppar{\conf{C}}{\conf{C'}})}}}}$ if $\gv{\tm{x}\notin\fv(\tm{\conf{C}})}$).  They guarantee acyclicity by defining an extrinsic type system for configurations.  In particular, the type system requires that in every parallel composition $\gv{\tm{\ppar{\conf C}{\conf D}}}$, configurations
$\gv{\tm{\conf C}}$ and $\gv{\tm{\conf D}}$ must have exactly one channel in common, and that in a name restriction $\gv{\tm{\res x {\conf C}}}$, channel $\tm{x}$ cannot be used until it is shared across a parallel composition.

These restrictions are sufficient to guarantee deadlock freedom.  Unfortunately, they are \emph{not} preserved by process equivalence.  As Lindley and Morris write,
(noting that their name restrictions bind \emph{channels} rather than endpoint pairs, and their $(\nu x y)$ abbreviates $(\nu x)(\nu y)$):

\begin{quotation}
  \usingnamespace{hgv}
  Alas, our notion of typing is not preserved by configuration equivalence. For example, assume that $\Gamma \vdash (\nu x y)(C_1 \parallel (C_2 \parallel C_3))$, where $x \in \fv(C_1), y \in \fv(C_2), \text{ and } x,y \in \fv(C_3)$.  We have that $C_1 \parallel (C_2 \parallel C_3) \equiv (C_1 \parallel C_2) \parallel C_3$, but $\Gamma \nvdash (\nu x y)((C_1 \parallel C_2) \parallel C_3)$, as both $x$ and $y$ must be shared between the processes $C_1 \parallel C_2$ and $C_3$.
\end{quotation}

As a result, standard notions of progress and preservation are not enough to guarantee deadlock freedom, as reduction sequences could include equivalence steps from well-typed to non-well-typed terms. Instead, they must prove a stronger result:
\newtheorem*{LM}{Theorem 3}
\begin{LM}[Lindley and Morris~\cite{lindleym15:semantics}]
If $\gv{\cseqold{\ty{\Gamma}}{\conf{C}}}$, $\gv{\tm{\conf{C}}\equiv\tm{\conf{C'}}}$, and $\gv{\tm{\conf{C'}}\longrightarrow\tm{\conf{D'}}}$, then there exists $\gv{\tm{\conf{D}}}$ such that $\gv{\tm{\conf{D}}\equiv\tm{\conf{D'}}}$ and $\gv{\cseqold{\ty{\Gamma}}{\conf{D}}}$.
\end{LM}
This is not a one-time cost: languages based on GV must either also give up on type preservation for structural congruence~\cite{fowlerlmd19:stwt} or admit deadlocks~\cite{igarashittvw19,thiemannv20:ldst}.

Note that CP only avoids the same issue through its combined $\tm{(\nu x)(P \parallel Q)}$ term; attempts to split the term into a separate name restriction and parallel composition would also lose typability of equivalence.

 } 
{\section{Hypersequent GV}\label{sec:hgv}
\usingnamespace{hgv}

\begin{figure}[t]
\small
  \usingnamespace{hgv}

  \headersig{Typing rules for terms}{$\tseq{\ty\Gamma}{M}{T}$}
  \begin{mathpar}
    \inferrule*[lab=TM-Var]{
    }{\tseq{\tmty{x}{T}}{x}{T}}

    \inferrule*[lab=TM-Const]{
    }{\tseq{\emptyenv}{K}{T}}

    \inferrule*[lab=TM-Lam]{
      \tseq{\ty{\Gamma},\tmty{x}{T}}{M}{U}
    }{\tseq{\ty{\Gamma}}{\lambda x.M}{\tylolli{T}{U}}}

    \inferrule*[lab=TM-App]{
      \tseq{\ty{\Gamma}}{M}{\tylolli{T}{U}}
      \\
      \tseq{\ty{\Delta}}{N}{T}
    }{\tseq{\ty{\Gamma},\ty{\Delta}}{M\;N}{U}}

    \inferrule*[lab=TM-Unit]{
    }{\tseq{\emptyenv}{\unit}{\tyunit}}

    \inferrule*[lab=TM-LetUnit]{
      \tseq{\ty{\Gamma}}{M}{\tyunit}
      \\
      \tseq{\ty{\Delta}}{N}{T}
    }{\tseq{\ty{\Gamma},\ty{\Delta}}{\letunit{M}{N}}{T}}

    \inferrule*[lab=TM-Pair]{
      \tseq{\ty{\Gamma}}{M}{T}
      \\
      \tseq{\ty{\Delta}}{N}{U}
    }{\tseq{\ty{\Gamma},\ty{\Delta}}{\pair{M}{N}}{\typrod{T}{U}}}

    \inferrule*[lab=TM-LetPair]{
      \tseq{\ty{\Gamma}}{M}{\typrod{T}{T'}}
      \\
      \tseq{\ty{\Delta},\tmty{x}{T},\tmty{y}{T'}}{N}{U}
    }{\tseq{\ty{\Gamma},\ty{\Delta}}{\letpair{x}{y}{M}{N}}{U}}

    \inferrule*[lab=TM-Absurd]{
      \tseq{\ty{\Gamma}}{M}{\tyvoid}
    }{\tseq{\ty{\Gamma}}{\absurd{M}}{T}}

    \inferrule*[lab=TM-Inl]{
      \tseq{\ty{\Gamma}}{M}{T}
    }{\tseq{\ty{\Gamma}}{\inl{M}}{\tysum{T}{U}}}

    \inferrule*[lab=TM-Inr]{
      \tseq{\ty{\Gamma}}{M}{U}
    }{\tseq{\ty{\Gamma}}{\inr{M}}{\tysum{T}{U}}}

    \inferrule*[lab=TM-CaseSum]{
      \tseq{\ty{\Gamma}}{L}{\tysum{T}{T'}}
      \\
      \tseq{\ty{\Delta},\tmty{x}{T}}{M}{U}
      \\
      \tseq{\ty{\Delta},\tmty{y}{T'}}{N}{U}
    }{\tseq{\ty{\Gamma},\ty{\Delta}}{\casesum{L}{x}{M}{y}{N}}{U}}
  \end{mathpar}

  \headersig{Type schemas for communication primitives}{$\tmty{K}{T}$}
  \begin{mathpar}
    \begin{array}[t]{@{}l@{}}
    \tmty{\link}{\tylolli{\typrod{S}{\co{S}}}{\tyends}} \\
    \tmty{\fork}{\tylolli{(\tylolli{S}{\tyends})}{\co{S}}} \\
    \end{array}

    \begin{array}[t]{@{}l@{}}
    \tmty{\send}{\tylolli{\typrod{T}{\tysend{T}{S}}}{S}} \\
    \tmty{\recv}{\tylolli{\tyrecv{T}{S}}{\typrod{T}{S}}} \\
    \end{array}

    \begin{array}[t]{@{}l@{}}
    \tmty{\wait}{\tylolli{\tyendr}{\tyunit}} \\
    \end{array}
  \end{mathpar}

  \headersig{Duality}{\ty{\co{S}}}
  \begin{mathpar}
    \ty{\co{\tysend{T}{S}}} = \ty{\tyrecv{T}{\co{S}}}

    \ty{\co{\tyrecv{T}{S}}} = \ty{\tysend{T}{\co{S}}}

    \ty{\co{\tyends}} = \ty{\tyendr}

    \ty{\co{\tyendr}} = \ty{\tyends}
  \end{mathpar}

  \caption{HGV, duality and typing rules for terms.}
  \label{fig:hgv-typing-static}
\end{figure}

We present Hypersequent GV (HGV), a linear $\lambda$-calculus extended
with session types and primitives for session-typed communication. HGV
shares its syntax and static typing with GV, but uses
hyper-environments for runtime typing to simplify and generalise its
semantics.

\paragraph{Types, terms, and static typing}
Types ($\ty{T}$, $\ty{U}$) comprise a unit type ($\ty{\tyunit}$), an
empty type ($\ty{\tyvoid}$), product types ($\ty{\typrod{T}{U}}$), sum
types ($\ty{\tysum{T}{U}}$), linear function types
($\ty{\tylolli{T}{U}}$), and session types ($\ty{S}$).
\begin{mathpar}
  \ty{T}, \ty{U}
  \Coloneqq \ty{\tyunit}
  \sep      \ty{\tyvoid}
  \sep      \ty{\typrod{T}{U}}
  \sep      \ty{\tysum{T}{U}}
  \sep      \ty{\tylolli{T}{U}}
  \sep      \ty{S}

  \ty{S}
  \Coloneqq \; \ty{\tysend{T}{S}}
  \sep      \ty{\tyrecv{T}{S}}
  \sep      \ty{\tyends}
  \sep      \ty{\tyendr}
\end{mathpar}
Session types ($\ty{S}$) comprise output ($\ty{\tysend{T}{S}}$: send a value of type
$\ty{T}$, then behave like $\ty{S}$), input ($\ty{\tyrecv{T}{S}}$: receive a value of type
$\ty{T}$, then behave like $\ty{S}$), and dual end types ($\ty{\tyends}$ and
$\ty{\tyendr}$).
The dual endpoints restrict process structure to
\emph{trees}~\cite{wadler14:sessions}; conflating them loosens this
restriction to \emph{forests}~\cite{atkeylm16}.
We let $\ty{\Gamma}, \ty{\Delta}$ range over type environments.

The terms and typing rules are given
in~\Cref{fig:hgv-typing-static}. The linear $\lambda$-calculus rules
are standard; communication primitives $\tm{K}$ are given as constants.
Each communication primitive $\tm{K}$ has a type schema: $\tm{\link}$ takes a
pair of compatible endpoints and forwards all messages between them;
$\tm{\fork}$ takes a function, which is passed one endpoint (of type
$\ty{S}$) of a fresh channel yielding a new child thread, and returns
the other endpoint (of type $\ty{\co{S}}$); $\tm{\send}$ takes a pair
of a value and an endpoint, sends the value over the endpoint, and
returns an updated endpoint; $\tm{\recv}$ takes an endpoint, receives
a value over the endpoint, and returns the pair of the received value
and an updated endpoint; and $\tm{\wait}$ synchronises on a terminated
endpoint of type $\ty{\tyendr}$.
Output is dual to input, and $\ty{\tyends}$ is dual to $\ty{\tyendr}$.
Duality is involutive, \ie, $\ty{\co{\co{S}}}=\ty{S}$.

We write $\tm{\andthen{M}{N}}$ for $\tm{\letunit{M}N}$,
$\tm{\letbind{x}{M}{N}}$ for $\tm{(\lambda x.N)\;M}$,
$\tm{\lambda\unit.M}$ for $\tm{\lambda z.\andthen{z}{M}}$, and
$\tm{\lambda\pair{x}{y}.M}$ for $\tm{\lambda z.\letpair{x}{y}{z}{M}}$.
We write $\tmty{K}{T}$ for $\tseq{\emptyenv}{K}{T}$ in typing
derivations.

\begin{rem}
  We include $\tm{\link}$ because it is convenient for the
  correspondence with CP, which interprets CLL's axiom as
  forwarding. We \emph{can} encode $\tm{\link}$ in GV via a type
  directed translation akin to CLL's \emph{identity expansion}.
\end{rem}

\begin{figure}
  \small
  \usingnamespace{hgv}

  \headersig{Typing rules for configurations}{$\cseq{\ty{\hyper{G}}}{\config{C}}{R}$}
  \begin{mathpar}
    \inferrule*[lab=TC-New]{
      \cseq{\ty{\hyper{G}}\hypersep\ty{\Gamma},\tmty{x}{S}\hypersep\ty{\Delta},\tmty{y}{\co{S}}}{\conf{C}}{R}
    }{\cseq{\ty{\hyper{G}}\hypersep\ty{\Gamma},\ty{\Delta}}{\res{x}{y}{\conf{C}}}{R}}

    \inferrule*[lab=TC-Par]{
      \cseq{\ty{\hyper{G}}}{\conf{C}}{R}
      \and
      \cseq{\ty{\hyper{H}}}{\conf{D}}{R'}
    }{\cseq{\ty{\hyper{G}}\hypersep\ty{\hyper{H}}}{\ppar{\conf{C}}{\conf{D}}}{{R}\isect{R'}}}
\\
    \inferrule*[lab=TC-Main]{
      \tseq{\ty{\Gamma}}{M}{T}
    }{\cseq{\ty{\Gamma}}{\main\;M}{\tymain{T}}}

    \inferrule*[lab=TC-Child]{
      \tseq{\ty{\Gamma}}{M}{\tyends}
    }{\cseq{\ty{\Gamma}}{\child\;M}{\tychild}}

    \inferrule*[lab=TC-Link]{
    }{\cseq{\tmty{x}{S},\tmty{y}{\co{S}},\tmty{z}{\tyendr}}{\linkconfig{z}{x}{y}}{\tychild}}
  \end{mathpar}

  \begin{minipage}[t]{0.25\linewidth}
    \header{Configuration types}
    \[
      \ty{R} \Coloneqq \ty{\tychild} \sep \ty{\tymain{T}}
    \]
  \end{minipage}\hfill
  \begin{minipage}[t]{0.7\linewidth}
    \headersig{Configuration type combination}{$\ty{R}\isect\ty{R'}$}
    \[
      \ty{\tymain{T}}\isect\ty{\tychild} = \ty{\tymain{T}}
      \qquad
      \ty{\tychild}\isect\,\ty{\tymain{T}} = \ty{\tymain{T}}
      \qquad
      \ty{\tychild}\isect\,\ty{\tychild} = \ty{\tychild}
\]
  \end{minipage}
  \caption{HGV, typing rules for configurations.}
  \label{fig:hgv-typing-runtime}
\end{figure}

\paragraph{Configurations and runtime typing}
Process configurations ($\tm{\conf{C}},\tm{\conf{D}},\tm{\conf{E}}$)
comprise child threads ($\tm{\child\;M}$), the main thread
($\tm{\main\;M}$), link threads ($\tm{\linkconfig{z}{x}{y}}$), name
restrictions ($\tm{\res{x}{y}{\conf{C}}}$), and parallel compositions
($\tm{\ppar{\conf{C}}{\conf{D}}}$).
We refer to a configuration of the
form $\tm{\child M}$ or $\tm{\linkconfig{z}{x}{y}}$ as an
\emph{auxiliary thread}, and a configuration of the form $\tm{\main M}$
as a \emph{main thread}. We let $\tm{\taux}$ range over auxiliary
threads and $\tm{\tany}$ range over all threads (auxiliary or main).
\[
  \tm{\phi}
  \Coloneqq \tm{\main}
  \sep      \tm{\child}
  \qquad\qquad
  \tm{\conf{C}}, \tm{\conf{D}}, \tm{\conf{E}}
  \Coloneqq \tm{\phi\;M}
  \sep      \tm{\linkconfig{z}{x}{y}}
  \sep      \tm{\ppar{\conf{C}}{\conf{D}}}
  \sep      \tm{\res{x}{y}{\conf{C}}}
\]
The configuration language is reminiscent of $\pi$-calculus processes,
but has some non-standard features. Name restriction uses double
binders~\cite{vasconcelos12:fundamentals} in which one name is bound
to each endpoint of the channel. Link
threads~\cite{lindleym16:bananas} handle forwarding. A link thread
$\tm{\linkconfig{z}{x}{y}}$ waits for the thread connected to $\tm{z}$
to terminate before forwarding all messages between $\tm{x}$ and
$\tm{y}$.

Configuration typing departs from GV~\cite{lindleym15:semantics},
exploiting \emph{hypersequents}~\cite{avron91:hypersequents} to
recover modularity and extensibility. Inspired by
HCP~\cite{montesip18:ct,kokkemp19:tlla,kokkemp19:popl}, configurations
are typed under a \emph{hyper-environment}, an unordered collection of disjoint
type environments.
We let $\ty{\hyper{G}, \hyper{H}}$ range over hyper-environments,
writing $\emptyhyperenv$ for the empty hyper-environment,
$\ty{\hyper{G}} \hypersep \ty{\Gamma}$ for disjoint extension of
$\ty{\hyper{G}}$ with type environment $\ty{\Gamma}$, and
$\ty{\hyper{G}}\hypersep\ty{\hyper{H}}$ for disjoint concatenation of
$\ty{\hyper{G}}$ and $\ty{\hyper{H}}$.

The typing rules for configurations are given
in~\Cref{fig:hgv-typing-runtime}.
Rules \rulename{TC-New} and \rulename{TC-Par} are key to deadlock
freedom: \rulename{TC-New} joins two disjoint configurations with a
new channel, and merges their type environments;
\rulename{TC-Par} combines two disjoint configurations, and registers
their disjointness by separating their type environments in the
hyper-environment.
Rules \rulename{TC-Main}, \rulename{TC-Child}, and \rulename{TC-Link}
type main, child, and link threads, respectively; all three require a
singleton hyper-environment.
A configuration has type $\ty{\tychild}$ if it has no main thread, and
$\ty{\tymain\;T}$ if it has a main thread of type $\ty{T}$.
The configuration type combination operator ensures that a well-typed
configuration has at most one main thread.

\begin{figure}
  \small
  \usingnamespace{hgv}
\header{Values and evaluation contexts}
  \[
    \begin{array}{lrcl}
      \text{Values}
      & \tm{U}, \tm{V}, \tm{W}
      & \Coloneqq &   \tm{K}
                    \sep        \tm{\lambda x.M}
                    \sep        \tm{\unit}
                    \sep        \tm{\pair{V}{W}}
                    \sep        \tm{\inl{V}}
                    \sep        \tm{\inr{V}}
      \\
      \text{Evaluation contexts}
      & \tm{E}
      & \Coloneqq & \tm{\hole} \\
      & & \sep  & \tm{E\;M}
                  \sep    \tm{V\;E} \\
      & & \sep  & \tm{\letunit{E}{M}} \\
      & & \sep  & \tm{\pair{E}{M}}
                  \sep    \tm{\pair{V}{E}}
                  \sep    \tm{\letpair{x}{y}{E}{M}} \\
      & & \sep  & \tm{\inl{E}}
                  \sep      \tm{\inr{E}}
                  \sep      \tm{\casesum{E}{x}{M}{y}{N}} \\
      \text{Thread contexts} &
                               \tm{F}
      & \Coloneqq & \tm{\phi\;E}
    \end{array}
  \]

  \headersig{Term reduction}{$M \tred N$}
  \[
    \begin{array}{llcl}
      \rulename{E-Lam}   & \tm{(\lambda x.M) \; V}
                           & \tred & \tm{\subst{M}{V}{x}}
      \\
      \rulename{E-Unit}  & \tm{\letunit{\unit}{M}}
                           & \tred & \tm{M}
      \\
      \rulename{E-Pair}  & \tm{\letpair{x}{y}{\pair{V}{W}}{M}}
                           & \tred & \tm{M\{V / x, W / y\}}
      \\
      \rulename{E-Inl}   & \tm{\casesum{\inl{V}}{x}{M}{y}{N}}
                           & \tred & \tm{\subst{M}{V}{x}}
      \\
      \rulename{E-Inr}   & \tm{\casesum{\inr{V}}{x}{M}{y}{N}}
                           & \tred & \tm{\subst{N}{V}{y}}
      \\
      \rulename{E-Lift}  & \tm{E[M]}
                           & \tred & \tm{E[N]},
                             \text{ if }\tm{M}\tred\tm{N}
    \end{array}
  \]

  \headersig{Structural congruence}{$\conf{C} \equiv \conf{D}$}
  \[
    \begin{array}[t]{@{}c@{\quad}c@{}}
      \begin{array}[t]{@{}l@{\quad}r@{\;}c@{\;}l@{}}
        \rulename{SC-ParAssoc}
        & \tm{\ppar{\conf{C}}{(\ppar{\conf{D}}{\conf{E}})}}
        & \equiv & \tm{\ppar{(\ppar{\conf{C}}{\conf{D}})}{\conf{E}}}
        \\
        \rulename{SC-NewComm}
        & \tm{\res{x}{y}{\res{z}{w}{\conf{C}}}}
        & \equiv & \tm{\res{z}{w}{\res{x}{y}{\conf{C}}}}
        \\
        \rulename{SC-ScopeExt}
        & \tm{\res{x}{y}{(\ppar{\conf{C}}{\conf{D}}})}
        & \equiv & \tm{\ppar{\conf{C}}{\res{x}{y}{\conf{D}}}}, \text{ if }{\tm{x},\tm{y}\notin\fv(\tm{\conf{C}})}
      \end{array}
      &
      \begin{array}[t]{@{}l@{\quad}r@{\;}c@{\;}l@{}}
        \rulename{SC-ParComm}
        & \tm{\ppar{\conf{C}}{\conf{D}}}
        & \equiv & \tm{\ppar{\conf{D}}{\conf{C}}}
        \\
        \rulename{SC-NewSwap}
        & \tm{\res{x}{y}{\conf{C}}}
        & \equiv & \tm{\res{y}{x}{\conf{C}}}
        \\
        \rulename{SC-LinkComm}
        & \tm{\linkconfig{z}{x}{y}}
        & \equiv & \tm{\linkconfig{z}{y}{x}}
      \end{array}
    \end{array}
  \]

  \headersig{Configuration reduction}{$\conf{C} \cred \conf{D}$}
  \[
    \begin{array}{l}
      \begin{array}{llcl}
        \rulename{E-Reify-Fork}
        & \tm{\plug{F}{\fork\;{V}}}
        & \cred & \tm{\res{x}{x'}{(\ppar{\plug{F}{x}}{\child\;{(V\;x')}})}},
                  \text{ where }\tm{x},\tm{x'}\text{ fresh}
        \\
        \rulename{E-Reify-Link}
        & \tm{\plug{F}{\link \; \pair{x}{y}}}
        & \cred & \tm{\res{z}{z'}{(\ppar{\linkconfig{z}{x}{y}}{\plug{F}{z'}})}},
                  \text{ where }\tm{z},\tm{z'}\text{ fresh}
      \end{array}
      \\ \\
      \begin{array}{llcl}
        \rulename{E-Comm-Link}
        & \tm{\res{z}{z'}{\res{x}{x'}{(\ppar{\ppar{\linkconfig{z}{x}{y}}{\child\; z'}}{\phi\;M})}}}
        & \cred & \tm{\phi\;(M \{ y / x' \})}
        \\
        \rulename{E-Comm-Send}
        & \tm{\res{x}{y}{(\ppar{\plug{F}{\send\;{\pair{V}{x}}}}{\plug{F'}{\recv\;{y}}})}}
        & \cred & \tm{\res{x}{y}{(\ppar{\plug{F}{x}}{\plug{F'}{\pair{V}{y}}})}}
        \\
        \rulename{E-Comm-Close}
        & \tm{\res{x}{y}{(\ppar{\child\;y}{\plug{F}{\wait\;{x}}})}}
        & \cred & \tm{{\plug{F}{\unit}}}
      \end{array}
      \\ \\
      \inferrule*[lab=E-Res]{
        \tm{\conf{C}}\cred\tm{\conf{C'}}
      }{\tm{\res{x}{y}{\conf{C}}}\cred\tm{\res{x}{y}{\conf{C'}}}}
      \quad
      \inferrule*[lab=E-Par]{
        \tm{\conf{C}}\cred\tm{\conf{C}'}
      }{\tm{\conf{C}\parallel\conf{D}}\cred\tm{\conf{C}'\parallel\conf{D}}}
      \quad
      \inferrule*[lab=E-Equiv]{
        \tm{\conf{C}}\equiv\tm{\conf{C'}}
        \quad
        \tm{\conf{C'}}\cred\tm{\conf{D'}}
        \quad
        \tm{\conf{D'}}\equiv\tm{\conf{D}}
      }{\tm{\conf{C}}\cred\tm{\conf{D}}}
      \quad
      \inferrule*[lab=E-Lift-M]{
      \tm{M}\tred\tm{M'}
      }{\tm{\plug{F}{M}}\cred\tm{\plug{F}{M'}}}
    \end{array}
  \]

  \caption{HGV, operational semantics.}
  \label{fig:hgv-reduction}
\end{figure}

\paragraph{Operational semantics}
\Cref{fig:hgv-reduction} gives the operational semantics for HGV, presented as a
deterministic reduction relation on terms and a nondeterministic
reduction relation on configurations.
HGV values ($\tm{U}$, $\tm{V}$, $\tm{W}$), evaluation contexts
($\tm{E}$), and term reduction rules ($\tred$) define a standard
call-by-value, left-to-right evaluation strategy.
A closed term either reduces to a value or is blocked on a
communication action.

Thread contexts ($\tm{F}$) extend evaluation contexts to threads.
The structural congruence rules are standard apart from
\rulename{SC-LinkComm}, which ensures links are undirected, and
\rulename{SC-NewSwap}, which swaps names in double binders.

The configuration reduction relation gives a semantics for HGV's communication and
concurrency constructs.
The first two rules, \rulename{E-Reify-Fork} and
\rulename{E-Reify-Link}, create child and link threads, respectively.
The next three rules, \rulename{E-Comm-Link}, \rulename{E-Comm-Send},
and \rulename{E-Comm-Close} perform communication actions.
The final four rules enable reduction under name restriction and
parallel composition, rewriting by structural congruence, and term
reduction in threads.
Two rules handle links: \rulename{E-Reify-Link} creates a new
\emph{link thread} $\tm{\linkconfig{z}{x}{y}}$ which blocks on $\tm{z}$ of
type $\ty{\tyendr}$, one endpoint of a fresh channel. The other
endpoint, $\tm{z'}$ of type $\ty{\tyends}$, is placed in the evaluation
context of the parent thread. When $\tm{z'}$ terminates a child thread,
\rulename{E-Comm-Link} performs forwarding by substitution.

\begin{rem}
Note that \rulename{E-Comm-Link} does not fire if $z'$ is returned by a main
thread.
In closed configurations, typing ensures that such a configuration
cannot arise: intuitively, a main thread can only obtain endpoints by
$\tm\fork$ or by receiving an endpoint.

Endpoints generated to communicate with forked
threads (i.e., those passed to a child thread) will always have a
session type terminating with $\ty\tyendr$, and a child thread cannot
transmit an endpoint ending in $\ty\tyends$, since the endpoint must be
returned.
Consequently, there is no way for a main thread to obtain endpoints
with dual session types as required by the type of $\tm\link$. The case
for open configurations is accounted for by our open progress result
(see~\Cref{sec:hgv-metatheory}).
\end{rem}

\paragraph{Choice}\label{sec:hgv-choice}
HGV does not include constructs for internal and external choice (for example,
as shown in the vending machine example in Section~\ref{sec:introduction}).
Internal and external choice are instead encoded with sum types and session
delegation~\cite{kobayashi02:type-systems,dardhags17:revisited}.
Prior encodings of choice in GV~\cite{lindleym15:semantics} are
asynchronous. Instead, to encode synchronous choice we add a `dummy'
synchronisation before exchanging the value of sum type, as follows:
\begin{center}
  \small \hfill \begin{minipage}{0.35\linewidth}
    \(
    \setlength{\arraycolsep}{2pt}
    \begin{array}{lcl}
      \ty{\tyselect{S}{S'}}
      & \defeq
      & \ty{\tysend\tyunit{\tysend{(\tysum{\co{S_1}}{\co{S_2}})}\tyends}}
      \\
      \ty{\tyoffer{S}{S'}}
      & \defeq
      & \ty{\tyrecv\tyunit{\tyrecv{(\tysum{S_1}{S_2})}\tyendr}}
      \\ \\
      \ty{\tyselectemp}
      & \defeq
      & \ty{\tysend\tyunit{\tysend\tyvoid\tyends}}
      \\
      \ty{\tyofferemp}
      & \defeq
      & \ty{\tyrecv\tyunit{\tyrecv\tyvoid\tyendr}}
    \end{array}
    \)
  \end{minipage}\hfill
  \begin{minipage}{0.65\linewidth}
    \(
    \setlength{\arraycolsep}{2pt}
    \begin{array}{lcl}
      \tm{\select \ell}
      & \defeq
      & \tm{\lambda{x}.\left(
        \begin{array}{l}
          \letbind{x}{\send\;\pair{\unit}{x}}{} \\
          \fork\;(\lambda y.\send\;\pair{\ell\;y}{x})
        \end{array}
      \right)}
      \vspace{0.5em}
      \\
      \multicolumn{3}{l}{\tm{\offer{L}{x}{M}{y}{N}}}
      \\
      & \defeq
      & \!\tm{\begin{array}{l}
          \letpair{\unit}{z}{\recv\;L}{\lablet\;\pair{w}{z}={\recv\;z}} \\
          \labin\;\andthen{\wait\;z}{\casesum{w}{x}{M}{y}{N}}
        \end{array}}
      \vspace{0.5em}
      \\
      \tm{\offeremp{L}}
      & \defeq
      & \!\tm{\begin{array}{l}
          \letpair{\unit}{c}{\recv\;L}{\lablet\;\pair{z}{c}={\recv\;c}}
          \\
          \labin\;\andthen{\wait\,c}{\absurd\,z}
        \end{array}}
    \end{array}
    \)
  \end{minipage}
\end{center}

\subsection{Metatheory}\label{sec:hgv-metatheory}
\medskip
HGV enjoys type preservation, deadlock freedom, confluence, and strong
normalisation.

\paragraph{Preservation}
Hyper-environments enable type preservation under structural
congruence, which significantly simplifies the metatheory compared to
GV.
\begin{restatable}[Preservation]{thm}{hgvpres}\hfill
  \label{thm:hgv-pres}
  \begin{enumerate}
      \item If $\cseq{\ty{\hyper{G}}}{\conf{C}}{R}$ and $\tm{\conf{C}}\equiv\tm{\conf{D}}$,
          then $\cseq{\ty{\hyper{G}}}{\conf{D}}{R}$.
      \item If $\cseq{\ty{\hyper{G}}}{\conf{C}}{R}$ and $\tm{\conf{C}}\cred\tm{\conf{D}}$, then
          $\cseq{\ty{\hyper{G}}}{\conf{D}}{R}$.
  \end{enumerate}
\end{restatable}
\begin{proof}
    By induction on the derivations of $\tm{\config{C}} \equiv
    \tm{\config{D}}$ and $\tm{\config{C}} \cred \tm{\config{D}}$.
    See~\Cref{sec:appendix:hgv}.
\end{proof}

Before moving onto progress, we must introduce some technical machinery to allow
us to reason about the structure of HGV programs.

\paragraph{Abstract process structures}
Unlike in GV, in HGV we cannot rely on the fact that exactly one
channel is split over each parallel composition.
Instead, we introduce the notion of an \emph{abstract process
  structure} (APS).
Abstract process structures are a crucial ingredient in showing that HGV
configurations can be written in \emph{tree canonical form}, which helps both
with establishing progress results and also the correspondence between HGV and
GV.

We begin by establishing the intuition behind the notion of an APS, and then
describe the formal definitions.
An APS is a graph defined over a hyper-environment $\ty{\hyper{G}}$
and a set of undirected pairs of co-names (a \emph{co-name set})
$\mathcal{N}$ drawn from the names in $\ty{\hyper{G}}$.

The nodes of an APS are the type environments in
$\ty{\hyper{G}}$. Each edge is labelled by a distinct co-name pair
$\{\tm{x_1},\tm{x_2}\} \in \mathcal{N}$, such that $\tmty{x_1}{S} \in
\ty{\Gamma_1}$ and $\tmty{x_2}{\co{S}} \in \ty{\Gamma_2}$.

\begin{exa}\label{ex:topograph-1}\hfill

    \begin{minipage}{0.67\textwidth}
        {\small
  Let $\ty{\hyper{G}} = \ty{\Gamma_1} \hypersep \ty{\Gamma_2} \hypersep \ty{\Gamma_3}$,
  where
$\ty{\Gamma_1} = \tmty{x}{S_1}, \tmty{y}{S_2}$, $\ty{\Gamma_2} =
  \tmty{x'}{\co{S_1}}, \tmty{z}{T}$, and $\ty{\Gamma_3} =
  \tmty{y'}{\co{S_2}}$, and suppose $\mathcal{N} = \{ \{ \tm{x}, \tm{x'} \}, \{ \tm{y}, \tm{y'} \} \}$.
The APS for $\ty{\hyper{G}}$ and $\mathcal{N}$ is illustrated to the
  right.
  }
  \end{minipage}
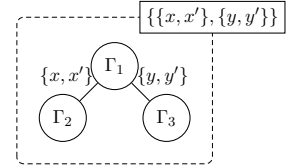
\hfill
  \begin{minipage}{0.3\textwidth}
  \begin{center}
      \scalebox{0.65}{\begin{tikzpicture}[>=stealth',auto,node distance=1.5cm,
                    semithick]

  \tikzstyle{every state}=[draw=black,text=black]

  \node[state]         (A)                    {$\Gamma_1$};
  \node[state]         (B) [below left of =A]  {$\Gamma_2$};
  \node[state]         (C) [below right of=A] {$\Gamma_3$};

  \path (A) edge              node[above, xshift=-0.45cm] {$\{x, x'\}$} (B)
    (A) edge              node[above, xshift=0.45cm] {$\{y, y'\}$} (C);

    \draw[densely dashed, rounded corners, thin]
      (-2, -2) rectangle (2, 1);
      \node[draw,rectangle, fill=white,minimum size=0.5cm] at (2,1){$\{\{x, x'\}, \{y, y'\}\}$};
\end{tikzpicture}

 }
  \end{center}
  \end{minipage}
\end{exa}

\begin{exa}\label{ex:topograph-2}\hfill

\begin{minipage}{0.67\textwidth}
    {\small
    Let $\ty{\hyper{G}} = \ty{\Gamma_1} \hypersep \ty{\Gamma_2} \hypersep
    \ty{\Gamma_3}$, where
$\ty{\Gamma_1} = \tmty{x}{S_1}, \tmty{z'}{\co{S_3}}$,
    and $\ty{\Gamma_2} = \tmty{x'}{\co{S_1}}, \tmty{y}{S_2}$,
    and
    $\ty{\Gamma_3} = \tmty{y'}{\co{S_2}}, \tmty{z}{S_3}$,
and suppose $\mathcal{N} = \{ \{ x, x' \}, \{ y, y' \}, \{ z, z' \} \}$.
The APS for $\ty{\hyper{G}}$ and $\mathcal{N}$ is illustrated to the
  right.

}
\end{minipage}
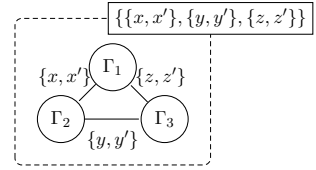
\hfill
  \begin{minipage}{0.3\textwidth}
  \begin{center}
      \scalebox{0.65}{\begin{tikzpicture}[,>=stealth',shorten >=1pt,auto,node distance=1.5cm,
                    semithick]
  \tikzstyle{every state}=[draw=black,text=black]

  \node[state]         (A)                    {$\Gamma_1$};
  \node[state]         (B) [below left of=A]  {$\Gamma_2$};
  \node[state]         (C) [below right of=A] {$\Gamma_3$};

  \path (A) edge              node[above, xshift=-0.45cm] {$\{ x, x'\}$} (B)
    (B) edge              node[below, yshift=-0.1cm] {$\{ y, y' \}$} (C)
    (C) edge              node[above, xshift=0.45cm] {$\{ z, z' \}$} (A);

    \draw[densely dashed, rounded corners, thin]
      (-2, -2) rectangle (2, 1);
      \node[draw,rectangle, fill=white,minimum size=0.5cm] at (2,1){$\{\{x,
      x'\}, \{y, y'\}, \{ z, z'\} \}$};

\end{tikzpicture}

 }
  \end{center}
\end{minipage}
\end{exa}

Let us now discuss the formal definition of an APS.
We begin by recalling the definition of an \emph{undirected edge-labelled
multigraph}: an undirected graph that allows multiple edges between vertices.

\begin{defi}[Undirected Multigraph]
  An \emph{undirected multigraph} $G$ is a 3-tuple $(\graphvar{V}, \graphvar{E}, r)$ where:

  \begin{enumerate}
    \item $\graphvar{V}$ is a set of vertices
    \item $\graphvar{E}$ is a set of edge names
    \item $r$ is a function $r : \graphvar{E} \mapsto \{ \{v, w \} : v, w \in \graphvar{V} \}$ from
      edge names to an unordered pair of vertices
  \end{enumerate}
\end{defi}

Denote the size of a set as $\size{\cdot}$.
A \emph{path} is a sequence of edges connecting two vertices. A multigraph $G = (\graphvar{V}, \graphvar{E}, r)$ is \emph{connected} if $\size{\graphvar{V}} = 1$, or if for every pair of vertices $v, w \in \graphvar{V}$ there is a path between $v$ and $w$. A multigraph is \emph{acyclic} if no path forms a cycle.
A \emph{leaf} is a vertex connected to the remainder of a graph by a
single edge.

\begin{defi}[Leaf]
  Given an undirected multigraph $(\mathcal{V}, \mathcal{E}, r)$, a vertex $v \in
  \mathcal{V}$ is a \emph{leaf} if there exists a single $e \in \mathcal{E}$
  such that $v \in r(e)$.
\end{defi}

In an undirected tree containing at least two vertices, there must be at least two leaves.

\begin{lem}\label{lem:two-leaves}
  If $G = (\mathcal{V}, \mathcal{E}, r)$ is an undirected tree where
  $\size{V} \ge 2$, then there exist at least two leaves in
  $\mathcal{V}$.
\end{lem}
\begin{proof}
  For $G$ to be an undirected tree where $\size{V} \ge 2$ and have fewer than
  two leaves, then there would need to be a cycle, contradicting
  acyclicity.
\end{proof}

With the graph preliminaries in place, we are now ready to introduce the formal
definition of an APS.

\begin{defi}[Abstract process structure]
    The \emph{abstract process structure of a hyper-environment}
    $\ty{\hyper{H}}$ \emph{with respect to a co-name set} $\mathcal{N} = \{ \{
    \tm{x_1}, \tm{y_1} \}, \ldots, \{ \tm{x_n}, \tm{y_n} \} \}$ is an undirected multigraph $(\graphvar{V}, \graphvar{E}, r)$ defined as follows:

  \begin{enumerate}
      \item $\graphvar{V} = \envs{\ty{\hyper{H}}}$
    \item $\graphvar{E} = \mathcal{N}$
    \item $r = (\{ \tm{x}, \tm{y} \} \mapsto \{ \ty{\Gamma_1}, \ty{\Gamma_2} \}
        )$ for each $\{ \tm{x}, \tm{y} \} \in \mathcal{N}$ such that
        $\ty{\Gamma_1} \in \envs{\ty{\hyper{H}}}, \ty{\Gamma_2} \in
        \envs{\ty{\hyper{H}}}, \tm{x} \in \fv(\ty{\Gamma_1}), \tm{y} \in
        \fv(\ty{\Gamma_2})$
  \end{enumerate}
\end{defi}

\begin{exa}
The formal definition of the APS described in Example~\ref{ex:topograph-1} is defined as:
      \begin{itemize}
          \item $\graphvar{V} = \{ \ty{\Gamma_1}, \ty{\Gamma_2}, \ty{\Gamma_3} \}$
          \item $\graphvar{E} = \{ \{ \tm{x}, \tm{x'} \}, \{ \tm{y}, \tm{y'} \} \}$
          \item $r(\{ \tm{x}, \tm{x'} \}) \mapsto \{ \ty{\Gamma_1},
              \ty{\Gamma_2} \})\\
                      r(\{ \tm{y}, \tm{y'} \}) \mapsto \{ \ty{\Gamma_1},
                      \ty{\Gamma_3} \})$
      \end{itemize}
\end{exa}

\noindent
Whereas Example~\ref{ex:topograph-1} is a tree,
Example~\ref{ex:topograph-2} contains a cycle.
Only configurations typeable under a hyper-environment with a \emph{tree
structure} can be written in tree canonical form.

\begin{defi}[Tree structure]
    A hyper-environment $\ty{\hyper{H}}$ with co-name set $\mathcal{N}$
    has a \emph{tree structure}, written $\treeprop{\ty{\hyper{H}}}{\nameset}$,
    if its APS is connected and acyclic.
\end{defi}

An HGV program $\tm{\main\;M}$ has a single type environment, so is
tree-structured; the same goes for child and link threads.
A key feature of HGV is a subformula principle, which states that all
hyper-environments arising in the derivation of an HGV program are
tree-structured.
It follows that a configuration resulting from the reduction of an HGV program
is also tree structured.
Read bottom-up, \LabTirName{TC-New} and \LabTirName{TC-Par} preserve
tree structure, which
is
illustrated by the following two pictures.
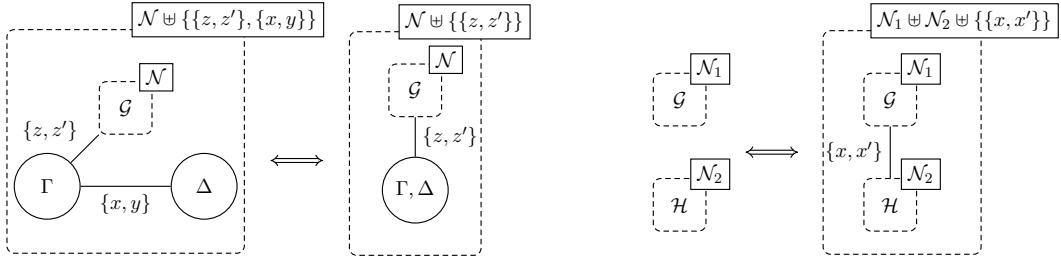
\begin{center}
  \hspace{1em}
  \begin{minipage}{0.4\textwidth}
    \scalebox{0.7}{\begin{tikzpicture}[ >=stealth'
  , auto
  , node distance=0.75cm
  , semithick
  ]
  \tikzstyle{every state}=[draw=black,text=black]

  \begin{scope}
    \node
    [ draw
    , densely dashed
    , rounded corners
    , minimum width=1cm
    , minimum height=1cm
    ] (C) {$\hyper{G}$};
    \node[draw] (NBOXC)
    [ above right=0.1cm of C
    , fill=white
    , anchor=center
    ] {$\mathcal{N}$};

    \node
    [ state
    , minimum width=1.25cm
    ] (D)
    [ align=center
    , below left=of C
    ] {$\Gamma$};

    \node
    [ state
    , minimum width=1.25cm
    ] (E)
    [ align=center
    , below right=of C
    ] {$\Delta$};

    \node
    [ draw
    , densely dashed
    , rounded corners
    , minimum width=3.5cm
    , minimum height=4.25cm
    , fit=(C) (NBOXC) (D) (E)
    ] (BOXB) {};

    \node[draw] (NBOXB)
    [ above right=0.25cm of BOXB
    , fill=white
    , anchor=center
    , xshift=-0.5cm
    ] {$\mathcal{N} \uplus \{ \{ z, z' \}, \{ x, y \} \}$};

    \path (C) edge node[above left] {$\{z, z'\}$} (D);
    \path (D) edge node[below] {$\{x, y\}$} (E);
  \end{scope}

  \begin{scope}[xshift=5.5cm]

    \node
    [ draw
    , densely dashed
    , rounded corners
    , minimum width=1cm
    , minimum height=1cm
    , yshift=0.33333cm
    ] (A) {$\hyper{G}$};

    \node [draw] (NBOXA)
    [ above right=0.1cm of A
    , fill=white
    , anchor=center
    ] {$\mathcal{N}$};

    \node
    [ state
    , minimum width=1.25cm
    ] (B)
    [ align=left
    , below=of A
    ] {$\Gamma, \Delta$};

    \path (A) edge node[right] {$\{z, z'\}$} (B);

    \node
    [ draw
    , densely dashed
    , rounded corners
    , minimum height=4.25cm
    , minimum width=2.5cm
    , fit=(A) (B)
    ] (BOX) {};

    \node[draw] (NBOX)
    [ above right=0.25cm of BOX
    , fill=white
    , anchor=center
    , xshift=-0.5cm
    ] {$\mathcal{N} \uplus \{ \{ z, z' \} \}$};
  \end{scope}

  \draw
  [ implies-implies
  , double equal sign distance
  ] (2.75,-1) -- (3.75,-1) {};
\end{tikzpicture}

 }
  \end{minipage}
  \hfill
  \begin{minipage}{0.4\textwidth}
    \scalebox{0.7}{\begin{tikzpicture}[>=stealth',auto,node distance=2cm,
                    semithick]

  \tikzstyle{every state}=[draw=black,text=black]

  \begin{scope}[x=0, y=0]

  \node[draw,densely dashed,rounded corners, minimum width=1cm, minimum height=1cm]
    (A) [align=center]                   {$\hyper{G}$};
  \node[draw,densely dashed,rounded corners, minimum width=1cm, minimum height=1cm]
    (B) [align=center,below of=A]       {$\hyper{H}$};

  \node[draw] (NA) [above right=0.1cm of A, fill=white, anchor=center] {$\hyper{N}_1$};
  \node[draw] (NB) [above right=0.1cm of B, fill=white, anchor=center] {$\hyper{N}_2$};
  \node[fit=(A) (B) (NA) (NB), minimum height=4.25cm] (BOUND) {};

  \end{scope}

  \begin{scope}[xshift=4cm]

  \node[draw,densely dashed,rounded corners, minimum width=1cm, minimum height=1cm]
    (C) [align=center]                   {$\hyper{G}$};
  \node[draw,densely dashed,rounded corners, minimum width=1cm, minimum height=1cm]
    (D) [align=center,below of=C]        {$\hyper{H}$};

  \node[draw] (NC) [above right=0.1cm of C, fill=white, anchor=center] {$\hyper{N}_1$};
  \node[draw] (ND) [above right=0.1cm of D, fill=white, anchor=center] {$\hyper{N}_2$};

  \node[draw,densely dashed,rounded corners, minimum width=3cm, minimum height=4.25cm, fit=(C) (D) (NC) (ND)] (BOX) {};

  \node[draw] (NBOX) [above right=0.25cm of BOX, fill=white, anchor=center, xshift=-0.5cm] {$\hyper{N}_1 \uplus \hyper{N}_2 \uplus \{ \{ x, x' \} \}$};
  \path (C) edge node[left] {$\{x, x'\}$} (D);

  \node[fit=(C) (D) (NC) (ND) (NBOX) (BOX), minimum height=4.25cm] (RBOUND) {};

  \end{scope}

  \draw[implies-implies,double equal sign distance]
  (1.25,-1) -- (2.25,-1) {};

\end{tikzpicture}

 }
  \end{minipage}
  \hspace{1em}
\end{center}

The following lemma states this intuition formally. By analogy to Kleene
equality, we write $\mathcal{P} \kleeneiff \mathcal{Q}$, to mean that either
$\mathcal{P}$ or $\mathcal{Q}$ is undefined, or $\mathcal{P} \iff \mathcal{Q}$.

\begin{lem}[Tree structure]\label{lem:tree-structure-iff}\hfill

\begin{itemize}
\item
  $
  \treeprop
  { (\ty{\hyper{H}} \hypersep \ty{\Gamma_1}, \tmty{x_1}{S}
      \hypersep \ty{\Gamma_2}, \tmty{x_2}{\co{S}}) }
      { \mathcal{N} \uplus \{ \{ \tm{x_1}, \tm{x_2} \} \} }
  ~\kleeneiff~
  \treeprop
  {(\ty{\hyper{H}} \hypersep \ty{\Gamma_1}, \ty{\Gamma_2}))}
    {\mathcal{N}}
  $
\item
    $
    \treeprop
    {(\ty{\hyper{H}_1} \hypersep \ty{\Gamma_1}, \tmty{x_1}{S})}
        {\mathcal{N}_1}
    \wedge
    \treeprop
    {(\ty{\hyper{H}_2} \hypersep \ty{\Gamma_2}, \tmty{x_2}{\co{S}})}
        {\mathcal{N}_2}
    ~\kleeneiff~
    \treeprop
    {(\ty{\hyper{H}_1} \hypersep \ty{\Gamma_1}, \tmty{x_1}{S}
        \hypersep \ty{\hyper{H}_2}
        \hypersep \ty{\Gamma_2}, \tmty{x_2}{\co{S}})}
        {\mathcal{N}_1 \uplus \mathcal{N}_2 \uplus \{ \{ \tm{x_1}, \tm{x_2} \} \}}
$
\end{itemize}
\end{lem}
\begin{proof}
    By the definition of $\kleeneiff$, we need only consider the cases where
    both sides of the bi-implication are defined.
Both results follow from the observation that adding an edge between two
    trees results in a tree, and removing an edge from a tree partitions the
    tree into two subtrees.
\end{proof}

\paragraph{Tree canonical form}
We now define a canonical form for configurations that captures the
tree structure of an APS.
Tree canonical form enables a succinct statement of \emph{open
progress}~(Lemma~\ref{lem:hgv:open-progress}) and a means for embedding HGV
in GV~(Proposition~\ref{lem:hgv-to-gv-tcf}).
\begin{defi}[Tree canonical form]
   \label{def:hgv-tree-canonical-forms}
   A configuration $\config{C}$ is in \emph{tree canonical form} if it can be written:
   $\tm{\res{x_1}{y_1}{(\taux_1 \parallel \cdots \parallel \res{x_n}{y_n}{(\taux_n \parallel \phi{N})} \cdots)}}$
   where $\tm{x_i} \in \fv(\tm{\taux_i})$ for $1 \le i \le n$.
 \end{defi}
Every well-typed HGV configuration typeable under a single type
environment can be written in tree canonical form.
\begin{restatable}[Well-typed configurations in tree canonical forms]{thm}{tcfs}\label{lem:hgv:tree-canonical-forms}
  If $\cseq{\ty{\Gamma}}{\config{C}}{R}$, then there exists some $\tm{\config{D}}$
  such that $\tm{\config{C}} \equiv \tm{\config{D}}$ and $\tm{\config{D}}$ is in tree canonical
  form.
\end{restatable}
\begin{proof}
    By induction on the number of $\nu$-binders in $\tm{\config{C}}$. In the case that
    $n=0$, it must be the case that $\cseq{\ty{\Gamma}}{\phi M}{R}$ for some thread
    $\tm{M}$, since parallel composition is only typeable under a hyper-environment
    containing two or more type environments. Therefore, $\tm{\config{C}}$ is in tree
  canonical form by definition.

  In the case that $n \ge 1$, by~\Cref{thm:hgv-pres}, we can
  rewrite the configuration as:
\[
    \tm{(\nu x_1 y_1) \cdots (\nu x_n y_n) (\child M_1 \parallel \cdots \parallel
    \child M_n \parallel \phi N)}
  \]Fix $\mathcal{N} = \{ \{ \tm{x_i}, \tm{y_i} \} \mid 1 \le i \le n \}$.
  By definition, $\ty{\Gamma}$ has a tree structure with respect to an empty co-name set.
By repeated applications of \textsc{TC-New}, there exists some $\ty{\hyper{G}}$ such
that
$\cseq
    {\ty{\hyper{G}}}
   {\child M_1 \parallel \cdots \parallel \child M_n \parallel \phi N}
   {T}$;
   by Lemma~\ref{lem:tree-structure-iff} (clause 1, right-to-left),
   $\ty{\hyper{G}}$ has a tree
structure.

Construct the APS for $\ty{\hyper{G}}$ using names $\mathcal{N}$; by
Lemma~\ref{lem:two-leaves}, there exist $\ty{\Gamma_1}, \ty{\Gamma_2} \in
\envs{\ty{\hyper{H}}}$ such that $\ty{\Gamma_1}$ and $\ty{\Gamma_2}$ are leaves of the tree and
therefore by the definition of the APS contain precisely one $\nu$-bound name.
By \textsc{TC-Par}, there must exist two threads $\tm{\config{C}_1},
\tm{\config{C}_2}$ such that
$\cseq{\ty{\Gamma_1}}{\config{C}_1}{R_1}$ and
$\cseq{\ty{\Gamma_2}}{\config{C}_2}{R_2}$. By runtime type combination,
at least one of $\ty{R_1}, \ty{R_2}$ must be $\ty{\tychild}$; without loss of
generality assume this is $\ty{R_1}$. Suppose (again without loss of generality)
that the $\nu$-bound name contained in $\ty{\Gamma_1}$ is $\tm{x_1}$ and
$\tm{L_1} = \tm{M_1}$.

Let $\tm{\config{D}} = \tm{(\nu x_2 y_2) \cdots
  (\nu x_n y_n)(\child M_2 \parallel \cdots \parallel \child M_n \parallel \phi
  N)}$.
By Theorem~\ref{thm:hgv-pres} and the fact that $\tm{x_1}$ is the only
  $\nu$-bound variable in $\tm{M_1}$, we have that
  $\tm{\config{C}} \equiv \tm{(\nu x_1 y_1)(\child M_1 \parallel \config{D})}$.
By the induction hypothesis, there exists some $\tm{\config{D}'}$ such that $\tm{\config{D}} \equiv
  \tm{\config{D}'}$ and $\tm{\config{D}'}$ is in canonical form. By construction we have
  that $\tm{\config{C}} \equiv \tm{(\nu x_1 y_1)(\child M_1 \parallel
  \config{D}')}$, which
is in tree canonical form as required.
\end{proof}

As hyper-environments capture parallelism, a configuration
$\tm{\config{C}}$ typeable under hyper-environment $\ty{\Gamma_1} \hypersep
\cdots \hypersep \ty{\Gamma_n}$ is equivalent to $n$ independent parallel
processes.
\begin{prop}[Independence]\label{prop:hgv:independence}
  If $\cseq{\ty{\Gamma_1} \hypersep \cdots \hypersep
    \ty{\Gamma_n}}{\config{C}}{R}$, then there exist $\ty{R_1},
  \ldots, \ty{R_n}$ and $\tm{\config{D}_1}, \ldots, \tm{\config{D}_n}$
  such that $\ty{R} = \ty{R_1} \isect \cdots \isect \ty{R_n}$ and
  $\tm{\config{C}} \equiv \tm{\config{D}_1 \parallel \cdots \parallel
    \config{D}_n}$ and $\cseq{\ty{\Gamma_i}}{\config{D}_i}{R_i}$ for
  each $i$.
\end{prop}
\begin{proof}
    By induction on the derivation of $\cseq{\ty{\Gamma_1} \hypersep \cdots \hypersep
    \ty{\Gamma_n}}{\config{C}}{R}$.
    The cases for \textsc{TC-Main}, \textsc{TC-Child}, and \textsc{TC-Link}
    follow immediately. The cases for \textsc{TC-New} and \textsc{TC-Par}
    follow from the IH and structural congruence rules.
\end{proof}

It follows from Theorem~\ref{lem:hgv:tree-canonical-forms}
and Proposition~\ref{prop:hgv:independence} that any well-typed HGV configuration can be
written as a forest of independent configurations in tree canonical form.

\paragraph{Progress and Deadlock Freedom}
With tree canonical forms defined, we can now state a progress result.
A thread is \emph{blocked} on an endpoint $x$ if it is ready to
perform a communication action on $x$.
\begin{defi}[Blocked thread]
  We say that thread $\tm{\tany}$ is \emph{blocked on variable} $\tm{z}$,
  written $\blocked{\tany}{z}$, if either: $\tm{\tany} = \tm{\child\;z}$;
  $\tm{\tany} = \tm{\linkconfig{z}{x}{y}}$, for some $\tm{x}$, $\tm{y}$;
  or $\tm{\tany} = \tm{\plug{F}{N}}$ for some $\tm{F}$,
  where $\tm{N}$ is $\tm{\send\:\pair{V}{z}}$, $\tm{\recv\:z}$, or $\tm{\wait\:z}$.
\end{defi}
We let $\ty{\Psi}$ range over type environments containing only
session-typed variables, \ie, $\ty{\Psi} \bnfdef \cdot \sep
\ty{\Psi}, \tmty{x}{S}$, which lets us reason about configurations
that are closed except for runtime names. Using
Lemma~\ref{lem:hgv:open-progress} we obtain \emph{open progress} for
configurations with free runtime names.
\begin{lem}[Open Progress]\label{lem:hgv:open-progress}
  Suppose $\cseq{\ty{\Psi}}{\config{C}}{T}$ where
\[
      \tm{\config{C}} = \tmcolor (\nu x_1 y_1)(\taux_1 \parallel \cdots
  \parallel (\nu x_n y_n)(\taux_n \parallel \phi N) \cdots )
  \] is in tree canonical form.
Either $\tm{\config{C}} \cred \tm{\config{D}}$ for some $\tm{\config{D}}$, or:

  \begin{enumerate}
  \item For each $\tm{\taux_i}$ $(1 \le i \le n)$,
$\blocked{\taux_i}{z}$ for some $\tm{z} \in \{ \tm{x_i} \} \cup \{ \tm{y_j}
        \mid 1 \le j < i \} \cup \fv(\ty{\Psi})$
      \item Either $\tm{N}$ is a value or
          $\blocked{\phi N}{z}$ for some $\tm{z} \in \{ \tm{y_i} \mid 1 \le i \le n \} \cup \fv(\ty{\Psi})$
    \end{enumerate}
\end{lem}
\begin{proof}
Open progress follows as a direct corollary of a slightly more verbose property
which holds on HGV processes, proved by induction on the derivation of an inductive definition of tree canonical forms.
See~\Cref{sec:appendix:hgv} for details.
\end{proof}
Closed configurations enjoy a stronger result: if a closed configuration cannot reduce, then each auxiliary thread must either be a value, or be blocked on its neighbouring endpoint.

\begin{lem}[Closed Progress]\label{lem:hgv:closed-progress}
    Suppose $\cseq{\ty{\Psi}}{\tm{\config{C}}}{\ty{R}}$
    where
    \[
    \tm{\config{C}} = \tm{(\nu x_1 y_1)(\taux_1 \parallel \cdots \parallel (\nu x_n
  y_n)(\taux_n
  \parallel \phi N) \cdots )}
  \] is in tree canonical form.
Either $\tm{\config{C}} \cred \tm{\config{D}}$ for some $\tm{\config{D}}$, or:

  \begin{enumerate}
      \item For each $\tm{\taux_j}$ for $1 \le j \le n$, $\blocked{\tm{\taux_j}}{\tm{x_j}}$
      \item $\tm{N}$ is a value
    \end{enumerate}
\end{lem}
\begin{proof}
  Since the environment is closed, by Lemma~\ref{lem:hgv:open-progress}, for
  each $\tm{\taux_j}$ it must be that
  $\blocked{\tm{\taux_j}}{\tm{z}}$ for some
  $\tm{z} \in \{ \tm{y_i} \mid i \in 1.. {j-1} \} \cup \{ \tm{x_j} \}$.

  Note that if two names $\tm{x}, \tm{y}$ are co-names, and one thread is blocked on
  $\tm{x}$, and another is blocked on $\tm{y}$, then due to typing the names must be dual
  and reduction can occur.

  Consider $\tm{\taux_1}$. Since the environment is closed, $\tm{\taux_1}$ must
  be blocked on $\tm{x_1}$.  Next, consider $\tm{\taux_2}$; the thread cannot be
  blocked on $\tm{y_1}$ as reduction would occur.
  By the definition of tree canonical forms, $\tm{\taux_2}$ must contain
  $\tm{x_2}$ and by the typing rules cannot contain $\tm{y_2}$, so the thread
  must be blocked on $\tm{x_2}$. The argument extends to the remainder of the
  configuration.
\end{proof}

Finally, for \emph{ground configurations}, where the main thread does
not return a runtime name or capture a runtime name in a closure, we
obtain a yet tighter result, \emph{global progress}, which implies
deadlock freedom~\cite{CarboneDM14}.
\begin{defi}[Ground configuration]
  A configuration $\tm{\config{C}}$ is a \emph{ground configuration} if
  $\cseq{\cdot}{\config{C}}{T}$, $\tm{\config{C}}$ is in canonical form,
  and $T$ does not contain session types or function types.
\end{defi}

Our main progress result states that a ground configuration can
reduce, or is a value.

\begin{thm}[Global progress]\label{thm:hgv-global-progress}
  Suppose $\tm{\config{C}}$ is a ground configuration. Either there exists some
  $\tm{\config{D}}$ such that $\tm{\config{C}} \cred \tm{\config{D}}$, or $\tm{\config{C}} =
  \tm{\main V}$ for some value $\tm{V}$.
\end{thm}
\begin{proof}
    By Lemma~\ref{lem:hgv:closed-progress}, either $\tm{\config{C}}$ can reduce, or
    $\tm{\config{C}}$ can be written:
    \[
  \tm{(\nu x_1 y_1) (\child \config{A}_1 \parallel \cdots \parallel (\nu x_n y_n)
  (\child \config{A}_n \parallel \main V) \cdots )}
  \]
  where
  $\blocked{\tm{\taux_i}}{\tm{x_i}}$  for each $\{ \tm{x_i} \mid i \in 1..n \}$.

  Since $\tm{\config{C}}$ is ground, $\fv(\tm{V}) = \emptyset$. By definition, tree
  canonical form ensures that no cycles are present amongst threads, so no
  auxiliary thread can be blocked.
It follows that if $\tm{\config{C}} \not\cred$, then there cannot be any auxiliary
  threads and thus $\tm{\config{C}} = \tm{\main V}$ for some value $\tm{V}$.
\end{proof}

\paragraph{Determinism and Strong Normalisation}
HGV enjoys a strong form of determinism known as the diamond property,
and due to linearity it enjoys strong normalisation. Unlike with
preservation and progress, the addition of hypersequents does not
substantially change the arguments from~\cite{lindleym15:semantics}.

\begin{thm}[Diamond property]
  If $\cseq{\ty{\hyper{G}}}{\conf{C}}{T}$, $\tm{\conf{C}}\cred\tm{\conf{D}}$, and $\tm{\conf{C}}\cred\tm{\conf{D}'}$, then $\tm{\conf{D}}\equiv\tm{\conf{D}'}$.
\end{thm}
\begin{proof}
  Similar to that of GV~\cite{lindleym15:semantics,fowler19:thesis}:
  $\tred$ is deterministic, and due to linearity, any overlapping
  reductions are separate and may be performed in either order.
\end{proof}

\begin{thm}[Termination]
  If $\cseq{\ty{\hyper{G}}}{\conf{C}}{T}$, there are no infinite
  sequences $\tm{\conf{C}}\cred\cred\cdots$.
\end{thm}
\begin{proof}
  As with GV~\cite{lindleym15:semantics,fowler19:thesis}, due to
  linearity, HGV has an elementary strong normalisation proof. Let the
  size of a configuration be the sum of the sizes of all abstract
  syntax trees of all terms contained in threads. The size of a
  configuration is invariant under $\equiv$ and strictly decreases
  under $\cred$, so no infinite reduction sequences can exist.
\end{proof}

 } 
{\section{Relation between HGV and GV}
\label{sec:relation-to-gv}
In this section, we show that well-typed GV configurations are
well-typed HGV configurations, and well-typed HGV configurations with
tree structure are well-typed GV configurations.

\paragraph{GV}
\begingroup
\usingnamespace{gv}
\begin{figure}
  \small
\headersig{Typing rules for configurations}{$\gvseq{\ty{\Gamma}}{\conf{C}}{T}$}
  \begin{mathpar}
    \inferrule*[lab=TG-New]{
      \gvseq{\ty{\Gamma},\tmty{\lock{x}{y}}{\lockty{S}}}{\conf{C}}{R}
    }{\gvseq{\ty{\Gamma}}{(\nu x y) \conf{C}}{R}}
~
    \inferrule*[lab=TG-Connect$_1$]{
      \gvseq{\ty{\Gamma_1},\tmty{x}{S}}{\conf{C}}{R}
      \\\\
      \gvseq{\ty{\Gamma_2},\tmty{y}{\co{S}}}{\conf{D}}{R'}
    }{\gvseq
      {\ty{\Gamma_1},\ty{\Gamma_2},\tmty{\lock{x}{y}}{\lockty{S}}}
      {\conf{C}\parallel\conf{D}}
      {\ty{R}\isect\ty{R'}}}
~
    \inferrule*[lab=TG-Connect$_2$]{
      \gvseq{\ty{\Gamma_1},\tmty{y}{\co{S}}}{\conf{C}}{R}
      \\\\
      \gvseq{\ty{\Gamma_2},\tmty{x}{S}}{\conf{D}}{R'}
    }{\gvseq
      {\ty{\Gamma_1},\ty{\Gamma_2},\tmty{\lock{x}{y}}{\lockty{S}}}
      {\conf{C}\parallel\conf{D}}
      {\ty{R}\isect\ty{R'}}}

    \inferrule*[lab=TG-Child]{
      \gvseq{\ty{\Gamma}}{M}{\tyends}
    }{\gvseq{\ty{\Gamma}}{\child M}{\tychild}}

    \inferrule*[lab=TG-Main]{
      \gvseq{\ty{\Gamma}}{M}{T}
    }{\gvseq{\ty{\Gamma}}{\main M}{\tymain{T}}}

    \inferrule*[lab=TG-Link]{
    }{\gvseq
      {\tmty{x}{S},\tmty{y}{\co{S}},\tmty{z}{\tyendr}}
      {\linkconfig{z}{x}{y}}
      {\tychild}}
  \end{mathpar}
  \caption{GV, typing rules for configurations.}
  \label{fig:gv-typing-runtime}
\end{figure}

HGV and GV share a common term language and reduction semantics, so
only differ in their runtime typing rules. \Cref{fig:gv-typing-runtime} gives
the runtime typing rules for GV. We adapt the rules to use a double-binder
formulation to concentrate on the essence of the relationship with HGV, but it
is trivial to translate GV with single binders into GV with double binders.

GV uses a pseudo-type $\ty{\lockty{S}}$ to type channels.
Unlike endpoints, channels cannot appear in terms.
Read bottom-up, rule \LabTirName{TG-New} types a name restriction
$\tm{\resd{x}{y}{\config{C}}}$, adding
$\tmty{\lock{x}{y}}{\lockty{S}}$ to the type environment, which along
with \LabTirName{TG-Connect$_1$} and \LabTirName{TG-Connect$_2$}
ensures that a session channel of type $S$ will be split into
endpoints $\tm{x}$ and $\tm{y}$ over a parallel composition. In turn, this
enforces a tree process structure. The remaining typing rules are as
in HGV.  \endgroup

\begingroup
\paragraph{A simple embedding of GV into HGV}
The simplest embedding of GV in HGV relies on the observation
from~\Cref{sec:problem-with-gv} that each parallel composition splits
a single channel. Let $\gv{\tm{\config{C} \parannot{x}{y} \config{D}}}$
denote two configurations $\tm{\config{C}}$ and $\tm{\config{D}}$ connected
by a channel with endpoints $x, y$. We can write an arbitrary closed GV
configuration in the form:
\[\gv{\tmcolor \config{C}_1 \parannot{x_1}{y_1}\cdots\parannot{x_{n - 2}}{y_{n - 2}}\config{C}_{n - 1} \parannot{x_{n - 1}}{y_{n - 1}} \config{C}_n
}\]
where each $\tm{\config{C}}$ does not contain a further parallel
composition, and any main thread is in $\tm{\config{C}_n}$. We can then
subsequently embed the configuration in HGV as:
\[\gv{\tmcolor (\nu x_1 y_1)(\config{C}_1 \parallel \cdots \parallel (\nu x_{n - 2} y_{n - 2})(\config{C}_{n - 2} \parallel (\nu x_{n - 1} y_{n - 1})(\config{C}_{n - 1} \parallel \config{C}_n)) \cdots)
}\]
which is well-typed by construction. As a corollary, every well-typed,
closed GV configuration is equivalent to a well-typed, closed HGV
configuration.

\paragraph{A structure-preserving embedding of GV into HGV}
Though the simple embedding of GV into HGV is sound, it does not
respect the \emph{intention} of GV. In fact, we can provide
a stronger result: every well-typed open GV configuration is exactly a
well-typed HGV configuration.

\begin{defi}[Flattening]
  Flattening, written $\flatten{}$, converts GV type environments and HGV hyper-environments into HGV environments.
  \[
    \begin{array}{lcl}
      \gv{\flatten{\emptyenv}}
      & = & \gv{\emptyenv}
      \\
      \gv{\flatten{(\ty{\Gamma},\tmty{\lock{x}{x'}}{\lockty{S}})}}
      & = & \gv{\flatten{\ty{\Gamma}},\tmty{x}{S},\tmty{x'}{\co{S}}}
      \\
      \gv{\flatten{(\ty{\Gamma},\tmty{x}{T})}}
      & = & \gv{\flatten{\ty{\Gamma}},\tmty{x}{T}}
    \end{array}
    \qquad\qquad
    \begin{array}{lcl}
      \hgv{\flatten{\emptyhyperenv}}
      & = & \hgv{\emptyhyperenv}
      \\
      \hgv{\flatten{(\ty{\hyper{G}\hypersep\ty{\Gamma}})}}
      & = & \hgv{\flatten{\ty{\hyper{G}},\ty{\Gamma}}}
      \\ \\
    \end{array}
    \hfill
  \]
\end{defi}
\begin{defi}[Splitting]
  Splitting converts GV type environments into hyper-environments.
  Given channels $\gv{\{ \tmty{\lock{x_i}{x'_i}}{\lockty{S_i}} \}_{i \in
  1..n}}$ in $\ty{\Gamma}$,
a hyper-environment
      $\hgv{\ty{\hyper{G}}}$ is a \emph{splitting} of
      $\hgv{\ty{\Gamma}}$ if $\hgv{\flatten{\ty{\hyper{G}}}} =
      \gv{\flatten{\ty{\Gamma}}}$ and $\exists \hgv{\ty{\Gamma_1},
          \ldots, \ty{\Gamma_{n + 1}}}$ such that $\hgv{\ty{\hyper{G}}
          = \ty{\Gamma_1} \hypersep \cdots \hypersep \ty{\Gamma_{n +
              1}}}$, and $\treeprop{\hgv{\ty{\hyper{G}}}}{\{ \{ \tm{x_1}, \tm{x'_1} \},
  \ldots, \{ \tm{x_n}, \tm{x'_n} \} \}}$.
\end{defi}
\noindent
A well-typed GV configuration is typeable in HGV under a splitting of its type environment.

\begin{restatable}[Typeability of GV configurations in HGV]{thm}{gvinhgv}\begingroup
    \usingnamespace{hgv}
  \label{thm:gv-in-hgv}
  If $\gv{\gvseq{\ty{\Gamma}}{\config{C}}{R}}$, then there exists some
  $\hgv{\ty{\hyper{G}}}$ such that $\hgv{\ty{\hyper{G}}}$ is a
  splitting of $\gv{\ty{\Gamma}}$ and
  $\hgv{\cseq{\ty{\hyper{G}}}{\config{C}}{R}}$.
  \endgroup
\end{restatable}
\begin{proof}
  By induction on the derivation of $\gv{\gvseq{\ty{\Gamma}}{\config{C}}{T}}$ (see \Cref{appendix:gv-hgv}).
\end{proof}

\begin{exa}
  Consider a configuration where a child thread pings the main thread:{
  \[
    \hgv{\tm{
        \res{x}{y}{(\ppar
          {\child\;(\send\;(\variable{ping}, x))}
          {\main\;(\letpair{()}{y}{\recv\;y}{\wait\;y})})}}}
  \]
  }We can write a GV typing derivation as follows:
  {\footnotesize
  \begin{mathpar}
    \usingnamespace{gv}\inferrule*{
      \inferrule*{
        \gvseq
        {\tmty{x}{\tysend{\tyunit}{\tyends}},\tmty{\variable{ping}}{\tyunit}}
        {\child\;(\send\;(\variable{ping}, x))}
        {\tychild}
        \and
        \gvseq
        {\tmty{y}{\tyrecv{\tyunit}{\tyendr}}}
        {\main\;(\letpair{()}{y}{\recv\;y}{\wait\;y})}
        {\tymain{\tyunit}}
      }{\gvseq
        {\tmty{\lock{x}{y}}{\lockty{\tysend{\tyunit}{\tyends}}},\tmty{\variable{ping}}{\tyunit}}
        {\resd{x}{y}{(\child(\send\;(\variable{ping}, x)) \parallel \main(\letpair{()}{y}{\recv\;y}{\wait\;y}))}}
        {\tyunit}}
    }{\gvseq
      {\tmty{\variable{ping}}{\tyunit}}
      {\resd{x}{y}{(\child(\send\;(\variable{ping}, x)) \parallel \main(\letpair{()}{y}{\recv\;y}{\wait\;y}))}}
      {\tyunit}}
  \end{mathpar}}The corresponding HGV derivation is:
  {\footnotesize
  \begin{mathpar}
    \usingnamespace{hgv}\inferrule*{
      \inferrule*{
        \cseq
        {\tmty{x}{\tysend{\tyunit}{\tyends}},\tmty{\variable{ping}}{\tyunit}}
        {\child\;(\send\;(\variable{ping}, x))}
        {\tychild}
        \and
        \cseq
        {\tmty{y}{\tyrecv{\tyunit}{\tyendr}}}
        {\main\;(\letpair{()}{y}{\recv\;y}{\wait\;y})}
        {\tymain{\tyunit}}
      }{\cseq
        {\tmty{x}{\tysend{\tyunit}{\tyends}},\tmty{\variable{ping}}{\tyunit}
          \hypersep \tmty{y}{\tyrecv{\tyunit}{\tyendr}}}
        {\res{x}{y}{(\ppar
            {\child(\send\;(\variable{ping}, x))}
            {\main(\letpair{()}{y}{\recv\;y}{\wait\;y})}
            )}}
        {\tymain{\tyunit}}}
    }{\cseq
      {\tmty{\variable{ping}}{\tyunit}}
      {\res{x}{y}{(\ppar
          {\child(\send\;(\variable{ping}, x))}
          {\main(\letpair{()}{y}{\recv\;y}{\wait\;y})}
          )}}
      {\tymain{\tyunit}}}
  \end{mathpar}}Note that
  $\hgv{
    \tmty{x}{\tysend{\tyunit}{\tyends}},
    \tmty{\variable{ping}}{\tyunit} \hypersep
    \tmty{y}{\tyrecv{\tyunit}{\tyendr}}
  }$ is a splitting of
  $\gv{
    \tmty{\lock{x}{y}}{\lockty{(\tysend{\tyunit}{\tyends})}},
    \tmty{\variable{ping}}{\tyunit}
  }$.
\end{exa}
\endgroup

\paragraph{Translating HGV to GV}\label{sec:hgv-to-gv}
\begingroup
\sloppypar
As we saw in~\secref{sec:problem-with-gv}, unlike in HGV, equivalence in GV is not
type-preserving. It follows that HGV types strictly more processes
than GV. Let us revisit Lindley and Morris' example
from~\secref{sec:introduction} (adapted to use double-binders), where
    $\gv{
      \gvseq
        {\ty{\Gamma_1}, \ty{\Gamma_2}, \ty{\Gamma_3}}
        { (\nu x x')(\nu y y')(\config{C} \parallel (\config{D} \parallel
        \config{E}))}
        { R_1 \isect R_2 \isect R_3 }
    }
    $
with
    $\gv{\gvseq{\ty{\Gamma_1}, \tmty{x}{S}}{\config{C}}{R_1}}$,
    $\gv{\gvseq{\ty{\Gamma_2}, \tmty{y}{S'}}{\config{D}}{R_2}}$,
    and
    $\gv{\gvseq{\ty{\Gamma_3}, \tmty{x'}{\co{S}}, \tmty{y'}{\co{S'}}}{\config{E}}{R_3}}$.

    The structurally-equivalent term
$\tm{(\nu x x')(\nu y y')((\config{C} \parallel \config{D}) \parallel \config{E})}$
is not typeable in GV, since we cannot split both channels over a
single parallel composition:

{
\begin{mathpar}
    \usingnamespace{gv}
    \inferrule*
    {
        \inferrule*
        {
            \inferrule*
            {
                \nseq
                    {\ty{\Gamma_1}, \ty{\Gamma_2}, \tmty{x}{S} }
                    {\config{C} \parallel \config{D}}
                    {R_1 \isect R_2}
                \\
                \nseq
                {\ty{\Gamma_3}, \tmty{x'}{\co{S}}, \tmty{\lock{y}{y'}}{\lockty{S'}}}
                    {\config{E}}
                    {R_3}
            }
            {
                \nseq
                { \ty{\Gamma_1}, \ty{\Gamma_2}, \ty{\Gamma_3},
                    \tmty{\lock{x}{x'}}{\lockty{S}},
                    \tmty{\lock{y}{y'}}{\lockty{S'}}
                }
                { (\config{C} \parallel \config{D}) \parallel \config{E} }
                { R_1 \isect R_2 \isect R_3 }
            }
        }
        {
            \nseq
            { \ty{\Gamma_1}, \ty{\Gamma_2}, \ty{\Gamma_3}, \tmty{\lock{x}{x'}}{\lockty{S}} }
            { (\nu y y') ((\config{C} \parallel \config{D}) \parallel \config{E}) }
            { R_1 \isect R_2 \isect R_3}
        }
    }
    { \nseq
        {\ty{\Gamma_1}, \ty{\Gamma_2}, \ty{\Gamma_3}}
        {(\nu x x')(\nu y y')((\config{C} \parallel \config{D}) \parallel \config{E})}
        {R_1 \isect R_2 \isect R_3}
    }
\end{mathpar}}However, we \emph{can} type this process in HGV:

{
\begin{mathpar}
    \usingnamespace{hgv}
    \inferrule*
    {
        \inferrule*
        {
            \inferrule*
            {
                \inferrule*
                {
                    \cseq
                    { \ty{\Gamma_1}, \tmty{x}{S} }
                    { \config{C} }
                    { R_1 }
                    \\
                    \cseq
                    { \ty{\Gamma_2}, \tmty{y}{S'} }
                    { \config{D} }
                    { R_2 }
                }
                {
                \cseq
                    { \ty{\Gamma_1}, \tmty{x}{S} \hypersep \ty{\Gamma_2}, \tmty{y}{S'}}
                    { \config{C} \parallel \config{D} }
                    { R_1 \isect R_2 }
                }
                \\
                \cseq
                    { \ty{\Gamma_3}, \tmty{x'}{\co{S}}, \tmty{y'}{\co{S'}}}
                    { \config{E} }
                    { R_3 }
            }
            {
                \cseq
                { \shade{\ty{\Gamma_1}, \tmty{x}{S} \hypersep \ty{\Gamma_2},
                \tmty{y}{S'} \hypersep \ty{\Gamma_3}, \tmty{x'}{\co{S}},
            \tmty{y'}{\co{S'}}} }
                {(\config{C} \parallel \config{D}) \parallel \config{E}}
                { R_1 \isect R_2 \isect R_3 }
            }
        }
        {
            \cseq
                {
                \ty{\Gamma_1}, \tmty{x}{S} \hypersep \ty{\Gamma_2},
                \ty{\Gamma_3}, \tmty{x'}{\co{S}}
                }
                {(\nu y y')((\config{C} \parallel \config{D}) \parallel \config{E})}
                { R_1 \isect R_2 \isect R_3 }
        }
    }
    { \cseq
        {\ty{\Gamma_1}, \ty{\Gamma_2}, \ty{\Gamma_3} }
        {(\nu x x')(\nu y y')((\config{C} \parallel \config{D}) \parallel \config{E})}
        { R_1 \isect R_2 \isect R_3 }
    }
\end{mathpar}
}Note in particular the shaded hyper-environment, which includes
hyper-environment separators to separate endpoints $x$ and $x'$, as well as $y$
and $y'$. It follows that, unlike in GV, \emph{both} channels can be split over
the same parallel composition. Similarly, the hyper-environment separator allows
$\tm{\config{C}}$ and $\tm{\config{D}}$ to be composed \emph{without} sharing any
channels.

Although HGV types more processes, every well-typed HGV configuration typeable
under a singleton hyper-environment $\hgv{\ty{\Gamma}}$ is \emph{equivalent} to
a well-typed GV configuration, which we show using tree canonical forms.
\begin{restatable}{prop}{hgvingv}\label{lem:hgv-to-gv-tcf}
    \begingroup
    \usingnamespace{hgv}
  Suppose $\hgv{\cseq{\ty{\Gamma}}{\config{C}}{R}}$ where
  $\hgv{\tm{\config{C}}}$ is in tree canonical form. Then,
  $\gv{\gvseq{\ty{\Gamma}}{\config{C}}{R}}$.
  \endgroup
\end{restatable}
\begin{proof}
    By induction on the derivation of
    $\hgv{\cseq{\ty{\Gamma}}{\config{C}}{R}}$, making use of an inductive
    definition of tree canonical forms. See~\Cref{appendix:gv-hgv} for
    details.
\end{proof}

\begin{rem}
  It is not the case that every HGV configuration typeable under an
  \emph{arbitrary} hyper-environment $\hgv{\ty{\hyper{H}}}$ is
  equivalent to a well-typed GV configuration.  This is because open
  HGV configurations can form \emph{forest} process structures,
  whereas (even open) GV configurations must form a \emph{tree}
  process structure.
\end{rem}
Since we can write all well-typed HGV configurations in canonical
form, and HGV tree canonical forms are typeable in GV, it
follows that every well-typed HGV configuration typeable under a
single type environment is equivalent to a well-typed GV
configuration.
\begin{cor}\label{thm:hgv-equiv-gv}
  If $\hgv{\cseq{\ty{\Gamma}}{\config{C}}{R}}$, then there exists some
  $\tm{\config{D}}$ such that $\tm{\config{C}}\equiv\tm{\config{D}}$ and
  $\gv{\gvseq{\ty{\Gamma}}{\config{D}}{R}}$.
\end{cor}
\endgroup

 }  
{\section{Relation between HGV and HCP}\label{sec:relation-to-cp}
In this section, we explore two translations, from HGV to HCP and from
HCP to HGV, together with their operational correspondence results.

\paragraph{Hypersequent CP}
\begingroup
\usingnamespace{hcp}

HCP~\cite{montesip18:ct,kokkemp19:tlla} is a session-typed process
calculus with a correspondence to CLL, which exploits hypersequents to
fix extensibility and modularity issues with CP.

\begin{figure}
  \small \usingnamespace{hcp}

  \headersig{Typing rules for processes}{$\seq{P}{\ty{\hyper{G}}}$}
  \begin{mathpar}
    \inferrule*[lab=TP-Link]{
    }{\seq{\link[\ty{A}]{x}{y}}{\tmty{x}{A},\tmty{y}{\co{A}}}}
    ~~
    \inferrule*[lab=TP-New]{
      \seq{P}{\ty{\hyper{G}}\hypersep\ty{\Gamma},\tmty{x}{A}\hypersep\ty{\Delta},\tmty{y}{\co{A}}}
    }{\seq{\res{x}{y}{P}}{\ty{\hyper{G}}\hypersep\ty{\Gamma},\ty{\Delta}}}
    ~~
    \inferrule*[lab=TP-Par]{
      \seq{P}{\ty{\hyper{G}}}
      \and
      \seq{Q}{\ty{\hyper{H}}}
    }{\seq{\ppar{P}{Q}}{\ty{\hyper{G}}\hypersep\ty{\hyper{H}}}}
    ~~
    \inferrule*[lab=TP-Halt]{
    }{\seq{\halt}{\emptyhyperenv}}
    ~~
    \inferrule*[lab=TP-Close]{
      \seq{P}{\emptyhyperenv}
    }{\seq{\close{x}{P}}{\tmty{x}{\tyone}}}

    \inferrule*[lab=TP-Wait]{
      \seq{P}{\ty{\Gamma}}
    }{\seq{\wait{x}{P}}{\ty{\Gamma},\tmty{x}{\tybot}}}
    ~~
    \inferrule*[lab=TP-Send]{
      \seq{P}{\ty{\Gamma},\tmty{y}{A}\hypersep\ty{\Delta},\tmty{x}{B}}
  }{\seq{\send{x}{y}{P}}{\ty{\Gamma},\ty{\Delta},\tmty{x}{\tytens{A}{B}}}}
    ~~
    \inferrule*[lab=TP-Recv]{
      \seq{P}{\ty{\Gamma},\tmty{y}{A},\tmty{x}{B}}
    }{\seq{\recv{x}{y}{P}}{\ty{\Gamma},\tmty{x}{\typarr{A}{B}}}}
    ~~
    \inferrule*[lab=TP-Offer-Absurd]{
    }{\seq{\absurd{x}}{\ty{\Gamma},\tmty{x}{\tytop}}}

    \inferrule*[lab=TP-Select-Inl]{
      \seq{P}{\ty{\Gamma},\tmty{x}{A}}
    }{\seq{\inl{x}{P}}{\ty{\Gamma},\tmty{x}{\typlus{A}{B}}}}
    ~~
    \inferrule*[lab=TP-Select-Inr]{
      \seq{P}{\ty{\Gamma},\tmty{x}{B}}
    }{\seq{\inr{x}{P}}{\ty{\Gamma},\tmty{x}{\typlus{A}{B}}}}
    ~~
    \inferrule*[lab=TP-Offer]{
      \seq{P}{\ty{\Gamma},\tmty{x}{A}}
      \and
      \seq{Q}{\ty{\Gamma},\tmty{x}{B}}
    }{\seq{\offer{x}{P}{Q}}{\ty{\Gamma},\tmty{x}{\tywith{A}{B}}}}
  \end{mathpar}
  \headersig{Duality}{$\co{\ty{A}}$}
  \[
    \begin{array}{lcl}
      \ty{\co{(\tytens{A}{B})}} & = & \ty{\typarr{\co{A}}{\co{B}}} \\
      \ty{\co{(\typarr{A}{B})}} & = & \ty{\tytens{\co{A}}{\co{B}}}
    \end{array}
    \quad
    \begin{array}{lcl}
      \ty{\co{(\tyone)}} & = & \ty{\tybot} \\
      \ty{\co{(\tybot)}} & = & \ty{\tyone}
    \end{array}
    \quad
    \begin{array}{lcl}
      \ty{\co{(\typlus{A}{B})}} & = & \ty{\tywith{\co{A}}{\co{B}}} \\
      \ty{\co{(\tywith{A}{B})}} & = & \ty{\typlus{\co{A}}{\co{B}}}
    \end{array}
    \quad
    \begin{array}{lcl}
      \ty{\co{(\tynil)}} & = & \ty{\tytop} \\
      \ty{\co{(\tytop)}} & = & \ty{\tynil}
    \end{array}
  \]
  \caption{HCP, duality and typing rules for processes.}
  \label{fig:hcp-typing}
\end{figure}

Types ($\ty{A}$, $\ty{B}$) consist of the connectives of
linear logic: the multiplicative operators ($\ty{\tytensop}$,
$\ty{\typarrop}$) and units ($\ty{\tyone}$, $\ty{\tybot}$) and the
additive operators ($\ty{\typlusop}$, $\ty{\tywithop}$) and units
($\ty{\tynil}$, $\ty{\tytop}$).
\[
  \ty{A}, \ty{B}
  \bnfdef \ty{\tyone}
  \sep      \ty{\tybot}
  \sep      \ty{\tynil}
  \sep      \ty{\tytop}
  \sep      \ty{\tytens{A}{B}}
  \sep      \ty{\typarr{A}{B}}
  \sep      \ty{\typlus{A}{B}}
  \sep      \ty{\tywith{A}{B}}
\]
Type environments ($\ty{\Gamma}$, $\ty{\Delta}$) associate names with
types. Hyper-environments ($\ty{\hyper{G}}$, $\ty{\hyper{H}}$) are
collections of type environments. The empty type environment and
hyper-environment are written $\emptyenv$ and $\emptyhyperenv$,
respectively. Names in type and hyper-environments must be unique
and environments may be combined, written $\ty{\Gamma},\ty{\Delta}$
and $\ty{\hyper{G}}\hypersep\ty{\hyper{H}}$, only if they are
disjoint.

Processes ($\tm{P}$, $\tm{Q}$) are a variant of the $\pi$-calculus
with forwarding~\cite{sangiorgi96,boreale98}, bound
output~\cite{sangiorgi96}, and double
binders~\cite{vasconcelos12:fundamentals}.
The syntax of processes is given by the typing rules
(\Cref{fig:hcp-typing}), which are standard for
HCP~\cite{montesip18:ct,kokkemp19:tlla}:
$\tm{\link[\ty{A}]{x}{y}}$ forwards messages between $\tm{x}$ and $\tm{y}$;
$\tm{\res{x}{y}{P}}$ creates a channel with endpoints $\tm{x}$ and $\tm{y}$, and continues as $\tm{P}$;
$\tm{\ppar{P}{Q}}$ composes $\tm{P}$ and $\tm{Q}$ in parallel;
$\tm{\halt}$ is the terminated process;
$\tm{\send{x}{y}{P}}$ creates a new channel, outputs one endpoint over $\tm{x}$, binds the other to $\tm{y}$, and continues as $\tm{P}$;
$\tm{\recv{x}{y}{P}}$ receives a channel endpoint, binds it to $\tm{y}$, and continues as $\tm{P}$;
$\tm{\close{x}{P}}$ and $\tm{\wait{x}{P}}$ close $\tm{x}$ and continue as $\tm{P}$;
$\tm{\inl{x}{P}}$ and $\tm{\inr{x}{P}}$ {make} a binary choice;
$\tm{\offer{x}{P}{Q}}$ {offers} a binary choice;
and $\tm{\absurd{x}}$ offers a nullary choice.
As HCP is synchronous, the only difference between $\tm{\send{x}{y}{P}}$ and
$\tm{\recv{x}{y}{P}}$ is their typing (and similarly for $\tm{\close{x}{P}}$ and
$\tm{\wait{x}{P}}$).
We write \emph{unbound} send as $\tm{\usend{x}{y}{P}}$ (short for
$\tm{\send{x}{z}{(\ppar{\link{y}{z}}{P})}}$), and synchronisation as
$\tm{\ping{x}{P}}$ (short for
$\tm{\send{x}{z}{(\ppar{\close{z}{\halt}}{P})}}$) and
$\tm{\pong{x}{P}}$ (short for $\tm{\recv{x}{z}{\wait{z}{P}}}$).
Duality is standard and is involutive, \ie,
$\ty{\co{(\co{A})}}=\ty{A}$.

We define a standard structural congruence ($\equiv$) similar to that
of HGV, \ie, parallel composition is commutative and associative, we
can commute name restrictions, swap the order of endpoints, swap
links, and have scope extrusion (similar to \Cref{fig:hgv-reduction}).
Note that since we base our formal developments on an LTS semantics, structural
congruence is not required for reduction.

\hcp{
  \begin{mathpar}
    \tm{\link[\ty{A}]{x}{y}} \equiv \tm{\link[\ty{\co{A}}]{y}{x}}

    \tm{\ppar{P}{\halt}} \equiv \tm{\tm{P}}

    \tm{\ppar{P}{Q}} \equiv \tm{\ppar{Q}{P}}

    \tm{\ppar{P}{(\ppar{Q}{R})}} \equiv \tm{\ppar{(\ppar{P}{Q})}{R}}

    \tm{(\nu x x')(\nu y y') P} \equiv \tm{(\nu y y')(\nu x x') P}

    \tm{(\nu x y)P} \equiv \tm{(\nu y x)P}

    \tm{(\nu x y)(\ppar{P}{Q})} \equiv \tm{\ppar{P}{(\nu x y)}{Q}} \quad
    \text{if } \tm{x}, \tm{y} \not\in \fv(\tm{P})
  \end{mathpar}
}

\begin{figure}[t]
  \small \usingnamespace{hcp}

  \header{Action rules}
  \begin{mathpar}
    \inferrule*[lab=Act-Pref]{}{\tm{\pi.P} \lto{\pi} \tm{P}}

    \inferrule*[lab=Act-Link\textsubscript{1}]{}{\tm{\link{x}{y}} \lto{\lablink{x}{y}} \tm{\halt}}

    \inferrule*[lab=Act-Link\textsubscript{2}]{}{\tm{\link{x}{y}} \lto{\lablink{y}{x}} \tm{\halt}}

    \inferrule*[lab=Act-Off-Inl]{}{\tm{\offer{x}{P}{Q}} \lto{\laboffinl{x}} \tm{P}}

    \inferrule*[lab=Act-Off-Inr]{}{\tm{\offer{x}{P}{Q}} \lto{\laboffinr{x}} \tm{Q}}
  \end{mathpar}

  \header{Communication Rules}
  \begin{mathpar}
\inferrule*[lab=Alp-Link]{
      \tm{P}\lto{\lablink{x}{z}}\tm{P'}
    }{\tm{\res{x}{y}{P}}\ato\tm{\subst{P'}{z}{y}}}

    \inferrule*[lab=Bet-Send]{
      \tm{P}\lto{\ppar{\labsend{x}{x'}}{\labrecv{y}{y'}}}\tm{P'}
    }{\tm{\res{x}{y}{P}}\bto\tm{\res{x}{y}{\res{x'}{y'}{P'}}}}

    \inferrule*[lab=Bet-Close]{
      \tm{P}\lto{\ppar{\labclose{x}}{\labwait{y}}}\tm{P'}
    }{\tm{\res{x}{y}{P}}\bto\tm{P'}}

    \inferrule*[lab=Bet-Inl]{
      \tm{P}\lto{\ppar{\labselinl{x}}{\laboffinl{y}}}\tm{P'}
    }{\tm{\res{x}{y}{P}}\bto\tm{\res{x}{y}{P'}}}

    \inferrule*[lab=Bet-Inr]{
      \tm{P}\lto{\ppar{\labselinr{x}}{\laboffinr{y}}}\tm{P'}
    }{\tm{\res{x}{y}{P}}\bto\tm{\res{x}{y}{P'}}}
  \end{mathpar}

  \header{Structural Rules}
  \begin{mathpar}
    \inferrule*[lab=Str-Res]{
      \tm{P}\lto{\ell}\tm{P'}
      \and
      \tm{x},\tm{y}\not\in\cn(\tm{\ell})
    }{\tm{\res{x}{y}{P}}\lto{\ell}\tm{\res{x}{y}{P'}}}

    \inferrule*[lab=Str-Par\textsubscript{1}]{
      \tm{P}\lto{\ell}\tm{P'}
      \and
      \bn(\tm{\ell})\cap\fn(\tm{Q})=\varnothing
    }{\tm{\ppar{P}{Q}}\lto{\ell}\tm{\ppar{P'}{Q}}}

    \inferrule*[lab=Str-Par\textsubscript{2}]{
      \tm{Q}\lto{\ell}\tm{Q'}
      \and
      \bn(\tm{\ell})\cap\fn(\tm{P})=\varnothing
    }{\tm{\ppar{P}{Q}}\lto{\ell}\tm{\ppar{P}{Q'}}}

    \inferrule*[lab=Str-Syn]{
      \tm{P}\lto{\ell}\tm{P'}
      \and
      \tm{Q}\lto{\ell'}\tm{Q'}
      \and
      \bn(\ell)\cap\bn(\ell')=\varnothing
    }{\tm{\ppar{P}{Q}}\lto{\lbl\parallel\lbl'}\tm{\ppar{P'}{Q'}}}
  \end{mathpar}
  \caption{HCP, label transition semantics.}
  \label{fig:hcp-lts}
\end{figure}

We define the labelled transition system for HCP as a small refinement of the
LTS for the additive-multiplicative fragment of the $\pi\text{LL}$ calculus
introduced by Montesi and Peressotti~\cite{montesip21:pill}, in turn inspired by
their previous system \emph{CT}~\cite{montesip18:ct}. The LTS is identical, save
for the fact that we distinguish two types of internal actions.
Action labels $\tm{l}$ represent the actions that a
process can fire.
Prefixes $\tm{\pi}$ are a convenient subset of action labels which can be
written as prefixes to processes, \ie, $\tm{\pi.P}$.
Transition labels $\tm\translbl$ include action labels and the parallel
composition of two action labels, along with \emph{internal} actions
$\tm\alpha$, $\tm\beta$, and $\tm\tau$.
The LTS gives rise to two types of internal action: $\ta$ represents only the
evaluation of links as \emph{renaming}, and $\tb$ represents only
\emph{communication}.
Labels $\tau$ arise only due to saturated transition
(Definition~\ref{def:saturated-transition}) and are not produced by the rules in
the LTS.
\[
  \begin{array}{lrcl}
    \text{Prefixes}
     & \tm{\pi}
     & \bnfdef
     & \tm{\labsend{x}{y}}
    \sep \tm{\labclose{x}}
    \sep \tm{\labrecv{x}{y}}
    \sep \tm{\labwait{x}}
    \sep \tm{\labselinl{x}}
    \sep \tm{\labselinr{x}}
    \\
    \text{Action Labels}
     & \tm{\lbl}
     & \bnfdef
     & \tm{\pi}
    \sep \tm{\lablink{x}{y}}
    \sep \tm{\laboffinl{x}}
    \sep \tm{\laboffinr{x}}
    \\
    \text{Transition Labels}
     & \tm{\translbl}
     & \bnfdef
     & \tm{\lbl}
    \sep \tm{\lbl \parallel \lbl'}
    \sep \tm{\ta}
    \sep \tm{\tb}
  \end{array}
\]

We let $\tm{\ell_{x}}$ range over labels on $\tm{x}$:
$\tm{\lablink{x}{y}}$, $\tm{\labsend{x}{y}}$, $\tm{\labclose{x}}$,
\etc. Labelled transition $\lto{\ell}$ is defined
in~\Cref{fig:hcp-lts}. We write $\lto{\ell}\lto{\ell'}$ for the
composition of $\lto{\ell}$ and $\lto{\ell'}$, $\lto{\ell^+}$ for the
transitive closure of $\lto{\ell}$, and $\lto{\ell^*}$ for the
reflexive-transitive closure of $\lto{\ell}$. We write
$\bn(\tm{\ell})$ and $\fn(\tm{\ell})$ for the bound and free names
contained in $\tm{\ell}$, respectively.  We write $\cn(\tm{\ell})$ for
all names in $\tm{\ell}$, \ie $\cn(\tm{\ell}) = \fn(\tm{\ell}) \cup
  \bn(\tm{\ell})$.

\paragraph{Metatheory}
Transitions preserve typeability. Since internal actions occur only under binders,
they are typable under the same hyper-environment.

\begin{restatable}[Type Preservation]{thm}{hcppres}
  Suppose $\seq{P}{\ty{\hyper{G}}}$ and $\tm{P} \lto{\ell} \tm{Q}$.
  \begin{itemize}
    \item If $\tm\ell$ is internal, then $\seq{Q}{\ty{\hyper{G}}}$.
    \item If $\tm\ell$ is not internal, then there exists some $\ty{\hyper{H}}$ such
          that $\seq{Q}{\ty{\hyper{H}}}$.
  \end{itemize}
\end{restatable}
\begin{proof}
  Following the approach of~\cite{kokkemp19:popl, montesip18:ct,
    montesip21:pill}, type preservation is
  established by defining proof transformations on typing derivations of each
  reducing process. The only difference with respect to~\cite{montesip18:ct,
    montesip21:pill} arises due to our
  separate treatment of $\alpha$ and $\beta$ actions, which does not
  materially impact the proof.
\end{proof}

Similarly, our LTS for HCP satisfies progress.
Following~\cite{kokkemp19:popl, montesip21:pill}, the key intermediate
step is to note that for every type environment in a
hyper-environment, there is some free name which can be acted
upon. Again, the stratification of internal actions does not
materially impact the proof.

\begin{thm}[Progress]
  If $\seq{\tm{P}}{\ty{\hyper{H}}}$ and $\tm{P} \not\equiv \tm{\halt}$, then there
  exist some $\tm\ell, \tm{Q}$ such that $\tm{P} \lto{\tm\ell} \tm{Q}$.
\end{thm}

\paragraph{Behavioural Theory}
The behavioural theory for HCP follows Kokke~\etal~\cite{kokkemp19:popl}, except
that we distinguish two subrelations of weak bisimilarity, following the subtypes of internal actions.
\begin{defi}[Strong bisimulation and strong bisimilarity]
  A symmetric relation $\rel{R}$ on processes is a \emph{strong bisimulation}
  if $\tm{P}\mathbin{\rel{R}}\tm{Q}$ implies that if
  $\tm{P}\lto{\ell}\tm{P'}$, then $\tm{Q}\lto{\ell}\tm{Q'}$ for some $\tm{Q'}$
  such that $\tm{P'}\mathbin{\rel{R}}\tm{Q'}$. \emph{Strong bisimilarity} is
  the largest relation $\sbis$ that is a strong bisimulation.
\end{defi}
\begin{defi}[Saturated transition]\label{def:saturated-transition}
The \em{$\lblset$-saturated transition relation}, for $\lblset \subseteq \{\tm\ta, \tm\tb\}$, is the smallest relation $\slto[\lblset]{}$ closed under the following rules, with saturated transition labels $\ell$ ranging over transition labels and the distinguished label $\tau$:
\begin{mathpar}
    \inferrule
    { }
    { \tm{P} \slto[\lblset]{\tau} \tm{P}}

    \inferrule
    { \tm{P} \slto[\lblset]{\tau} \tm{P'} \lto{\ell} \tm{Q'} \slto[\lblset]{\tau} \tm{Q} \\
      \ell \in \lblset
    }
    { \tm{P} \slto[\lblset]{\tau} \tm{Q} }

    \inferrule
    { \tm{P} \slto[\lblset]{\tau} \tm{P'} \lto{\ell} \tm{Q'} \slto[\lblset]{\tau} \tm{Q} \\
      \ell \notin \lblset
    }
    { \tm{P} \slto[\lblset]{\ell} \tm{Q} }
  \end{mathpar}
We write $\slto[\tm\ell]{}$ as shorthand for $\slto[\{ \tm\ell \}]{}$, and
  we write $\slto[]{}$ as shorthand for $\slto[\{ \tm\alpha, \tm\beta \}]{}$.
\end{defi}
\begin{defi}[Weak bisimulation and weak bisimilarity]
  \label{def:bisimulation}
  A symmetric relation $\rel{R}$ on processes is an $\lblset$-bisimulation, for
  $\lblset\subseteq\{\tm\ta,\tm\tb\}$, if $\tm{P}\mathbin{\rel{R}}\tm{Q}$ implies that
  if $\tm{P}\slto[\lblset]{\ell'}\tm{P'}$, then
  $\tm{Q}\slto[\lblset]{\ell'}\tm{Q'}$ for some $\tm{Q'}$ such that
  $\tm{P'}\mathbin{\rel{R}}\tm{Q'}$.
The $\lblset$-bisimilarity relation is the
  largest relation $\bis_\lblset$ that is an $\lblset$-bisimulation.
We write $\bis$ as shorthand for $\bis_{\{ \ta, \tb \}}$.
\end{defi}
\begin{restatable}{lem}{corhcpsbisbis}
  \label{cor:hcp-sbis-to-bis}
  Structural congruence, strong bisimilarity and the various forms of weak
  bisimilarity are related as follows:
  \begin{mathpar}
    {\equiv}\subset{\sbis}

    {\sbis}\subset{\bis}

    {\sbis}\subset{\bis_\ta}

    {\sbis}\subset{\bis_\tb}
  \end{mathpar}
\end{restatable}

\paragraph{Differences with previous version}
The LTS in~\Cref{fig:hcp-lts} is similar to that in the previous version of this
work~\cite{FKDLM21}, with the exception that we have removed the rules
\rulename{Tau-Alp} and \rulename{Tau-Bet}:
\begin{mathpar}
  \evil{\inferrule*[lab=Tau-Alp]{
      \tm{P}\ato\tm{P'}
    }{\tm{P}\tto\tm{P'}}}

  \evil{\inferrule*[lab=Tau-Bet]{
      \tm{P}\bto\tm{P'}
    }{\tm{P}\tto\tm{P'}}}
\end{mathpar}

To see why these rules are problematic, consider processes
$\tm{P}=\tm{\res{x}{y}{(\ppar{\link{z}{x}}{\close{y}{\halt}})}}$ and
$\tm{Q}=\tm{\close{z}{\halt}}$.  Following~\Cref{def:bisimulation}, $\tm{P}$ and
$\tm{Q}$ are $\ta$-bisimilar, as $\tm{P}$ only has the $\ta$-transition
$\tm{P}\ato\tm{Q}$ and $\tm{Q}$ has no transitions.  In the previous version,
\rulename{Tau-Alp} gave $\tm{P}$ the derived $\tau$-transition
$\tm{P}\tto\tm{Q}$, which meant that $\tm{P}\not\bis_\ta\tm{Q}$, as
$\tm{Q}\centernot\stto\tm{Q}$. Therefore \rulename{Tau-Alp} collapses $\bis_\ta$
to $\sbis$ and \rulename{Tau-Bet} collapses $\bis_\tb$ to $\sbis$.

The solution we adopted was to remove \rulename{Tau-Alp} and \rulename{Tau-Bet}
from the label transition relation $\lto{}$, and instead lift $\ta$- and
$\tb$-transitions to $\tau$-transitions in the definition of saturated
transition\footnote{We thank Marco Peressotti for notifying us of the error and
suggesting the fix.}.
\endgroup

\paragraph{Translating HGV to HCP}\label{sec:hgv-to-hcp}
We factor the translation from HGV to HCP into two translations: (1) a
translation into \fgHGV, a fine-grain call-by-value~\cite{levypt03}
variant of HGV, which makes control flow explicit; and (2) a
translation from \fgHGV to HCP. In so doing, we can concentrate on the essence
of the translations as opposed to concerning ourselves with administrative
reductions.

\paragraph{\fgHGV}
We define \fgHGV as a refinement of HGV in which any non-trivial term
must be named by a let-binding before being used.
While $\hgv{\tm{\calcwd{let}}}$ is syntactic sugar in HGV, it is part of the core language
in \fgHGV. Correspondingly, the reduction rule for
$\hgv{\tm{\calcwd{let}}}$ follows from
the encoding in HGV, \ie
$\hgv{\tm{\letbind{x}{V}{M}}\tred\tm{\subst{M}{V}{x}}}$.
\[
  \small
  \usingnamespace{hgv}
  \begin{array}{lrcl}
    \text{Terms}
     & \tm{L}, \tm{M}, \tm{N}
     & \bnfdef                & \tm{V}
    \sep      \tm{\letbind{x}{M}{N}}
    \sep      \tm{V\;W}                                            \\
     &                        & \sep         & \tm{\letunit{V}{M}}
    \sep      \tm{\letpair{x}{y}{V}{M}}                            \\
     &                        & \sep         & \tm{\absurd{V}}
    \sep      \tm{\casesum{V}{x}{M}{y}{N}}
    \\
    \text{Values}
     & \tm{V}, \tm{W}
     & \bnfdef                & \tm{x}
    \sep      \tm{K}
    \sep      \tm{\lambda x.M}
    \sep      \tm{\unit}
    \sep      \tm{\pair{V}{W}}
    \sep      \tm{\inl{V}}
    \sep      \tm{\inr{V}}                                         \\

    \text{Evaluation contexts}
     & \tm{E}
     & \bnfdef                & \tm{\hole}
    \sep      \tm{\letbind{x}{E}{M}}                               \\
    \text{Thread contexts}
     & \tm{\conf{F}}
     & \Coloneqq              & \tm{\phi\;E}
  \end{array}
\]

\begin{rem}
  Fine-grain call-by-value $\lambda$-calculi typically include an explicit
  $\hgv{\tm{\calcwd{return}\:V}}$ construct to embed values into the term language.
  As there is no difference between the shapes of the value and term typing
  judgements, we allow ourselves to embed values directly for simplicity.
\end{rem}

We can \emph{na\"ively} translate HGV to \fgHGV ($\fgcbv{\cdot}$) by
let-binding each subterm in a value position, \eg,
$\hgv{\tm{\fgcbv{\inl{M}}} =
    \tm{\letbind{z}{\fgcbv{M}}{\inl{z}}}}$.

\begin{defi}[Na\"ive translation of HGV to \fgHGV]
  \label{def:hgv-to-fghgv}
  \[
  \hspace*{-5pt}
    \begin{array}{lcl}
      \hgv{\tm{\fgcbv{x}}}
       &\!=\!& \hgv{\tm{\ret{x}}}
      \\
      \hgv{\tm{\fgcbv{\lambda x.M}}}
       &\!=\!& \hgv{\tm{\ret{\lambda x.\fgcbv{M}}}}
      \\
      \hgv{\tm{\fgcbv{L\;M}}}
       &\!=\!& \hgv{\tm{\letbind{x}{\fgcbv{L}}{\letbind{y}{\fgcbv{M}}{x\;y}}}}
      \\
      \hgv{\tm{\fgcbv{\unit}}}
       &\!=\!& \hgv{\tm{\ret{\unit}}}
      \\
      \hgv{\tm{\fgcbv{\letunit{L}{M}}}}
       &\!=\!& \hgv{\tm{\letbind{z}{\fgcbv{L}}{\letunit{z}{\fgcbv{M}}}}}
      \\
      \hgv{\tm{\fgcbv{\pair{M}{N}}}}
       &\!=\!& \hgv{\tm{\letbind{x}{\fgcbv{M}}{\letbind{y}{\fgcbv{N}}{\ret{\pair{x}{y}}}}}}
      \\
      \hgv{\tm{\fgcbv{\letpair{x}{y}{L}{M}}}}
       &\!=\!& \hgv{\tm{\letbind{z}{\fgcbv{L}}{\letpair{x}{y}{z}{\fgcbv{M}}}}}
      \\
      \hgv{\tm{\fgcbv{\inl{M}}}}
       &\!=\!& \hgv{\tm{\letbind{z}{\fgcbv{M}}{\ret{\inl{z}}}}}
      \\
      \hgv{\tm{\fgcbv{\inr{M}}}}
       &\!=\!& \hgv{\tm{\letbind{z}{\fgcbv{M}}{\ret{\inr{z}}}}}
      \\
      \hgv{\tm{\fgcbv{\casesum{L}{x}{M}{y}{N}}}}
       &\!=\!& \hgv{\tm{\letbind{z}{\fgcbv{L}}{\casesum{z}{x}{\fgcbv{M}}{y}{\fgcbv{N}}}}}
      \\
      \hgv{\tm{\fgcbv{\absurd{L}}}}
       &\!=\!& \hgv{\tm{\letbind{z}{\fgcbv{L}}{\absurd{z}}}}
      \\
    \end{array}
  \]
\end{defi}

Standard techniques can be used to avoid administrative
redexes~\cite{Plotkin75,DanvyMN07}.
We give a full definition of \fgHGV in Appendix~\ref{appendix:hcp-to-hgv}.

\paragraph{\fgHGV to HCP}
The translation from \fgHGV to HCP is given in
\Cref{fig:fghgv-to-hcp}. All control flow is encapsulated in values
and let-bindings. We define a pair of translations on types,
$\hcp{\ty{\gvcpdown{\cdot}}}$ and $\hcp{\ty{\gvcpup{\cdot}}}$, such
that $\hcp{\ty{\gvcpdown{T}}}=\hcp{\ty{\co{\gvcpup{T}}}}$. We extend
these translations pointwise to type environments and
hyper-environments. We define translations on configurations
($\tm{\gvcpcnf{\cdot}{r}}$), terms ($\tm{\gvcpcom{\cdot}{r}}$) and
values ($\tm{\gvcpval{\cdot}{r}}$), where $\tm{r}$ is a fresh name
denoting a distinguished output channel.

We translate an HGV sequent
$\hgv{\cseq{\ty{\hyper{G}}\hypersep\ty{\Gamma}}{\conf{C}}{T}}$ as
$\hcp{\seq{\gvcpcnf{\conf{C}}{r}}{\ty{\gvcpdown{\hyper{G}}\hypersep\gvcpdown{\Gamma}},\tmty{r}{\co{\gvcpdown{T}}}}}$,
where $\ty{\Gamma}$ is the type environment corresponding to the main
thread.
The translation of computations includes synchronisation action in
order to faithfully simulate a call-by-value reduction strategy.
The (term) translation of a value $\hgv{\tm{\gvcpcom{\ret{V}}{r}}}$
immediately pings the output channel $\tm{r}$ to announce that it is a
value. The translation of a let-binding
$\hgv{\tm{\gvcpcom{\letbind{w}{M}{N}}{r}}}$ first evaluates $\tm{M}$
to a value, which then pings the internal channel $\tm{x}/\tm{x'}$ and
unblocks the continuation $\hcp{\tm{\pong{x}{\gvcpcom{N}{r}}}}$.
The translations of main and child threads each make use of an
internal result channel.
The translation of a child thread consumes the yielded unit endpoint
once the child thread has terminated.
The translation of the main thread forwards the result value along the
external output channel once the main thread has terminated.

\begin{figure}[t]
  \small 

  \headersig{Translation on types}{$\ty{\gvcpdown{T}}$ and $\ty{\gvcpup{T}}$}
  \[
  \begin{array}{@{}c@{}}
    \begin{array}{@{}c@{\qquad}c@{\qquad}c@{}}
      \begin{array}[t]{@{}l@{~}c@{~}l@{}}
        \hgv{\ty{\gvcpdown{\tysend{T}{S}}}}
        & = & \hcp{\ty{\tytens{\co{\gvcpdown{T}}}{\gvcpdown{S}}}} \\
        \hgv{\ty{\gvcpdown{\tyrecv{T}{S}}}}
        & = & \hcp{\ty{\typarr{\co{\gvcpdown{T}}}{\gvcpdown{S}}}} \\
      \end{array}
      &
      \begin{array}[t]{@{}l@{~}c@{~}l@{}}
        \hgv{\ty{\gvcpdown{\tyends}}}
        & = & \hcp{\ty{\tyone}} \\
        \hgv{\ty{\gvcpdown{\tyendr}}}
        & = & \hcp{\ty{\tybot}} \\
      \end{array}
      &
      \begin{array}[t]{@{}l@{~}c@{~}l@{}}
        \hgv{\ty{\gvcpdown{T}}}
        & = & \hcp{\ty{\co{\gvcpup{T}}}},\\
        &   & \text{if $\ty{T}$ is not a session type}
      \end{array}
    \end{array}
    \bigskip
    \\
    \begin{array}{@{}c@{\qquad}c@{\qquad}c@{}}
      \begin{array}[t]{@{}l@{~}c@{~}l@{}}
        \hgv{\ty{\gvcpup{\typrod{T}{U}}}}
        & = & \hcp{\ty{\tytens{\gvcpup{T}}{\gvcpup{U}}}} \\
        \hgv{\ty{\gvcpup{\tysum{T}{U}}}}
        & = & \hcp{\ty{\typlus{\gvcpup{T}}{\gvcpup{U}}}} \\
      \end{array}
      &
      \begin{array}[t]{@{}l@{~}c@{~}l@{}}
        \hgv{\ty{\gvcpup{\tyunit}}}
        & = & \hcp{\ty{\tyone}} \\
        \hgv{\ty{\gvcpup{\tyvoid}}}
        & = & \hcp{\ty{\tynil}} \\
      \end{array}
      &
      \begin{array}[t]{@{}l@{~}c@{~}l@{}}
        \hgv{\ty{\gvcpup{\tylolli{T}{U}}}}
        & = & \hcp{\ty{\typarr{\co{\gvcpup{T}}}{(\tytens{\tyone}{\gvcpup{U}})}}} \\
        \hgv{\ty{\gvcpup{S}}}
        & = & \hcp{\ty{\co{\gvcpdown{S}}}}
      \end{array}
    \end{array}
  \end{array}
  \]

  \headersig{Translation on configurations, terms, and values}{$\tm{\gvcpcnf{\tm{\conf{C}}}{\tm{r}}}$, $\tm{\gvcpcom{\tm{M}}{\tm{r}}}$, and $\tm{\gvcpval{\tm{V}}{\tm{r}}}$}
\[
\begin{array}{@{}c@{}}
\begin{array}{@{}l@{\qquad}l@{\qquad}l@{}}
      \begin{array}[t]{@{}l@{~}c@{~}l@{}}
        \hgv{\tm{\gvcpcnf{\child\;M}{r}}}
        & = & \hcp{\tm{\res{y}{y'}{(\ppar{\gvcpcom{M}{y}}{\pong{y'}{\close{y'}{\halt}}})}}}
        \\
        \hgv{\tm{\gvcpcnf{\main\;M}{r}}}
        & = & \hcp{\tm{\res{y}{y'}{(\ppar{\gvcpcom{M}{y}}{\pong{y'}{\link{y'}{r}}})}}}
        \\
\end{array}
&
      \begin{array}[t]{@{}l@{~}c@{~}l@{}}
        \hgv{\tm{\gvcpcnf{\res{x}{x'}{\conf{C}}}{r}}}
        & = & \hcp{\tm{\res{x}{x'}{\gvcpcnf{\conf{C}}{r}}}}
        \\
        \hgv{\tm{\gvcpcnf{\ppar{\,\conf{C}}{\conf{D}}}{r}}}
        & = & \hcp{\tm{\ppar{\gvcpcnf{\conf{C}}{r}}{\gvcpcnf{\conf{D}}{r}}}}
        \\
      \end{array}
&
      \begin{array}[t]{@{}l@{~}c@{~}l@{}}
        \hgv{\tm{\gvcpcnf{\linkconfig{z}{x}{y}}{r}}}
        & = & \hcp{\tm{\ping{z}{\wait{z}{\link{x}{y}}}}}
        \\
        \\
      \end{array}
  \end{array}
  \bigskip
  \\
\begin{array}{@{}l@{\qquad}l@{\qquad}l@{}}
    \begin{array}[t]{@{}l@{~}c@{~}l@{}}
            \hgv{\tm{\gvcpval{x}{r}}}
            & = & \hcp{\tm{{\link{r}{x}}}}
            \\
            \hgv{\tm{\gvcpval{\lambda x.M}{r}}}
            & = & \hcp{\tm{{\recv{r}{x}{\gvcpcom{M}{r}}}}}
            \\
    \end{array}
    &
    \begin{array}[t]{@{}l@{~}c@{~}l@{}}
            \hgv{\tm{\gvcpval{\unit}{r}}}
            & = & \hcp{\tm{{\close{r}{\halt}}}}
            \\
            \hgv{\tm{\gvcpval{\pair{V}{W}}{r}}}
            & = & \hcp{\tm{\send{r}{x}{(\ppar{\gvcpval{V}{x}}{\gvcpval{W}{r}})}}}
            \\
    \end{array}
    &
    \begin{array}[t]{@{}l@{~}c@{~}l@{}}
            \hgv{\tm{\gvcpval{\inl{V}}{r}}}
            & = & \hcp{\tm{\inl{r}{\gvcpval{V}{r}}}}
            \\
            \hgv{\tm{\gvcpval{\inr{V}}{r}}}
            & = & \hcp{\tm{\inr{r}{\gvcpval{V}{r}}}}
    \end{array}
  \end{array}
  \bigskip
  \\
\begin{array}[t]{@{}l@{~}c@{~}l@{}}
      \hgv{\tm{\gvcpcom{V\;W}{r}}}
      & = & \hcp{\tm{\res{x}{x'}{\res{y}{y'}{(\ppar
            {\usend{y}{x}{{\link{r}{y}}}}
            {\ppar
            {\gvcpval{V}{y'}}
            {\gvcpval{W}{x'}}
            }
            )}}}}
      \\
      \hgv{\tm{\gvcpcom{\letunit{V}{M}}{r}}}
      & = & \hcp{\tm{\res{x}{x'}{(\ppar
            {{\wait{x}{\gvcpcom{M}{r}}}}
            {\gvcpval{V}{x'}}
            )}}}
      \\
      \hgv{\tm{\gvcpcom{\letpair{x}{y}{V}{M}}{r}}}
      & = & \hcp{\tm{\res{y}{y'}{(\ppar
            {{\recv{y}{x}{\gvcpcom{M}{r}}}}
            {\gvcpval{V}{y'}}
            )}}}
      \\
      \hgv{\tm{\gvcpcom{\casesum{V}{x}{M}{y}{N}}{r}}}
      & = & \hcp{\tm{\res{x}{x'}{(\ppar
            {{\offer{x}{\gvcpcom{M}{r}}{\gvcpcom{\subst{N}{x}{y}}{r}}}}
            {\gvcpval{V}{x'}}
            )}}}
      \\
      \hgv{\tm{\gvcpcom{\absurd{V}}{r}}}
      & = & \hcp{\tm{\res{x}{x'}{(\ppar
            {{\absurd{x}}}
            {\gvcpval{V}{x'}}
            )}}}
      \\
      \hgv{\tm{\gvcpcom{\letbind{x}{M}{N}}{r}}}
      & = & \hcp{\tm{\res{x}{x'}{(\ppar
            {\pong{x}{\gvcpcom{N}{r}}}
            {\gvcpcom{M}{x'}}
            )}}}
      \\
      \hgv{\tm{\gvcpcom{\ret{V}}{r}}}
      & = & \hcp{\tm{\ping{r}{\gvcpval{V}{r}}}}
      \\
  \end{array}
  \bigskip
  \\
\begin{array}{@{}l@{~}c@{~}l@{}}
    \hgv{\tm{\gvcpval{\link}{r}}}
    & = & \hcp{\tm{\recv{r}{y}{\recv{y}{x}{\ping{r}{\wait{r}{\link{x}{y}}}}}}}
    \\
    \hgv{\tm{\gvcpval{\fork}{r}}}
    & = & \hcp{\tm{\res{y}{y'}{(\recv{r}{x}{\usend{y}{x}{\ping{r}{\link{r}{y}}}}
                                \parallel
                                \recv{y'}{x}{\usend{x}{y'}{\pong{x}{\close{x}{\halt}}}}})}}
\\
    \hgv{\tm{\gvcpval{\send}{r}}}
     & = & \hcp{\tm{\recv{r}{y}{\recv{y}{x}{\usend{y}{x}{\ping{r}{\link{r}{y}}}}}}}
    \\
    \hgv{\tm{\gvcpval{\recv}{r}}}
     & = & \hcp{\tm{\recv{r}{x}{\recv{x}{y}{\ping{r}{\usend{r}{y}{\link{r}{x}}}}}}}
    \\
    \hgv{\tm{\gvcpval{\wait}{r}}}
    & = & \hcp{\tm{\recv{r}{x}{\wait{x}{\ping{r}{\close{r}{\halt}}}}}}
    \\
  \end{array}
\end{array}
\]

  \caption{Translation from \fgHGV to HCP.}
  \label{fig:fghgv-to-hcp}
\end{figure}

There are two changes with respect to the translation of our earlier
paper~\cite{FKDLM21}.
First, in the earlier work the translation of the main thread output
directly to the external output channel instead of forwarding via an
intermediary as in the current translation. This change is purely
aesthetic.
Second, in the earlier work the translation of $\tm{\hgv\fork}$ was
not sufficiently concurrent. Correspondingly there was an error in the
case of the operational correspondence proof which is fixed in the
current paper.

\begin{restatable}[Type Preservation]{lem}{lemhgvtohcptyping}
  \label{lem:hgv-to-hcp-typing}
  ~
  \begin{enumerate}
    \item If $\hgv{\tseq{\ty \Gamma} V T}$, then $\hcp{\seq{\gvcpval V r}{\ty{\gvcpdown\Gamma, \tmty r {\co{\gvcpdown T}}}}}$.
    \item If $\hgv{\tseq{\ty \Gamma} M T}$, then $\hcp{\seq{\gvcpcom M r}{\ty{\gvcpdown\Gamma, \tmty r {\tytens \tyone {\co{\gvcpdown T}}}}}}$.
    \item If $\hgv{\cseq{\ty {\hyper G \hypersep \Gamma}} {\conf C} T}$, where $\Gamma$ is the type environment for the main thread in $\conf C$, then $\hcp{\seq{\gvcpcnf{\conf C} r}{\ty{\gvcpdown{\hyper G} \hypersep \gvcpdown\Gamma, \tmty r {\co{\gvcpdown T}}}}}$.
  \end{enumerate}
\end{restatable}

\begin{restatable}[Substitution]{lem}{corhgvtohcpsubstitution}
  \label{cor:hgv-to-hcp-substitution}
  If $\tm{M}$ is a well-typed term with $\tm{w}\in\fv(\tm{M})$, and $\tm{V}$ is a well-typed value, then $\hcp{\tm{\res{w}{w'}{(\ppar{\gvcpcom{M}{r}}{\gvcpval{V}{w'}})}}\bis_\ta\tm{\gvcpcom{\subst{M}{V}{w}}{r}}}$.
\end{restatable}

\begin{restatable}[Operational Correspondence]{thm}{thmhgvtohcpocconfs}
  \label{thm:hgv-to-hcp-oc-confs}
  Suppose $\hgv{\tm{\conf{C}}}$ is a well-typed configuration.
  \begin{enumerate}
    \item (Preservation of reductions)
If $\hgv{\tm{\conf{C}\cred\tm{\conf{C}'}}}$, then there exists a $\tm{P}$ such that
          $\hcp{\tm{\gvcpcnf{\conf{C}}{r}}\mathbin{\slto[\alpha]{\beta^+}} \tm{P}}$ and
          $\hcp{\tm{P} \bis_{\ta} \tm{\gvcpcnf{\conf{C}'}{r}}}$; and
\item (Reflection of transitions)\hfill
          \begin{itemize}
            \item
                  if $\hcp{\tm{\gvcpcnf{\conf{C}}{r}}\ato\tm{P}}$, then
                  $\hcp{\tm{P}\bis_{\ta}\tm{\gvcpcnf{\conf{C}}{r}}}$; and
\item
                  if $\hcp{\tm{\gvcpcnf{\conf{C}}{r}}\bto\tm{P}}$, then
                  there exists a $\hgv{\tm{\conf{C}'}}$ and a $\hcp{\tm{P'}}$
                  such that $\hgv{\tm{\conf{C}}\cred\tm{\conf{C}'}}$ and
                  $\hcp{\tm{P} \mathbin{\slto[\alpha]{\beta^*}} \tm{P'}}$ and
                  $\hcp{\tm{P'}\bis_{\ta}\tm{\gvcpcnf{\conf{C}'}{r}}}$.
                  Furthermore, $\hgv{\tm{\conf{C}'}}$ is unique up to structural congruence.
          \end{itemize}
  \end{enumerate}
\end{restatable}
The proof is in Appendix~\ref{appendix:hcp-to-hgv}.
One might strive for a tighter operational correspondence here, but
our current translation generates multiple administrative
$\beta$-transitions.
The only term reduction that translates to multiple
$\beta$-transitions is the one for let-bindings. This is because we
choose to encode synchronisation using two $\beta$-transitions. We
could adjust the accounting here by treating synchronisation as a
single $\beta$-transition or its own special kind of administrative
transition.
Many more administrative reductions arise from the configuration
translation. These are due to a combination of synchronisations and
also the fact that we use constants along with pairs and application
for our communication primitives instead of building-in fully-applied
communication primitives.

\paragraph{Translating HCP to HGV}\label{sec:hcp-to-hgv}
We cannot translate HCP processes to HGV terms directly: HGV's term
language only supports $\hgv{\tm{\fork}}$
(see~\Cref{sec:hgv-plus} for further discussion),
so there is no way to translate an individual
name restriction or parallel composition. However, we can still translate HCP
into HGV via the composition of known translations.

\begin{description}
  \item[HCP into CP]
    We must first reunite each
    parallel composition with its corresponding name restriction, \ie,
    translate to CP using the \emph{disentanglement}
    translation shown by~Kokke~\emph{et al.}~\cite[Lemma 4.7]{kokkemp19:tlla}.
    The result is a collection of independent CP processes.
  \item[CP into GV]
    Next, we can translate each CP process into a GV configuration using (a
    variant of) Lindley and Morris' translation~\cite[Figure 8]{lindleym15:semantics}.
  \item[GV into HGV]
    Finally, we can use our embedding of GV into HGV~(\Cref{thm:gv-in-hgv})
    to obtain a collection of well-typed HGV configurations, which can be
    composed using \textsc{TC-Par} to result in a single well-typed HGV
    configuration.
\end{description}

The translation from HCP into CP and the embedding of GV into HGV preserve and
reflect reduction. However, as previously mentioned, Lindley and Morris's original translation from CP to
GV preserves but does \emph{not} reflect reduction due to an asynchronous encoding of
choice.
By adapting their translation to use a synchronous encoding of choice
(\Cref{sec:hgv-choice}), we obtain a translation from CP to GV that both
preserves and reflects reduction.
Thus, composing all three translations
together we obtain a translation from HCP to HGV that preserves and reflects
reduction.

 }  
{\section{Extensions}\label{sec:extensions}
\usingnamespace{hgv}

In this section, we outline three extensions to HGV that exploit
generalising the tree structure of processes to a forest structure.
These extensions are of particular interest since
HGV already supports a core aspect of forest structure, enabling its
full utilisation merely through the addition of a structural rule.
In contrast, to extend GV with forest structure one must distinguish
two distinct introduction rules for parallel
composition~\cite{lindleym15:semantics, fowler19:thesis}.
Other extensions to GV such as shared
channels~\cite{lindleym15:semantics}, polymorphism~\cite{LindleyM17},
and recursive session types~\cite{lindleym16:bananas} adapt to HGV
almost unchanged.

\paragraph{From trees to forests}
The \LabTirName{TC-Par} rule allows two processes to be composed in parallel
if they are typeable under separate hyper-environments. In a closed program,
hyper-environment separators are introduced by \LabTirName{TC-Res}, meaning that
each process must be connected by a channel.

The following \textsc{TC-Mix} rule allows two type environments $\ty{\Gamma_1},
\ty{\Gamma_2}$ to be
split by a hyper-environment separator \emph{without} a channel connecting them,
and is inspired by Girard's \mix rule~\cite{girard87:ll}; in the concurrent
setting, \mix can be interpreted as concurrency \emph{without} communication~
\cite{lindleym15:semantics,atkeylm16}. \mixrule admits a much simpler treatment
of $\tm{\link}$ and provides a crucial ingredient for handling exceptional behaviour.

{
  \small
  \begin{mathpar}
    \inferrule[TC-Mix]{
      \cseq{\ty{\hyperg} \hypersep \ty{\Gamma_1} \hypersep \ty{\Gamma_2}}{\config{C}}{R}
    }{\cseq{\ty{\hyperg} \hypersep \ty{\Gamma_1} ,         \ty{\Gamma_2}}{\config{C}}{R}}
  \end{mathpar}
}

Atkey \emph{et al.}~\cite{atkeylm16} show that conflating the $\hcp{\ty{\tyone}}$
and $\hcp{\ty{\tybot}}$
types in CP (which correspond respectively to the $\ty{\tyends}$ and
$\ty{\tyendr}$ types
in GV) is logically equivalent to adding the \textsc{Mix} rule and a
\textsc{0-Mix} rule (used to type an empty process).
It follows that in the presence of \textsc{TC-Mix}, we use self-dual
$\ty{\tyend}$ type; in the GV setting, by using a self-dual $\ty{\tyend}$ type, we
decouple closing a channel from process termination. We therefore refine the
\LabTirName{TC-Child} rule and the type schema for $\tm{\fork}$ to ensure that
each child thread returns the unit value, and replace the $\tm{\wait}$ constant with
a $\tm{\close}$ constant which eliminates an endpoint of type $\ty{\tyend}$.

{\small
\begin{mathpar}
    \tm{\fork} : \ty{\tylolli{(\tylolli{S}{\tyunit})}{\co{S}}}

    \tm{\close} : \ty{\tylolli{\tyend}{\tyunit}}

  \inferrule
  [TC-Child]
  { \tseq{\ty{\Gamma}}{M}{\ty{\tyunit}} }
  { \cseq{\ty{\Gamma}}{\child M }{\ty{\tychild}} }

  \inferrule
  [E-Close]
  {}
  { \tm{\res{x}{y}{(E[\close\;x] \parallel E'[\close\;y])}} \cred \tm{E[()]
  \parallel E'[()]} }
\end{mathpar}
}

Given \mixrule, we might expect a term-level construct $\tm{\spawn} :
\ty{\tylolli{(\tylolli{\tyunit}{\tyunit})}{\tyunit}}$ which spawns a parallel thread without a
connecting channel. We can encode such a construct using $\tm{\fork}$ and
$\tm{\close}$
(assuming fresh $\tm{x}$ and $\tm{y}$):
\[
    \tm{\spawn\;M} \defeq
{
\tm{\letbindtwo{x}{\fork (\lambda y . \close\;y;\; M)}} \; \tm{\close\;x}
}
\]
Assuming the encoded $\tm{\spawn}$ is running in a main thread, after two
reduction steps, we are left with the configuration:

{\small
\begin{mathpar}
    \inferrule*
    [right=TC-Mix]
    {
        \inferrule*
        [right=TC-Par]
        {
            \inferrule*
            [right=TC-Child]
            {
                \tseq{\cdot}{M}{\tyunit}
            }
            { \cseq{\cdot}{\child M}{\tychild} }
            \\
            \inferrule*
            [right=TC-Main]
            { \tseq{\cdot}{M}{\tyunit} }
            { \cseq{\cdot}{\main ()}{\tyunit} }
        }
        { \cseq{\cdot \hypersep \cdot}{\child M \parallel \main
        ()}{\tyunit} }
    }
    { \cseq{\cdot}{\child M \parallel \main ()}{\tyunit} }
\end{mathpar}
}
Note the essential use of \mixrule to insert a hyper-environment separator.

The addition of \LabTirName{TC-Mix} does not affect preservation or progress.
The result follows from routine adaptations of the proof
of~\Cref{thm:hgv-pres} and~\Cref{thm:hgv-global-progress}.

By relaxing the tree process structure restriction using \LabTirName{TC-Mix},
we can obtain a more efficient treatment of $\tm{\link}$, and can support the
treatment of exceptions advocated by Fowler~\emph{et al.}~\cite{fowlerlmd19:stwt}.

\paragraph{A simpler link}
The $\tm{\link\;(x, y)}$ construct forwards messages from $\tm{x}$ to $\tm{y}$
and vice-versa.
Consider threads $\tm{L} = \tm{\plug{F}{\tm{\link}\;(x, y)}}$, $\tm{M}$,
$\tm{N}$, where $\tm{L}$ connects to $\tm{M}$ by $\tm{x}$ and to $\tm{N}$ by
$\tm{y}$.
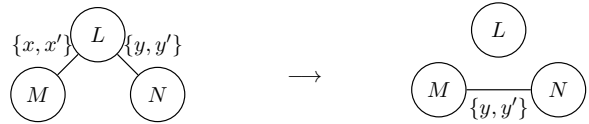
\begin{center}
\scalebox{0.75}{
\begin{minipage}{0.425\textwidth}
    \begin{center}
    \begin{tikzpicture}[>=stealth',auto,node distance=1.5cm,
                    semithick]
  \tikzstyle{every state}=[draw=black,text=black]

  \node[state]         (A)                    {$L$};
  \node[state]         (B) [below left of=A]  {$M$};
  \node[state]         (C) [below right of=A] {$N$};

  \path (A) edge              node[above, xshift=-0.45cm] {$\{x, x'\}$} (B)
    (A) edge              node[above, xshift=0.45cm] {$\{y, y'\}$} (C);
\end{tikzpicture}

 \end{center}
\end{minipage}
\hfill
\begin{minipage}{0.025\textwidth}
    \[
        \longrightarrow
    \]
\end{minipage}
\begin{minipage}{0.425\textwidth}
    \begin{center}
    \begin{tikzpicture}[>=stealth',auto,node distance=1.5cm,
                    semithick]
  \tikzstyle{every state}=[draw=black,text=black]

  \node[state]         (A)                    {$L$};
  \node[state]         (B) [below left of=A]  {$M$};
  \node[state]         (C) [below right of=A] {$N$};

  \path (B) edge              node[below] {$\{y, y'\}$} (C);
\end{tikzpicture}
 \end{center}
\end{minipage}
}
\end{center}
The result of link reduction has forest structure.
Well-typed closed programs in both GV and HGV must \emph{always}
maintain tree structure. Different versions of GV do so in various
unsatisfactory ways: one is pre-emptive
blocking~\cite{lindleym15:semantics}, which breaks confluence; another
is two-stage linking (\Cref{fig:hgv-reduction}), which defers
forwarding via a special link thread~\cite{lindleym16:bananas}.

{\paragraph{Pre-emptive blocking}}

Lindley and Morris~\cite{lindleym15:semantics} implement $\tm\link$ using the
following rule (modified here to use a double-binder formulation):
\[
    { \tm{\res{x}{x'}{(F[\link\;(x, y)] \parallel \config{F}'[M])}}
    \cred
    \tm{\res{x}{x'}{(F[x] \parallel \config{F}'[\wait\;x'; M\{ y /
    x'\}])}}
  }
  \quad
  { \text{where } \tm{x'} \in \fv(\tm{M}) }
\]
The first thread will eventually reduce to $\tm{\child x}$, at which point the second
thread will synchronise to eliminate $\tm{x}$ and $\tm{x'}$ and then evaluate
the continuation $\tm{M}$ with endpoint $\tm{y}$ substituted for $\tm{x'}$.
Unfortunately, this formulation of $\tm{\link}$ preemptively inhibits reduction in
the second thread, since the evaluation rule inserts a blocking $\tm{\wait}$.  The
resulting system does not satisfy the diamond property.

{\paragraph{Link threads}}

HGV uses the incarnation of $\tm\link$ advocated by Lindley and
Morris~\cite{lindleym16:bananas}, where linking is split into two
stages: the first generates a fresh pair of endpoints $\tm{z},
\tm{z'}$ and a link thread of the form $\tm{\linkconfig{z'}{x}{y}}$,
and returns $\tm{z}$ to the calling thread. Once the calling thread
has evaluated to a value (which must by typing be $\tm{z}$), then the
link substitution can take place.
This formulation recovers confluence, but we still lose a degree of concurrency:
communication on $\tm{y}$ is blocked until the linking thread has fully evaluated. In
an ideal implementation, the behaviour of the linking thread would be irrelevant
to the remainder of the configuration.  The operation requires additional
runtime syntax and thus complicates the metatheory.

{\paragraph{With \LabTirName{TC-Mix}}}
The above issues are symptomatic of the fact that the process structure after a
link takes place is a forest rather than a tree. However, with
\LabTirName{TC-Mix}, we can refine the type schema for $\tm\link$ to
$\ty{\tylolli{(\typrod{S}{\co{S}})}{\tyunit}}$ and we can use the following rule:
\[
    \tm{\res{x}{x'}{(F[\link\;(x, y)] \parallel \phi N)}}
    \cred
    \tm{F[()] \parallel \phi N \{ y / x' \}}
\]
This formulation enables immediate substitution, maximimising
concurrency.
A variant of HGV replacing
\textsc{E-Reify-Link} and \textsc{E-Comm-Link} with \textsc{E-Link-Mix} retains
HGV's metatheory.

\paragraph{Exceptions}
In order to support exceptions in the presence of linear
endpoints~\cite{fowlerlmd19:stwt, mostrousv18:affine} we must have a
way of \emph{cancelling} an endpoint.
Mostrous and Vasconcelos~\cite{mostrousv18:affine} describe a process calculus allowing the
\emph{explicit cancellation} of a channel endpoint, accounting for exceptional
scenarios such as a client disconnecting, or a thread encountering an
unrecoverable error. Attempting to communicate with a cancelled endpoint
raises an exception. Fowler \emph{et al.}~\cite{fowlerlmd19:stwt} extend these
ideas to the functional setting, introducing Exceptional GV (EGV).
EGV supports exceptional behaviour by adding three term-level constructs:

\begin{itemize}
    \item a new constant, $\tm{\cancel} : \ty{\tylolli{S}{\tyunit}}$, which
        allows us to discard an arbitrary session endpoint with type $\ty{S}$
    \item a
        construct $\tm{\raiseexn}$, which raises an exception
    \item an exception handling construct $\tm{\tryasinotherwise{L}{x}{M}{N}}$
        in the style of Benton \& Kennedy~\cite{bentonk01}, which attempts
        possibly-failing computation $\tm{L}$, binding the result to $\tm{x}$ in success
        continuation $\tm{M}$ if successful and evaluating $\tm{N}$ if an exception is raised
\end{itemize}

Cancellation generates a \emph{zapper thread} ($\zap{x}$) which severs a tree topology
into a forest as in the following example.

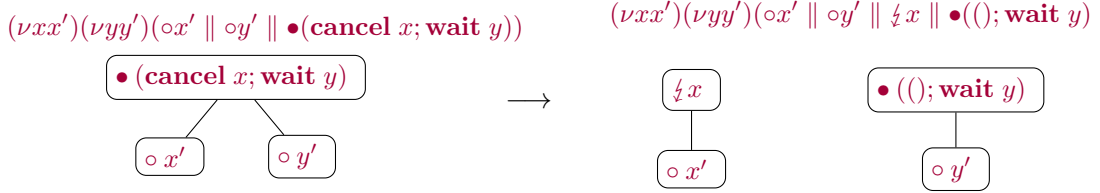
\begin{minipage}{0.4\textwidth}
\small
\centering
\[
  \tmcolor
  (\nu x x')(\nu y y')(\child x' \parallel \child y' \parallel \main
  (\cancel\;{x}; \wait\;{y}))
\]
\begin{tikzpicture}[sibling distance=5em, level distance=3em,
    every node/.style = {shape=rectangle, rounded corners,
    draw, align=center}]
    \node { $\tm{\main\:(\andthen{\cancel\;x}{\wait\;y})}$ }
      child { node { $\tm{\child\:x'}$ } }
      child { node { $\tm{\child\:y'}$ } };
\end{tikzpicture}
\end{minipage}
\begin{minipage}{0.1\textwidth}
    \[
        \cred
    \]
\end{minipage}
\begin{minipage}{0.45\textwidth}
\small
\[
  \tmcolor
  (\nu x x')(\nu y y')(\child x' \parallel \child y' \parallel
  \zap{x} \parallel \main ((); \wait\;{y})
\]
\centering
\begin{tikzpicture}[sibling distance=10em, node distance = 10em,
    level distance=3em,
    every node/.style = {shape=rectangle, rounded corners,
    draw, align=center}]
    \node (a) { $\tm{\zap{x}}$ }
      child { node { $\tm{\child\:x'}$ } };

    \node[right of=a] (b) { $\tm{\main\:(\andthen{()}{\wait\;y})}$ }
      child { node { $\tm{\child\:y'}$ } };
\end{tikzpicture}
\end{minipage}

\bigskip
\noindent
The configuration on the left has a tree process structure. However, after
reduction, we obtain the configuration on the right which is clearly a forest
and thus needs \LabTirName{TC-Mix} to be typeable.
We have described a \emph{synchronous} version of EGV, but extending our
treatment to asynchrony as in the work of~\cite{fowlerlmd19:stwt} is a routine
adaptation.

 }      
{\section{Can we separate fork?}\label{sec:hgv-plus}
\usingnamespace{hgv}

Hyper-environments allow us to cleanly separate name restriction and parallel
composition in process configurations. A natural follow-on question is whether
we could use the same technique at the level of \emph{terms} in order to
split $\hgv{\tm{\calcwd{fork}}}$ into separate constructs for creating a channel
and spawning a process. As tantalising a prospect this is, we argue that
the disadvantages outweigh the benefits.

Suppose we were to extend term typing to allow hyper-environments,
$\tseq{\ty{\hyper{G}}}{M}{T}$,\linebreak[4]and were to introduce terms
$\tm{\letnew{x}{x'}{M}}$ to create a channel and\linebreak[4]
$\tm{\letspawn{M}{N}}$ to spawn a thread, with the following typing rules:
\begin{mathpar}
  \inferrule[TM-LetNew]{
    \tseq{\ty{\hyper{G}}\hypersep\ty{\Gamma_1},\tmty{x}{S}\hypersep\ty{\Gamma_2},\tmty{x'}{\co{S}}}{M}{T}
  }{\tseq{\ty{\hyper{G}}\hypersep\ty{\Gamma_1},\ty{\Gamma_2}}{{\letnew{x}{x'}{M}}}{T}}

  \inferrule[TM-LetSpawn]{
    \tseq{\ty{\hyper{G}}}{M}{\tyends}
    \\
    \tseq{\ty{\hyper{H}}}{N}{T}
  }{\tseq{\ty{\hyper{G}}\hypersep\ty{\hyper{H}}}{{\letspawn{M}{N}}}{T}}
\end{mathpar}
These rather ad-hoc rules mirror hypersequent cut and hypersequent composition:\linebreak[4]\rulename{TM-LetNew} creates a new channel with endpoints $\tm{x}$ and $\tm{x'}$, and requires them to be used in separate threads in the continuation $\tm{M}$; and \rulename{TM-LetSpawn} takes a term $\tm{M}$, spawns it as a child thread, and continues as $\tm{N}$. Using these rules, we can encode $\tm{\fork\;M}$ as $\tm{\letnew{x}{x'}{\letspawn{(M\;x)}{x'}}}$.

Where else can we allow hyper-environments? In HCP, we have two options: (1) if
we restrict \emph{all logical rules} to singleton hypersequents and allow
hyper-environments only in the rules for name restriction and parallel
composition, we can use standard sequential
semantics~\cite{montesip18:ct,kokkemp19:tlla}; but (2) if we allow
hyper-environments in \emph{any logical rules}, we must use a semantics which
allows the corresponding actions to be delayed~\cite{kokkemp19:popl}.
This is unlikely to be a property of logical rules, but rather due to the fact
that the logical rules correspond exactly to the communication actions---which
block reduction---and the structural rules to name restriction and parallel
composition---which do not. Therefore, we expect the positions where
hypersequents can safely occur to follow from the structure of evaluation
contexts and whether any blocking term perform a communication action.

Regardless of our choice, we would be left with restrictions on the syntax of
terms that seem sensible in a process calculus, but are surprising in a $\lambda$-calculus. In the strictest variant, where we disallow hyper-environments in all but the above two rules, uses of \rulename{TM-LetNew} and \rulename{TM-LetSpawn} may be interleaved, but no other construct may appear between a \rulename{TM-LetNew} and its corresponding \rulename{TM-LetSpawn}. Consider the following terms, where $\tm{M}$ uses $\tm{x}$ and $\tm{y}$, and $\tm{N}$ uses $\tm{x'}$. Term (\ref{eq:yep}) may be well-typed, but (\ref{eq:nope}) is always ill-typed:
\begin{align}
  \label{eq:yep}
  \tm{\letbind{y}{1}{\letnew{x}{x'}{\letspawn{M}{N}}}}
  \\
  \label{eq:nope}
  \tm{\letnew{x}{x'}{\letbind{y}{1}{\letspawn{M}{N}}}}
\end{align}

Note that $\tm{\letnew{x}{x'}{M}}$ is a single, monolithic term
constructor---exactly what hypersequents were meant to prevent! However, if we
attempt to decompose these constructors, we find that these are not the regular
product and unit types.

 }        
{\section{Related work}\label{sec:related-work}

\paragraph{Session Types and Functional Languages}

Session types were originally introduced in the context of process
calculi~\cite{honda93,takeuchihk94,hondavk98}, however they have been vastly integrated also in functional calculi, a line of work initiated by Gay and collaborators~\cite{VascoRG04, VascoRG06, gayv10:last}.
This family of calculi builds session types directly into a lambda
calculus.
Toninho~\emph{et al.}~\cite{ToninhoCP13} take an alternative approach,
stratifying their system into a session-typed process calculus and a
separate functional calculus.
There are many pragmatic embeddings of session type systems in
existing functional programming
languages~\cite{neubauerthiemann04,pucellatov08,sackmaneisenbach08,imaiyuen10,orchardyoshida16,KD21}. A
detailed survey is given by Orchard and Yoshida
\cite{orchardyoshida17}.

\paragraph{Propositions as Sessions}
When Girard introduced linear logic~\cite{girard87:ll} he suggested a
connection with concurrency.
Abramsky~\cite{abramsky94} and Bellin and Scott~\cite{bellins94} give
embeddings of linear logic proofs in $\pi$-calculus, where cut
reduction is simulated by $\pi$-calculus reduction.
Both embeddings interpret tensor as parallel composition. The
correspondence with $\pi$-calculus is not tight in that these systems
allow independent prefixes to be reordered.
Caires and Pfenning~\cite{cairesp10:pidill} give a propositions as
types correspondence between dual intuitionistic linear logic and a
session-typed $\pi$-calculus called $\pi\text{DILL}$.
They interpret tensor as output. The correspondence with
$\pi$-calculus is tight in that independent prefixes may not be
reordered.
With CP~\cite{wadler14:sessions}, Wadler adapts $\pi\text{DILL}$ to
classical linear logic.
Aschieri and Genco~\cite{AschieriG20} give an interpretation of
classical multiplicative linear logic as concurrent functional
programs. They interpret $\parr$ as parallel composition, and the
connection to session types is less direct.

\paragraph{Priority-based Calculi}
Systems such as $\pi\text{DILL}$, CP, and GV (and indeed HCP and HGV)
ensure deadlock freedom by exploiting the type system to statically
impose a tree structure on the communication topology --- there can be
at most one communication channel between any two processes.
Another line of work explores a more liberal approach to deadlock
freedom enabling some cyclic communication topologies, where deadlock
freedom is guaranteed via \emph{priorities}, which impose an order on
actions. Priorites were introduced by Kobayashi and Padovani
\cite{kobayashi06,padovani14} and adopted by Dardha and Gay
\cite{dardhag18:pcp} in Priority CP (PCP), and Kokke and Dardha in
Priority GV (PGV) \cite{kokked21:pgv}.
Dezani~\etal~\cite{dezani-ciancagliniliguoro09progress} and Vieira and Vasconcelos~\cite{vieiravasconcelos13} use a partial order on channels to guarantee deadlock freedom, following Kobayashi's work \cite{kobayashi06}. Later on Dezani~\etal~\cite{dezani-ciancaglinimostrous06} guarantee progress by allowing only one active session at a time.
Carbone~\etal~\cite{CarboneDM14} use catalysers to show that progress is a compositional form of lock freedom for standard typed $\pi$ calculus. The authors describe how this technique can be used for session typed $\pi$-calculus by using the the encoding of session types to linear types~\cite{dardhags17:revisited,dardha14beat,Dardha16}. Dardha and Perez \cite{DardhaP22} compare the different calculi and techniques for deadlock freedom using CP and CLL as a yardstick and showing that the class of processes in CP is strictly included in the class of processes typed by Kobayashi \cite{kobayashi06}.

\paragraph{Graph-theoretic Approaches}
Carbone and Debois~\cite{carbonedebois10} define a graph-theoretic approach for
a session typed $\pi$-calculus. They define an explicit dependency graph defined
inductively on the structure of a process, in contrast to our approach of
inducing a graph on type environments given a co-name set. They ensure
progress for processes with acyclic graphs using a \emph{catalyser}, which provides a
missing counterpart to a process.
Jacobs~\etal~\cite{jacobsbk22} also define a graph-theoretic approach to
deadlock freedom, but differently from Carbone and Debois, their work is based
on separation logic.  A line of work on many-writer, single-reader process
calculi~\cite{padovani18,deLiguorop18} uses explicit dependency graphs to both
ensure resource separation and guarantee deadlock freedom, however it is not immediate how to apply this approach to functional calculi.

 }    
{\section{Conclusion and future work}\label{sec:conclusion}

HGV exploits hypersequents to resolve fundamental modularity issues
with GV.
As a consequence, we have obtained a tight operational correspondence
between HGV and HCP.
HGV is a modular and extensible core calculus for functional
programming with \emph{binary} session types.
In future we intend to apply hypersequents to \emph{multiparty} versions of
CP~\cite{CarboneLMSW16} and GV~\cite{JacobsBK22a} to exhibit a similarly strong
operational correspondence.
 }

\section*{Acknowledgements}
We are deeply grateful to Marco Peressotti for discovering an error in
our presentation of the LTS for HCP in the original conference paper,
and for suggesting a fix via the definition of saturated transitions.
We thank the CONCUR and LMCS reviewers for their helpful comments and
suggestions.
This work was supported by EPSRC grants EP/K034413/1 (ABCD),
EP/T014628/1 (STARDUST), EP/L01503X/1 (CDT PPar), ERC Consolidator
Grant 682315 (Skye), UKRI Future Leaders Fellowship MR/T043830/1
(EHOP), a UK Government ISCF Metrology Fellowship grant, EU HORIZON
2020 MSCA RISE project 778233 (BehAPI), and NSF grant CCF-2044815.

\bibliographystyle{alphaurl}
\bibliography{main}

\clearpage
\appendix
{

\etoctocstyle{1}{\huge Appendices}
\etocsettagdepth{mtchapter}{none}
\etocsettagdepth{mtappendix}{subsection}

\tableofcontents
%\clearpage

\etocdepthtag.toc{mtappendix}
{\section{Omitted Proofs for \Cref{sec:hgv}: Hypersequent GV}\label{sec:appendix:hgv}

In this Appendix we give full definitions and proofs
for~\Cref{sec:hgv}.

\begingroup
\usingnamespace{hgv}

\subsection{Derived typing rules for syntactic sugar}

\begin{figure}[H]
\begin{mathpar}
  \inferrule*[lab=T-Seq]{
    \tseq{\ty{\Gamma}}{M}{\tyunit}
    \\
    \tseq{\ty{\Delta}}{N}{T}
  }{\tseq{\ty{\Gamma},\ty{\Delta}}{\andthen{M}{N}}{T}}

  \inferrule*[lab=T-LamUnit]{
    {\tseq{\ty{\Gamma}}{M}{T}}
  }{\tseq
    {\ty{\Gamma}}
    {\lambda\unit.M}
    {\tylolli{\tyunit}{T}}}

  \inferrule*[lab=T-LamPair]
  {\tseq
    {\ty{\Gamma},\tmty{x}{T},\tmty{y}{T'}}
    {M}
    {U}}
  {\tseq
    {\ty{\Gamma}}
    {\lambda\pair{x}{y}.M}
    {\tylolli{\typrod{T}{T'}}{U}}}

  \inferrule*[lab=T-Let]{
    \tseq{\ty{\Gamma}}{M}{T}
    \\
    \tseq{\ty{\Delta},\tmty{x}{T}}{N}{U}
  }{\tseq{\ty{\Gamma},\ty{\Delta}}{\letbind{x}{M}{N}}{U}}

  \inferrule*[lab=T-Select-Inl]{
  }{\tseq{\emptyenv}{\select{\labinl}}{\tylolli{\tyselect{S}{S'}}{S}}}

  \inferrule*[lab=T-Select-Inr]{
  }{\tseq{\emptyenv}{\select{\labinr}}{\tylolli{\tyselect{S}{S'}}{S'}}}

  \inferrule*[lab=T-Offer]{
    {\tseq
      {\ty{\Gamma}}
      {L}
      {\tyoffer{S}{S'}}}
    \\
    {\tseq
      {\ty{\Delta},\tmty{x}{S}}
      {M}
      {T}}
    \\
    {\tseq
      {\ty{\Delta},\tmty{y}{S'}}
      {N}
      {T}}
  }{\tseq
    {\ty{\Gamma},\ty{\Delta}}
    {\offer{L}{x}{M}{y}{N}}
    {T}}

  \inferrule*[lab=T-Offer-Absurd]{
    \tseq
    {\ty{\Gamma}}
    {L}
    {\tyofferemp}
  }{\tseq
    {\ty{\Gamma},\ty{\Delta}}
    {\offeremp{L}}
    {T}}
\end{mathpar}
\caption{Derived rules for syntactic sugar}
\label{fig:derived-sugar-rules}
\end{figure}

The main body makes use of syntactic sugar, and encodings of branching and
selection. Figure~\ref{fig:derived-sugar-rules} shows the derived typing rules.

\subsection{Preservation Proof}

Next, we detail the proof of preservation. We begin with the usual lemmas to
manipulate evaluation contexts, and the usual substitution lemma.

\begin{lem}[Subterm typeability]\label{lem:hgv:subterm-typ}
    Suppose $\deriv{D}$ is a derivation of
    $\tseq{\ty{\Gamma}}{\tm{E[M]}}{\ty{T}}$.
    Then, there exist $\ty{\Gamma_1}, \ty{\Gamma_2}$ such that
    $\ty{\Gamma} = \ty{\Gamma_1}, \ty{\Gamma_2}$,
      a type $\ty{U}$, and some subderivation $\deriv{D}'$ of $\deriv{D}$ concluding
  $\tseq{\ty{\Gamma_2}}{\tm{M}}{\ty{U}}$, where the position of $\deriv{D}'$ in $\deriv{D}$
  coincides with the position of the hole in $\deriv{D}$.
\end{lem}
\begin{proof}
    By induction on the structure of $\tm{E}$.
\end{proof}

\begin{lem}[Replacement, Evaluation Contexts]\label{lem:hgv:replacement}
  If:
  \begin{itemize}
      \item $\deriv{D}$ is a derivation of $\tseq{\ty{\Gamma_1},
          \ty{\Gamma_2}}{\tm{E}[\tm{M}]}{\ty{T}}$
      \item $\deriv{D}'$ is a subderivation of $\textbf{D}$ concluding
          $\tseq{\ty{\Gamma_2}}{\tm{M}}{\ty{U}}$
    \item The position of $\textbf{D}'$ in $\textbf{D}$ corresponds to that of
        the hole in $\tm{E}$
    \item $\tseq{\ty{\Gamma_3}}{\tm{N}}{\ty{U}}$
    \item $\ty{\Gamma_1}, \ty{\Gamma_3}$ is defined
  \end{itemize}

  then $\tseq{\ty{\Gamma_1}, \ty{\Gamma_3}}{\tm{E}[\tm{N}]}{\ty{T}}$.
\end{lem}
\begin{proof}
    By induction on the structure of $\tm{E}$.
\end{proof}

\begin{lem}[Substitution]\label{lem:hgv:substitution}
  If:
  \begin{enumerate}
      \item $\tseq{\ty{\Gamma_1}, \tm{x} : \ty{U}}{\tm{M}}{\ty{T}}$
      \item $\tseq{\ty{\Gamma_2}}{\tm{N}}{\ty{U}}$
      \item $\ty{\Gamma_1}, \ty{\Gamma_2}$ is defined
  \end{enumerate}

  then $\tseq{\ty{\Gamma_1}, \ty{\Gamma_2}}{\tm{M \{ N / x \}}}{\ty{T}}$.
\end{lem}
\begin{proof}
    By induction on the derivation of $\tseq{\ty{\Gamma_1}, \tm{x} : \ty{U}}{M}{T}$.
\end{proof}

Preservation of typing under term reduction is standard.

\begin{lem}[Preservation, $\tred$]\label{lem:gv:term-pres}
    If $\tseq{\ty{\Gamma}}{M}{T}$ and $\tm{M} \tred \tm{N}$, then
    $\tseq{\ty{\Gamma}}{N}{T}$.
\end{lem}
\begin{proof}
    A standard induction on the derivation of $\tred$.
\end{proof}

Runtime type merging is commutative and associative. We make use of these
properties implicitly in the remainder of the proofs.

\begin{lem} \hfill % Pushed the first point to a new line
    \begin{enumerate}
        \item $\ty{R_1} \isect \ty{R_2} \iff \ty{R_2} \isect \ty{R_1}$
        \item $\ty{R_1} \isect (\ty{R_2} \isect \ty{R_3}) \iff (\ty{R_1} \isect \ty{R_2}) \isect \ty{R_3}$
    \end{enumerate}
\end{lem}
\begin{proof}
    Immediate from the definition of $\isect$.
\end{proof}

The first more major result is preservation of configuration typing under
structural congruence.

\begin{lem}[Preservation ($\equiv$)]\label{lem:hgv-equiv-pres}
If $\cseq{\ty{\hyper{G}}}{\conf{C}}{R}$ and $\tm{\conf{C}}\equiv\tm{\conf{D}}$,
          then $\cseq{\ty{\hyper{G}}}{\conf{D}}{R}$.
\end{lem}
\begin{proof}
  We consider the cases for the equivalence axioms; the congruence cases are
  straightforward applications of the IH.

  \begin{case}{\LabTirName{SC-ParAssoc}}
   \[
       \tm{\config{C} \parallel (\config{D} \parallel \config{E})}
       \equiv
       \tm{(\config{C} \parallel \config{D}) \parallel \config{E}}
   \]

      {\small
    \begin{mathpar}
      \begin{array}{r@{\quad}c@{\quad}l}
        \inferrule*
        {
          \cseq{\ty{\hyper{G}_1}}{\config{C}}{R_1} \\
\inferrule*
          { \cseq{\ty{\hyper{G}_2}}{\config{D}}{R_2} \\
            \cseq{\ty{\hyper{G}_3}}{\config{E}}{R_3} }
            { \cseq{\ty{\hyper{G}_2} \hypersep \ty{\hyper{G}_3}}{\config{D} \parallel \config{E}}{R_2 \isect R_3}  }
        }
        {
            \cseq
              {\ty{\hyper{G}_1} \hypersep \ty{\hyper{G}_2} \hypersep \ty{\hyper{G}_3}}
              {\config{C} \parallel (\config{D} \parallel \config{E}) }
              {R_1 \isect R_2 \isect R_3}
        }
&\hspace*{-10pt}\iff\hspace*{-10pt}&
\inferrule*
        {
          \inferrule*
          {
              \cseq{\ty{\hyper{G}_1}}{\config{C}}{R_1} \\
              \cseq{\ty{\hyper{G}_2}}{\config{D}}{R_2}
          }
          { \cseq
              {\ty{\hyper{G}_1} \hypersep \ty{\hyper{G}_2}}
              { \config{C} \parallel \config{D}}
              {R_1 \isect R_2}
          } \\
\cseq{\ty{\hyper{G}_3}}{\config{E}}{R_3}
        }
        {
           \cseq
             {\ty{\hyper{G}_1} \hypersep \ty{\hyper{G}_2} \hypersep \ty{\hyper{G}_3}}
             { (\config{C} \parallel \config{D}) \parallel \config{E} }
             {R_1 \isect R_2 \isect R_3}
         }
      \end{array}
    \end{mathpar}
}
  \end{case}
  \begin{case}{\LabTirName{SC-ParComm}}
    \[
        \tm{\config{C} \parallel \config{D}} \equiv \tm{\config{D} \parallel \config{C}}
    \]
    \begin{mathpar}
      \begin{array}{r@{\quad}c@{\quad}l}
        \inferrule*
        {
            \cseq{\ty{\hyper{G}}}{\config{C}}{R_1} \\
            \cseq{\ty{\hyper{H}}}{\config{D}}{R_2} \\
        }
        { \cseq
            {\ty{\hyper{G}} \hypersep \ty{\hyper{H}}}
            {\config{C} \parallel \config{D}}
            {R_1 \isect R_2}
        }
& \iff &
        \inferrule*
        {
            \cseq{\ty{\hyper{H}}}{\config{D}}{U} \\
            \cseq{\ty{\hyper{G}}}{\config{C}}{T}
        }
        {
            \cseq
            { \ty{\hyper{G}} \hypersep \ty{\hyper{H}} }
            { \config{D} \parallel \config{C}}
            { R_1 \isect R_2 }
        }
      \end{array}
    \end{mathpar}
  \end{case}

  \begin{case}{\LabTirName{SC-NewComm}}
\[
    \tm{(\nu x x')(\nu y y')\config{C}} \equiv \tm{(\nu y y')(\nu x x')\config{C}}
\]
    Two illustrative subcases:

    {\small
    \begin{subcase}{1}
    \begin{mathpar}
      \begin{array}{c}
        \inferrule*
        {
          \inferrule*
          {
            \cseq
                { \ty{\hyper{G}} \hypersep \ty{\Gamma_1}, \tm{x} : \ty{S}
                    \hypersep \ty{\Gamma_2}, \tm{x'} : \ty{\co{S}}
                    \hypersep \ty{\Gamma_3}, \tm{y} : \ty{S'} \hypersep
                    \ty{\Gamma_4}, \tm{y'} : \ty{\co{S'}}
                }
                { \config{C} }
                { R }
          }
          {
            \cseq
            { \ty{\hyper{G}} \hypersep \ty{\Gamma_1}, \tm{x} : \ty{S} \hypersep
                \ty{\Gamma_2}, \tm{x'} : \ty{\co{S}}
                \hypersep \ty{\Gamma_3, \Gamma_4}
              }
              { (\nu y y') \config{C} }
              { R }
          }
        }
        {
          \cseq
          { \ty{\hyper{G}} \hypersep \ty{\Gamma_1, \Gamma_2} \hypersep
              \ty{\Gamma_3, \Gamma_4} }
            { (\nu x x')(\nu y y') \config{C} }
            { R }
        }
\\
        \iff
        \\
        \inferrule*
        {
          \inferrule*
          {
              \cseq
              { \ty{\hyper{G}} \hypersep \ty{\Gamma_1}, \tm{y} : \ty{S'}
                  \hypersep \ty{\Gamma_2}, \tm{y'} : \ty{\co{S'}}
              \hypersep \ty{\Gamma_3}, \tm{x} : \ty{S} \hypersep
              \ty{\Gamma_4}, \tm{x'} : \ty{\co{S}} }
              { \config{C} }
              { R }
          }
          {
              \cseq
              { \ty{\hyper{G}} \hypersep \ty{\Gamma_1}, \tm{y} : \ty{S'} \hypersep
          \ty{\Gamma_2}, \tm{y'} : \ty{\co{S'}} \hypersep \ty{\Gamma_3}, \ty{\Gamma_4} }
                { (\nu x x') \config{C} }
                { R }
          }
        }
        {
            \cseq
            { \ty{\hyper{G}} \hypersep \ty{\Gamma_1}, \ty{\Gamma_2} \hypersep
                \ty{\Gamma_3}, \ty{\Gamma_4} }
                { (\nu y y')(\nu x x') \config{C} }
                { R }
        }
      \end{array}
    \end{mathpar}
    \end{subcase}

    \begin{subcase}{2}
    \begin{mathpar}
      \begin{array}{c}
        \inferrule*
        {
          \inferrule*
          {
              \cseq
              {  \ty{\hyper{G}} \hypersep \ty{\Gamma_1}, \tm{x} : \ty{S}, \tm{y} :
                  \ty{S'} \hypersep \ty{\Gamma_2}, \tm{y'} :
              \ty{\co{S'}} \hypersep \ty{\Gamma_3}, \tm{x'} : \ty{\co{S}} }
                { \config{C} }
                { R }
          }
          { \cseq
              { \ty{\hyper{G}} \hypersep \ty{\Gamma_1}, \ty{\Gamma_2}, \tm{x} :
              \ty{S} \hypersep \ty{\Gamma_3}, \tm{x'} : \ty{\co{S}} }
              { (\nu y y') \config{C} }
              { R }
          }
        }
        {
            \cseq
            {  \ty{\hyper{G}} \hypersep \ty{\Gamma_1}, \ty{\Gamma_2}, \ty{\Gamma_3} }
                {  (\nu x x')(\nu y y') \config{C} }
                { R }
        }
        \\
        \iff
        \\
        \inferrule*
        {
          \inferrule*
          {
              \cseq
              { \ty{\hyper{G}} \hypersep \ty{\Gamma_1}, \tm{x} : \ty{S}, \tm{y} :
                  \ty{S'} \hypersep
                  \ty{\Gamma_2}, \tm{y'} : \ty{\co{S'}} \hypersep \ty{\Gamma_3}, \tm{x'}
          : \ty{\co{S}} }
                { \config{C} }
                { R }
          }
          {
              \cseq
              { \ty{\hyper{G}} \hypersep \ty{\Gamma_1}, \ty{\Gamma_3}, \tm{y} :
              \ty{S'} \hypersep \ty{\Gamma_2}, \tm{y'} : \ty{\co{S'}} }
                { (\nu x x') \config{C} }
                { R }
          }
        }
        {
            \cseq
            { \ty{\hyper{G}} \hypersep \ty{\Gamma_1}, \ty{\Gamma_2}, \ty{\Gamma_3} }
                { (\nu y y')(\nu x x') \config{C} }
                { R }
        }
    \end{array}
  \end{mathpar}
  \end{subcase}
}
  \end{case}

  \begin{case}{\LabTirName{SC-NewSwap}}
      \[
          \tm{(\nu x y)\config{C}} \equiv \tm{(\nu y x)\config{C}}
\]
    Follows immediately since hyper-environments are treated as unordered.
  \end{case}

  \begin{case}{\LabTirName{SC-ScopeExt}}
      \[
          \tm{\config{C} \parallel (\nu x y)\config{D}} \equiv
          \tm{(\nu x y)(\config{C} \parallel \config{D})}
        \]
        (where $\tm{x}, \tm{y} \not\in \fv(\tm{\config{C}})$)

    {\small
    \begin{mathpar}
      \begin{array}{c}
        \inferrule*
        {
          \inferrule*
          {
              \cseq
                {\ty{\hyper{G}}}
                {\config{C}}
                {R_1} \\
\cseq
              { \ty{\hyper{H}} \hypersep \ty{\Gamma_1}, \tm{x} : \ty{S}
              \hypersep \ty{\Gamma_2}, \tm{y} : \ty{\co{S}}}
                {\config{D}}
                {R_2}
          }
          {
              \cseq
              { \ty{\hyper{G}} }
              { \config{C} }
              { R_1 } \\
\cseq
              { \ty{\hyper{H}} \hypersep \ty{\Gamma_1}, \ty{\Gamma_2} }
              { (\nu x y) \config{D} }
              { R_2 }
          }
        }
        {
            \cseq
                { \ty{\hyper{G}} \hypersep \ty{\hyper{H}} \hypersep \ty{\Gamma_1}, \ty{\Gamma_2} }
                { \config{C} \parallel (\nu x y) \config{D} }
                { R_1 \isect R_2 }
        }
        \\
        \iff
        \\
        \inferrule*
        {
          \inferrule*
          {
              \cseq
                  { \ty{\hyper{G}} }
                  { \config{C} }
                  { R_1 }
              \\
              \cseq
              { \ty{\hyper{H}} \hypersep \ty{\Gamma_1}, \tm{x} : \ty{S}
              \hypersep \ty{\Gamma_2}, \tm{y} : \ty{\co{S}} }
                { \config{D} }
                { R_2 }
          }
          {
            \cseq
                { \ty{\hyper{G}} \hypersep \ty{\hyper{H}} \hypersep
                    \ty{\Gamma_1}, \tm{x} : \ty{S} \hypersep \ty{\Gamma_2},
                    \tm{y} : \ty{\co{S}}
                }
                { \config{C} \parallel \config{D} }
                { R_1 \isect R_2 }
          }
        }
        {
            \cseq
                {  \ty{\hyper{G}} \hypersep \ty{\hyper{H}} \hypersep \ty{\Gamma_1}, \ty{\Gamma_2} }
                { (\nu x y) (\config{C} \parallel \config{D} ) }
                { R_1 \isect R_2 }
        }
      \end{array}
    \end{mathpar}
    }
  \end{case}

  \begin{case}{\LabTirName{SC-LinkComm}}
      \[
          \tm{\linkconfig{z}{x}{y}} \equiv \tm{\linkconfig{z}{y}{x}}
      \]

      Assumption:

          \begin{mathpar}
              \inferrule*
              { }
              { \cseq{\tm{x} : \ty{S}, \tm{y} : \ty{\co{S}}}{\linkconfig{z}{x}{y}}{\tychild} }
          \end{mathpar}

          By dualising both variables, we have that $ \tm{x} : \ty{\co{S}},
          \tm{y} : \ty{\co{\co{S}}}$.
          Since duality is an involution, we can show $\tm{x} : \ty{S}, \tm{y} :
          \ty{\co{S}} \iff \tm{x} : \ty{\co{S}}, \tm{y} : \ty{S}$.

          Thus:

          \begin{mathpar}
              \inferrule*
              { }
              { \cseq{\tm{y} : \ty{S}, \tm{x} :
              \ty{\co{S}}}{\tm{\linkconfig{z}{y}{x}}}{\ty{\tychild}} }
          \end{mathpar}

          The reasoning for the symmetric case is identical. \qedhere
  \end{case}
\end{proof}

The next result shows that configuration typeability is preserved under
configuration reduction. Note that this lemma makes crucial use of
Lemma~\ref{lem:hgv-equiv-pres} due to \rulename{E-Equiv}.

\begin{lem}[Preservation ($\cred$)]\label{lem:hgv-cred-pres}
    If $\cseq{\ty{\hyper{G}}}{\conf{C}}{R}$ and $\tm{\conf{C}}\cred\tm{\conf{D}}$, then
          $\cseq{\ty{\hyper{G}}}{\conf{D}}{R}$.
\end{lem}
  \begin{proof}
      By induction on the derivation of $\tm{\config{C}} \cred \tm{\config{D}}$.
    Where there is a choice for $\phi$, we prove the case for $\phi = \main$ and
    expand $\tm{\config{T}[M]}$ to $\tm{\main (E[M])}$ for some evaluation
    context $\tm{E}$; the
    other cases are similar.

    \begin{case}{E-Reify-Fork}
      \[
          \tm{
                \main E[\fork\;V] \cred (\nu x y)(\main E[x] \parallel
                \child V \app y )
            }
      \]

      Assumption:
      \begin{mathpar}
        \inferrule*
        { \tseq{\ty{\Gamma}}{ E[\fork\;V]}{ T } }
        { \cseq{\ty{\Gamma}}{ \main E[\fork\;V] }{T} }
      \end{mathpar}

      By Lemma~\ref{lem:hgv:subterm-typ}, there exist $\ty{\Gamma_1}, \ty{\Gamma_2},
      \ty{S}$ such that
      $\ty{\Gamma} = \ty{\Gamma_1}, \ty{\Gamma_2}$ and
      $\tseq{\ty{\Gamma_1}, \ty{\Gamma_2}}{E[\fork\;V]}{T}$ and:

      \begin{mathpar}
        \inferrule*
        { \tseq{\ty{\Gamma_2}}{ V }{ \tylolli{S}{\tyends } } }
        { \tseq{\ty{\Gamma_2}}{\fork\;V}{\co{S}} }
      \end{mathpar}

      By Lemma~\ref{lem:hgv:replacement}:
      \begin{mathpar}
        \inferrule*
        { \tseq{\ty{\Gamma_1}, \tm{x} : \ty{\co{S}}}{E[x]}{T} }
        { \cseq{\ty{\Gamma_1}, \tm{x} : \ty{\co{S}}}{\main E[x] }{T} }
      \end{mathpar}

      By \textsc{TM-App}, $\tseq{\ty{\Gamma_2}, \tm{y} : \ty{S}}{V \app y}{\tyends}$ and so
      by \textsc{TC-Child}, $\cseq{\ty{\Gamma_2}, \tm{y} : \ty{S}}{V \app y}{\tychild}$

      Recomposing:

      \begin{mathpar}
        \inferrule*
        {
          \inferrule*
          {
            \inferrule*
            { \tseq{\ty{\Gamma_1}, \tm{x} : \ty{\co{S}}}{E[x]}{T} }
            { \cseq{\ty{\Gamma_1}, \tm{x} : \ty{\co{S}}}{\main E[x]}{T}} \\
\inferrule*
            { \tseq{\ty{\Gamma_2}, \tm{y} : \ty{S}}{V \app y}{\tyends} }
            { \cseq{\ty{\Gamma_2}, \tm{y} : \ty{S}}{\child (V \app y) }{\tychild} }
          }
          {
              \cseq
              { \ty{\Gamma_1}, \tm{x} : \ty{\co{S}} \hypersep \ty{\Gamma_2},
              \tm{y} : \ty{S} }
                { \main E[x] \parallel \child (V \app y)}
                { T }
          }
        }
        {
            \cseq
                {\ty{\Gamma_1}, \ty{\Gamma_2}}
                { (\nu x y) (\main E[x] \parallel \child (V \app y) }
                {T}
        }
      \end{mathpar}

      as required.
    \end{case}

 \begin{case}{E-Comm-Send}
   \[
       \tm{(\nu x y)( \main E[\gvsend{V}{x}] \parallel \child E'[\recv\;{y}])}
        \cred
        \tm{(\nu x y)( \main E[x] \parallel \child E'[(V, y)])}
   \]

  Assumption:
  \begin{mathpar}
    \inferrule*
    {
      \inferrule*
      {
        \inferrule*
        { \tseq{\ty{\Gamma}, \tm{x} : \ty{S}}{E[\gvsend{V}{x}]}{U}}
        { \cseq{\ty{\Gamma}, \tm{x} : \ty{S}}{\main E[\gvsend{V}{x}] }{U} }
        \\
        \inferrule*
        { \tseq{\ty{\Gamma'}, \tm{y} : \ty{\co{S}}}{E'[\recv\;{y}]}{\tyends } }
        { \cseq{\ty{\Gamma'}, \tm{y} : \ty{\co{S}}}{\child E'[\recv\;{y}] }{\tychild} }
      }
      { \cseq{\ty{\Gamma}, \tm{x} : \ty{S}
          \hypersep \ty{\Gamma'}, \tm{y} : \ty{\co{S}}}
          {\main E[\gvsend{V}{x}] \parallel \child E'[\recv\;{y}] }{U} }
    }
    { \cseq{\ty{\Gamma}, \ty{\Gamma'}}{(\nu x y)( \main E[\gvsend{V}{x}]
        \parallel \child E'[\recv\;{y}]) }{U}
    }
  \end{mathpar}

  By Lemma~\ref{lem:hgv:subterm-typ}, there exist $\ty{\Gamma_1}, \ty{\Gamma_2},
  \ty{S}$ such
  that $\ty{\Gamma} = \ty{\Gamma_1}, \ty{\Gamma_2}$, and $
  \tseq{\ty{\Gamma_1}, \ty{\Gamma_2}, \tm{x} : \ty{S}}{E[\gvsend{V}{x}]}{U}$ and:

  \begin{mathpar}
    \inferrule*
    {
        \tseq{\ty{\Gamma_2}}{V}{T} \\
        \tseq{\tm{x} : \ty{\tysend{T}{S'}}}{x}{\tysend{T}{S'}}
    }
    { \tseq{\ty{\Gamma_2}, \tm{x} : \ty{\tysend{T}{S'}}}{\gvsend{V}{x}}{S'} }
  \end{mathpar}

  With the knowledge that $\ty{S} = \ty{\tysend{T}{S'}}$, we can refine our original
  derivation:

  \begin{mathpar}
    \inferrule*
    {
      \inferrule*
      {
        \inferrule*
        { \tseq{\ty{\Gamma_1}, \ty{\Gamma_2},
            \tm{x} : \ty{\tysend{T}{S'}}}{E[\gvsend{V}{x}]}{U} }
            { \cseq{\ty{\Gamma_1}, \ty{\Gamma_2}, \tm{x} : \ty{\tysend{T}{S'}}}{\main E[\gvsend{V}{x}] }{U} }
        \\
        \inferrule*
        { \tseq{\ty{\Gamma'}, \tm{y} : \ty{\tyrecv{T}{\co{S'}}}}{E'[\recv\;{y}]}{ \tyends } }
        { \cseq{\ty{\Gamma'}, \tm{y} : \ty{\tyrecv{T}{\co{S'}}}}{\child E'[\recv\;{y}] }{\tychild} }
      }
      { \cseq
          {\ty{\Gamma_1}, \ty{\Gamma_2}, \tm{x} : \ty{\tysend{T}{S'}}
          \hypersep \ty{\Gamma'}, \tm{y} : \ty{\tyrecv{T}{\co{S'}}}}
          {\main E[\gvsend{V}{x}] \parallel \child E'[\recv\;{y}] }
          { U }
      }
    }
    {
        \cseq
            {\ty{\Gamma_1}, \ty{\Gamma_2}, \ty{\Gamma'} }
            { (\nu x y)( \main E[\gvsend{V}{x}] \parallel \child E'[\recv\;{y}]) }
            { U }
    }
  \end{mathpar}

  Again by Lemma~\ref{lem:hgv:subterm-typ}, we have that
  $\tseq{\ty{\Gamma'}, \tm{y} : \ty{\tyrecv{T}{\co{S'}}}}{E'[\recv\;{y}]}{\tyends}$ and:

  \begin{mathpar}
    \inferrule*
    { \tseq{\tm{y} : \ty{\tyrecv{T}{\co{S'}}}}{y}{\tyrecv{T}{\co{S'}}} }
    { \tseq{\tm{y} : \ty{\tyrecv{T}{\co{S'}}}}{\recv\;{y}}{\typrod{T}{\co{S'}} } }
  \end{mathpar}

  We can show:
  \begin{mathpar}
    \inferrule*
    { \tseq{\ty{\Gamma_2}}{V}{T} \\ \tseq{\tm{y} : \ty{\co{S'}}}{\tm{y}}{\ty{\co{S'}}} }
    { \tseq{\ty{\Gamma_2}, \tm{y} : \ty{\co{S'}}}{(V, y)}{\typrod{T}{\co{S'}} } }
  \end{mathpar}

  By Lemma~\ref{lem:hgv:replacement}, we have that
  $\tseq{\ty{\Gamma_2}, \ty{\Gamma'}, \tm{y} : \ty{\co{S'}}}{E'[(V, y)]}{\co{S'}}$.

  Recomposing:

  \begin{mathpar}
    \inferrule*
    {
      \inferrule*
      {
        \inferrule*
        { \tseq{\ty{\Gamma_1}, \tm{x} : \ty{S'}}{E[x]}{U} }
        { \cseq{\ty{\Gamma_1}, \tm{x} : \ty{S'}}{\main E[x]}{U} } \\
\inferrule*
        { \tseq{\ty{\Gamma_2}, \ty{\Gamma'}, \tm{y} : \ty{\co{S'}}}{E'[(V, y)]}{\tyends} }
        { \cseq{\ty{\Gamma_2}, \ty{\Gamma'}, \tm{y} : \ty{\co{S'}}}{\child E'[(V, y)]}{\tychild} }
      }
      {
          \cseq
          {\ty{\Gamma_1}, \tm{x} : \ty{S'} \hypersep \ty{\Gamma_2},
          \ty{\Gamma'}, \tm{y} : \ty{\co{S'}} }
            {\main E[x] \parallel \child E'[(V, y)] }
            { U }
      }
    }
    { \cseq
        { \ty{\Gamma_1}, \ty{\Gamma_2}, \ty{\Gamma'} }
        { (\nu x y)( \main E[x] \parallel \child E'[(V, y)]) }
        { U }
    }
  \end{mathpar}

  as required.
 \end{case}

\begin{case}{E-Comm-Close}

  \[
      \tm{(\nu x y)( \config{T}[\wait\;{x}] \parallel \child y)} \cred
      \tm{\config{T}[()]}
  \]

  Taking $\config{T} = \main E$, assumption:

  \begin{mathpar}
    \inferrule*
    {
      \inferrule*
      {
        \inferrule*
        { \tseq{\ty{\Gamma}, \tm{x} : \ty{\tyendr}}{E[\wait\;{x}]}{T} }
        { \tseq{\ty{\Gamma}, \tm{x} : \ty{\tyendr}}{\main E[\wait\;{x}] }{T} }
        \\
        \inferrule*
        {
          \inferrule*
          { }
          { \tseq{\tm{y} : \ty{\tyends}}{y}{\tyends} }
        }
        { \cseq{\tm{y} : \ty{\tyends}}{\child y}{\tychild} }
      }
      { \cseq
          {\ty{\Gamma}, \tm{x} : \ty{\tyendr} \hypersep \tm{y} : \ty{\tyends} }
          {\main E[\wait\;{x}] \parallel \child y}
          { T }
      }
    }
    { \cseq
        { \ty{\Gamma} }
        { (\nu x y) (\main E[\wait\;{x}] \parallel \child y)  }
        { T }
    }
  \end{mathpar}

  By Lemma~\ref{lem:hgv:subterm-typ}, we have that:

  \begin{mathpar}
    \inferrule*
    { \tseq{\tm{x} : \ty{\tyendr}}{x}{\tyendr} }
    { \tseq{\tm{x} : \ty{\tyendr}}{\wait\;{x}}{\tyunit} }
  \end{mathpar}

  By Lemma~\ref{lem:hgv:replacement}, $\tseq{\ty{\Gamma}}{E[()]}{T}$.

  Recomposing:
  \begin{mathpar}
    \inferrule*
    { \tseq{\ty{\Gamma}}{E[()]}{T} }
    { \cseq{\ty{\Gamma}}{\main E[()]}{T} }
  \end{mathpar}

  as required.
\end{case}

\begin{case}{E-Reify-Link}

  \[
      \tm{F[\gvlink{x}{y}]} \cred
      \tm{(\nu z z')(\linkconfig{z}{x}{y} \parallel F[z'] )}
  \]
  where $\tm{z}, \tm{z'}$ fresh.

  Taking $\tm{F} = \tm{\main E}$, we have that:

  \begin{mathpar}
    \inferrule*
    { \tseq{\ty{\Gamma}}{E[\gvlink{x}{y}]}{T} }
    { \cseq{\ty{\Gamma}}{\main E[\gvlink{x}{y}]}{T} }
  \end{mathpar}

  By Lemma~\ref{lem:hgv:subterm-typ}, we have that $\ty{\Gamma} = \ty{\Gamma'},
  \tm{x} : \ty{S}, \tm{y} : \ty{\co{S}}$ such that:

  \begin{mathpar}
    \inferrule*
    {
      \inferrule*
      { \tseq{\tm{x} : \ty{S}}{x}{S} \\ \tseq{\tm{y} : \ty{\co{S}}}{y}{\co{S} } }
      { \tseq{\tm{x} : \ty{S}, \tm{y} : \ty{\co{S}}}{(x, y)}{\typrod{S}{\co{S}}} }
    }
    { \tseq{\tm{x} : \ty{S}, \tm{y}: \ty{\co{S}}}{\gvlink{x}{y}}{\tychild} }
  \end{mathpar}

  By Lemma~\ref{lem:hgv:replacement}, we have that $\tseq{\ty{\Gamma'}, \tm{z} :
  \ty{\tyends}}{E[z]}{T}$.

  Reconstructing:

  \begin{mathpar}
    \inferrule*
    {
      \inferrule*
      {
        \inferrule*
        { }
        { \cseq
            {\tm{z} : \ty{\tyendr}, \tm{x} : \ty{S}, \tm{y} : \ty{\co{S}} }
            {\linkconfig{z}{x}{y}}
            {\tychild}
        } \\
\cseq
            { \ty{\Gamma'}, \tm{z} : \ty{\tyends} }
            {\main E[z]}
            {T}
}
      {
          \cseq
          { \tm{z} : \ty{\tyendr}, \tm{x} : \ty{S}, \tm{y} : \ty{\co{S}}
                  \hypersep \ty{\Gamma'}, \tm{z} : \ty{\tyends} }
          { \linkconfig{z}{x}{y} \parallel \main E[z] }
          { T }
      }
    }
    { \cseq
        { \ty{\Gamma'}, \tm{x} : \ty{S}, \tm{y} : \ty{\co{S}} }
        { (\nu z z')( \linkconfig{z}{x}{y} \parallel \main E[z] )}
        { T }
    }
  \end{mathpar}

  as required.
\end{case}

\begin{case}{E-Comm-Link}

  \[
    \tm{(\nu z z')(\nu x x')(\linkconfig{z}{x}{y} \parallel \child z \parallel
    \main M)}
    \cred
    \tm{\main (M \{ y / x \})}
  \]

  Assumption:

  \begin{mathpar}
    \inferrule*
    {
      \inferrule*
      {
        \inferrule*
        {
          \inferrule*
          { }
          { \cseq
              { \tm{x} : \ty{S}, \tm{y} : \ty{\co{S}}, \tm{z} : \ty{\tyendr} }
              { \linkconfig{z}{x}{y}}
              { \tychild }
          }
          \\
          \inferrule*
          {
            \inferrule*
            { \tseq{\tm{z'} : \ty{\tyends}}{z}{\tyends} }
            { \cseq{\tm{z'} : \ty{\tyends}}{\child z}{\tychild} }
            \\
            \inferrule*
            { \tseq{\ty{\Gamma}, \tm{x'} : \ty{\co{S}}}{M}{T} }
            { \cseq{\ty{\Gamma}, \tm{x'} : \ty{\co{S}}}{\main M}{T} }
          }
          { \cseq
              {\tm{z'} : \ty{\tyends} \hypersep \ty{\Gamma}, \tm{x'} :
              \ty{\co{S}}}
            {\child z  \parallel \main M}
            { T }
          }
        }
        { \cseq
            {\tm{x} : \ty{S}, \tm{y} : \ty{\co{S}}, \tm{z} : \ty{\tyendr}
            \hypersep \tm{z'} : \ty{\tyends} \hypersep \ty{\Gamma}, \tm{x'} : \ty{\co{S}} }
            { \linkconfig{z}{x}{y} \parallel \child z' \parallel \main M }
            { T }
        }
      }
      {
          \cseq
          { \ty{\Gamma}, \tm{y} : \ty{\co{S}}, \tm{z} : \ty{\tyendr} \hypersep
          \tm{z'} : \ty{\tyends} }
            { (\nu x x')(\linkconfig{z}{x}{y} \parallel \child z' \parallel \main M)}
            { T }
      }
    }
    {
        \cseq
            { \ty{\Gamma}, \tm{y} : \ty{\co{S}} }
            { (\nu z z')(\nu x x')(\linkconfig{z}{x}{y} \parallel \child z' \parallel \main M)}
            { T }
    }
  \end{mathpar}

  By Lemma~\ref{lem:hgv:substitution}, $\tseq{\ty{\Gamma}, \tm{y'} : \ty{\co{S}}}{M \{ y / x' \}}{T}$, thus:

\begin{mathpar}
  \inferrule*
  {
      \tseq{\ty{\Gamma}, \tm{y'} : \ty{\co{S}}}{M \{ y / x' \}}{T}
  }
  { \cseq{\ty{\Gamma}, \tm{y'} : \ty{\co{S}}}{\main M \{ y / x' \}}{T} }
\end{mathpar}

as required.
\end{case}

  \begin{case}{E-Res}
    \[
        \tm{(\nu x y) \config{C}} \cred \tm{(\nu x y) \config{D}} \qquad
        \text{if } \tm{\config{C}} \cred \tm{\config{D}}
     \]

     Immediate by the IH.
  \end{case}

  \begin{case}{E-Par}
    \[
        \tm{\config{C} \parallel \config{D}} \cred
        \tm{\config{C}' \parallel \config{D}} \qquad
        \text{if } \tm{\config{C}} \cred \tm{\config{C}'}
    \]

    Immediate by the IH.
  \end{case}

  \begin{case}{E-Equiv}
    \[
        \tm{\config{C}}  \cred  \tm{\config{D}} \quad \text{if }
        \tm{\config{C}} \equiv \tm{\config{C}'}, \tm{\config{C}'} \cred
        \tm{\config{D}'}, \text{ and }
        \tm{\config{D}'} \equiv \tm{\config{D}}
    \]
  \end{case}

  Assumption: $\cseq{\ty{\hyper{G}}}{\config{C}}{R}$.

  By~\Cref{lem:hgv-equiv-pres}, $\cseq{\ty{\hyper{G}}}{\config{C}'}{R}$.

  By the IH, $\cseq{\ty{\hyper{G}}}{\config{D}'}{R}$.

  By~\Cref{lem:hgv-equiv-pres}, $\cseq{\ty{\hyper{G}}}{\config{D}}{R}$, as
  required.

  \begin{case}{E-Lift-M}
    \[
        \tm{\phi M} \cred \tm{\phi N} \qquad \text{if } \tm{M} \tred \tm{N}
    \]

    Immediate by Lemma~\ref{lem:gv:term-pres}. \qedhere
  \end{case}
\end{proof}

\hgvpres*
\begin{proof}
    A direct corollary of Lemmas~\ref{lem:hgv-equiv-pres}
    and~\ref{lem:hgv-cred-pres}.
\end{proof}

\subsection{Progress}

Functional reduction satisfies progress: under an environment only containing
runtime names, a term will either reduce, be a value, or be ready to perform a
communication action.

\begin{lem}[Progress, Terms]\label{lem:gv:term-progress}
    If $\tseq{\ty{\Psi}}{M}{T}$, then either $\tm{M}$ is a value,
    or there exists some $\tm{N}$ such that
    $\tm{M} \tred \tm{N}$, or $\tm{M}$ can be written $\tm{E[N]}$ for some\\
    $\tm{N} \in \{ \tm{\fork\;{V}},
      \tm{\send\;{(V, W)}},
      \tm{\recv\;{V}},
      \tm{\wait\;{V}},
      \tm{\link\;(V, W)} \}$.
\end{lem}
\begin{proof}
    A standard induction on the derivation of $\tseq{\ty{\Psi}}{M}{T}$.
\end{proof}

Note that tree canonical forms can be defined inductively:

\[
    \tm{\tcf} ::= \tm{\phi M} \sep \tm{(\nu x y)(\taux \parallel \tcf)}
\]

We assume the same requirement for configurations $\tm{\tcf}$ as the
non-inductive definition of tree canonical forms: i.e., that for a
configuration $\tm{(\nu x y)(\taux \parallel \tcf)}$, that $x \in
\fv(\tm{\taux})$.

\Cref{lem:hgv:open-progress} follows as a direct corollary of a slightly more
verbose property, which follows from the inductive definition of TCFs.

\begin{defi}[Open progress]
    Suppose $\cseq{\ty{\Psi}}{\config{\tcf}}{R}$, where $\tm{\tcf} \not\cred$.
    We say that $\tm{\tcf}$
  \emph{satisfies open progress} if:

  \begin{enumerate}
      \item $\tm{\config{C}} = \tm{(\nu x x')(\taux \parallel \tcf'})$, where:
      \begin{enumerate}
          \item There exist $\ty{\Psi_1}, \ty{\Psi_2}$ such that $\ty{\Psi} =
              \ty{\Psi_1}, \ty{\Psi_2}$
\item $\cseq{\ty{\Psi_1}, \tm{x} {:} \ty{S}}{\taux}{\tychild}$ for some
              session type $\ty{S}$,
              and $\blocked{\tm{\taux}}{\tm{y}}$ for some $\tm{y} \in
              \fv(\ty{\Psi_1}, \tm{x} {:} \ty{S})$
\item $\cseq{\ty{\Psi_2}, \tm{x'} {:} \ty{\co{S}}}{\config{D}}{R}$, where
              $\tm{\tcf'}$ satisfies open progress
      \end{enumerate}
  \item $\tm{\tcf} = \tm{\phi M}$, and either $\tm{M}$ is a value,
      or $\blocked{\tm{\phi M}}{\tm{x}}$ for some $\tm{x} \in \fv(\ty{\Psi})$.
  \end{enumerate}
\end{defi}

\begin{lem}[Open progress]
    If $\cseq{\ty{\Psi}}{\tcf}{R}$ and $\tm{\tcf} \not\cred$, then $\tm{\tcf}$
    satisfies open progress.
\end{lem}
\begin{proof}
    By induction on the derivation of $\cseq{\ty{\hyper{G}}}{\tcf}{R}$.
    By the definition of canonical forms, it must be the case that
    $\tm{\config{C}}$
    is of the form $\tm{(\nu x y)(\taux \parallel \tcf')}$ or $\tm{\main M}$.

  We show the case where
  $\tm{\config{C}} = \tm{(\nu x y)(\child M \parallel \tcf')}$; the cases
  for
  $\tm{\config{A}} = \tm{\linkconfig{z}{x}{x'}}$ and
  $\tm{\config{C}} = \tm{\main M}$ follow similar reasoning.

  Assumption:
  \begin{mathpar}
    \inferrule*
    {
      \inferrule*
      {
          \cseq{\ty{\Psi_1}, \tm{x} : \ty{S}}{\tm{\taux}}{\tychild}
        \\
        \cseq{\ty{\Psi_2}, \tm{y} : \ty{\co{S}}}{\tcf'}{R}
      }
      { \cseq
          { \ty{\Psi_1}, \tm{x} : \ty{S} \hypersep \ty{\Psi_2}, \tm{y} : \ty{\co{S}}}
          {\taux \parallel \tcf' }
          { R }
      }
    }
    { \cseq
        { \ty{\Psi_1}, \ty{\Psi_2} }
        { (\nu x y) (\child M \parallel \tcf')}
        { R }
    }
  \end{mathpar}
  In both cases, by the induction hypothesis,
  $\cseq{\ty{\Psi_2}, \tm{y} : \ty{\co{S}}}{\tm{\tcf'}}{\ty{T}}$ satisfies open progress.

  \begin{subcase}{$\tm{\taux} = \tm{\child M}$}
    %\sloppypar
    By Lemma~\ref{lem:gv:term-progress}, either $\tm{M}$ is a value, or $\tm{M}$ can be
    written $\tm{E[N]}$ for some communication and concurrency construct $\tm{N \in \{
    \fork\;{V}, \gvsend{V}{W}, \recv\;{V}, \wait\;{V},}$ $\tm{ \gvlink{V}{W} \}}$.

    Otherwise, $\tm{M}$ is a communication or concurrency construct. If $\tm{N} =
    \tm{\fork\;V}$, then reduction could occur by \textsc{E-Reify-Fork}.
If $\tm{N} = \tm{\link\;(V, W)}$, then by the type schema for $\tm\link$, we
    have that $\tm{\gvlink{V}{W}}$ must be of the form $\tm{\gvlink{z}{z'}}$ for
    $\tm{z, z'} \in \fv(\ty{\Psi}, \tm{x} : \ty{S})$ and could reduce by
    \textsc{E-Reify-Link}.

    Otherwise, it must be the case that $\blocked{\tm{\child M}}{\tm{z}}$ for
    some $\tm{z} \in \fv(\ty{\Psi_1}, \tm{x} : \ty{S})$.

    Thus, $\tm{(\nu x y) (\child M \parallel \config{D})}$ satisfies open progress, as
    required.
\end{subcase}

\begin{subcase}{$\tm{\taux} = \tm{\linkconfig{z_1}{z_2}{z_3}}$}
    We have that $\tm{z_1, z_2, z_3} \in \fv(\ty{\Psi_1}, \tm{x} : \ty{S})$, and the thread must be
    blocked by definition. \qedhere
\end{subcase}
\end{proof}

\endgroup

 }
%  \clearpage
  {\section{Omitted Proofs for \cref{sec:relation-to-gv}: Relation between HGV and GV}\label{appendix:gv-hgv}
\usingnamespace{hgv}

\gvinhgv*
\begin{proof}
    By induction on the derivation of $\cseq{\ty{\Gamma}}{\config{C}}{R}$.

  \begin{case}{\textsc{TG-New}}
    Assumption:

    \gv{
    \begin{mathpar}
      \inferrule*
      { \gvseq{\ty{\Gamma}, \tm{\lock{y}{y'}} :
      \ty{\lockty{S}}}{\tm{\config{C}}}{\ty{R}} }
      { \gvseq{\ty{\Gamma}}{(\nu y y') \config{\tm{C}}}{\ty{R}} }
    \end{mathpar}
    }

    Suppose $\ty{\Gamma} = \tm{\lock{x_1}{x'_1}} : \ty{\lockty{S_1}},
    \ldots, \tm{\lock{x_n}{x'_n}}: \ty{\lockty{S_n}}$
    (for clarity, without loss of generality, we assume the
    absence of non-session variables.  As these are simply split between
    environments, they can be added orthogonally).

    By the IH, we have that there exists some hyper-environment $\ty{\hyper{G}}$ such
    that $\cseq{\ty{\hyper{G}}}{\tm{\config{C}}}{\ty{R}}$,
    where $\ty{\hyper{G}}$ is a splitting of
    $\ty{\Gamma}, \tm{\lock{y}{y'}} : \ty{\lockty{S}}$.

    Since $\ty{\hyper{G}}$ is a splitting of $\tm{\config{C}}$, we know that
    $\tm{y} : \ty{S} \in \ty{\hyper{G}}$ and
    $\tm{y'} : \ty{\co{S}} \in \ty{\hyper{G}}$, and that $\ty{\hyper{G}}$ has a
    tree structure with respect to names $\{ \{ \tm{x_1}, \tm{x'_1} \}, \ldots,
    \{ \tm{x_n}, \tm{x'_n} \}, \{ \tm{y}, \tm{y'} \} \}$.

    Since $\ty{\hyper{G}}$ has a tree structure, by definition we have that
    $\ty{\hyper{G}} = \ty{\hyper{G}'} \hypersep
    \ty{\Gamma_1}, \tm{y} : \ty{S} \hypersep \ty{\Gamma_2}, \tm{y'} : \ty{\co{S}}$
    for some $\ty{\hyper{G}'}, \ty{\Gamma_1}, \ty{\Gamma_2}$, where
    $\ty{\hyper{G}'}$ has a tree structure.

    By Lemma~\ref{lem:tree-structure-iff} (clause 1, left-to-right),
    $\ty{\hyper{G}'} \hypersep
    \ty{\Gamma_1}, \ty{\Gamma_2}$ has a tree structure with respect to names
    $\{ \{ \tm{x_1}, \tm{x'_1} \}, \ldots, \{ \tm{x_n}, \tm{x'_n} \} \}$.

    Thus, we can show:

    \begin{mathpar}
      \inferrule*
      { \cseq
          {\ty{\hyper{G}'} \hypersep \ty{\Gamma_1},
          \tm{y} : \ty{S} \hypersep \ty{\Gamma_2}, \tm{y'} : \ty{\co{S}} }
          { \tm{\config{C}} }
          { \ty{R} }
      }
      { \cseq{\ty{\hyper{G}'} \hypersep \ty{\Gamma_1},
      \ty{\Gamma_2}}{(\nu \tm{y} \tm{y'})\tm{\config{C}}}{\ty{R}} }
    \end{mathpar}

    where
    $\ty{\hyper{G}'} \hypersep \ty{\Gamma_1}, \ty{\Gamma_2}$ has a tree
    structure with respect to names $\{ \{ \tm{x_1}, \tm{x'_1} \}, \ldots,
        \{ \tm{x_n}, \tm{x'_n} \} \}$ and is therefore a
        splitting of $\ty{\Gamma}$, as required.
  \end{case}

  \begin{case}{\textsc{TG-Connect}$_1$}
    Assumption:

    \gv{
    \begin{mathpar}
      \inferrule*
      {
          \gvseq{\ty{\Gamma_1}, \tm{y} : \ty{S}}{\config{C}}{\ty{R_1}} \\
          \gvseq{\ty{\Gamma_2}, \tm{y'} : \ty{\co{S}}}{\tm{\config{D}}}{\ty{R_2}}
      }
      { \gvseq
          {\ty{\Gamma_1}, \ty{\Gamma_2}, \tm{\lock{y}{y'}} : \ty{\lockty{S}} }
          { \tm{\config{C} \parallel \config{D}} }
          { \ty{R_1} \isect \ty{R_2} }
      }
    \end{mathpar}
    }

    Suppose $\ty{\Gamma_1} = \tm{\lock{x_1}{x'_1}} : \ty{\lockty{S_1}}, \ldots,
    \tm{\lock{x_m}{x'_m}} : \ty{\lockty{S_m}}$
    and $\ty{\Gamma_2} = \tm{\lock{x_{m + 1}}{x'_{m + 1}}} :
    \ty{\lockty{S_{m + 1}}}, \ldots, \tm{\lock{x_n}{x'_n}} : \ty{\lockty{S_n}}$.

    By the IH, there exist hyper-environments $\ty{\hyper{G}}, \ty{\hyper{H}}$ such that:
    \gv{
    \begin{enumerate}
        \item $\ty{\hyper{G}}$ is a splitting of $\ty{\Gamma_1}, \tm{y} : \ty{S}$
        \item $\ty{\hyper{H}}$ is a splitting of $\ty{\Gamma_2}, \tm{y'} : \ty{\co{S}}$
        \item $\gvseq{\ty{\hyper{G}}}{\tm{\config{C}}}{\ty{R_1}}$
        \item $\gvseq{\ty{\hyper{H}}}{\tm{\config{D}}}{\ty{R_2}}$
    \end{enumerate}
    }

    By the definition of splittings, $\ty{\hyper{G}}$ and $\ty{\hyper{H}}$ can be
    written $\ty{\hyper{G}} = \ty{\hyper{G}'} \hypersep \ty{\Gamma'_1}, \tm{y} :
    \ty{S}$ and $\ty{\hyper{H}} =
    \ty{\hyper{H}'} \hypersep \ty{\Gamma'_2}, \tm{y'} : \co{\ty{S}}$ for some
    $\ty{\Gamma'_1}, \ty{\Gamma'_2}$.
    Furthermore, $\ty{\hyper{G}}$ has a tree structure with respect\linebreak[4]to
    $\{ \{ \tm{x_1}, \tm{x'_1} \}, \ldots, \{ \tm{x_m}, \tm{x'_m} \} \} $ and
    $\ty{\hyper{H}}$ has a tree structure with respect to\linebreak[4]
    $\{ \{ \tm{x_{m + 1}}, \tm{x'_{m + 1}} \}, \ldots, \{ \tm{x_n}, \tm{x'_n} \} \} $.

    By Lemma~\ref{lem:tree-structure-iff} (clause 2, left-to-right),
    $\ty{\hyper{G}'} \hypersep \ty{\Gamma'_1}, \tm{y} : \ty{S} \hypersep
     \ty{\hyper{H}'} \hypersep \ty{\Gamma'_2}, \tm{y'} : \ty{\co{S}}$
    has a tree structure with respect to
    $ \{ \{ \tm{x_1}, \tm{x'_1} \}, \ldots, \{ \tm{x_n}, \tm{x'_n} \}, \{
    \tm{y}, \tm{y'} \}  \}$
    and therefore $\ty{\hyper{G}} \hypersep \ty{\hyper{H}}$ is a splitting of
    $\ty{\Gamma_1}, \ty{\Gamma_2}, \tm{\lock{y}{y'}} : \ty{\lockty{S}}$.

    Recomposing in HGV:
    \begin{mathpar}
      \inferrule*
      {
          \cseq{\ty{\hyper{G}}}{\tm{\config{C}}}{\ty{R_1}} \\
          \cseq{\ty{\hyper{H}}}{\tm{\config{D}}}{\ty{R_2}}
      }
      { \cseq
          {\ty{\hyper{G}} \hypersep \ty{\hyper{H}} }
          { \tm{\config{C}} \parallel \tm{\config{D}} }
          { \ty{R_1} \isect \ty{R_2} }
      }
    \end{mathpar}

    as required.
  \end{case}

  \begin{case}{\textsc{TG-Connect}$_2$}
    Similar to \textsc{TG-Connect}$_1$.
  \end{case}

  \begin{case}{\textsc{TG-Child}}
    Assumption:

    \gv{
    \begin{mathpar}
      \inferrule*
      { \tseq{\ty{\Gamma}}{\tm{M}}{\ty{\tyends}} }
      { \gvseq{\ty{\Gamma}}{\tm{\child M}}{\ty{\tychild}} }
    \end{mathpar}
    }

    Since we mandated that variables of type $\ty{\lockty{S}}$ cannot appear in
    terms, there are no names of type $\ty{\lockty{S}}$ in $\ty{\Gamma}$.
    Therefore, the singleton hyper-environment $\ty{\Gamma}$ is a valid
    splitting, and so we can conclude by \textsc{TC-Child} in HGV.
  \end{case}

  \begin{case}{\textsc{TG-Main}}
    Similar to \textsc{TG-Child}. \qedhere
  \end{case}
\end{proof}

\hgvingv*
\begin{proof}
  By induction on the number of $\nu$-bound names.

  In the case that $n = 0$, the result follows immediately by \textsc{TG-Child}
  or \textsc{TG-Main}.

  In the case that $n \ge 1$, we have that $\ty{\Gamma} = \ty{\Gamma_1}, \ty{\Gamma_2}$
  for some $\ty{\Gamma_1}, \ty{\Gamma_2}$ and:

  \begin{mathpar}
    \inferrule*
    {
      \inferrule*
      {
          \cseq{\ty{\Gamma_1}, \tm{x} : \ty{S}}{\tm{\child L}}{\ty{\tychild}} \\
      \cseq{\ty{\Gamma_2}, \tm{y} : \ty{\co{S}}}{\tm{\config{D}}}{\ty{R}}
      }
      {  \cseq
          { \ty{\Gamma_1}, \tm{x} : \ty{S} \hypersep \ty{\Gamma_2}, \tm{y} : \ty{\co{S}} }
          { \tm{\child L \parallel \config{D}} }
          {R}
      }
    }
    { \cseq
        { \ty{\Gamma_1, \Gamma_2} }
        { \tm{(\nu x y)(\child L \parallel \config{D})} }
        { \ty{R} }
    }
  \end{mathpar}

  such that $\tm{\config{D}}$ is in tree canonical form.
  That $\cseq{\ty{\Gamma_1}, \tm{x} : \ty{S}}{\tm{\child L}}{\ty{\tychild}}$
  follows by the definition of tree canonical forms, since $\tm{x} \in \fv(\tm{L})$.

  By the IH, $\cseq{\ty{\Gamma_2}, \tm{y} :
  \ty{\co{S}}}{\tm{\config{D}}}{\ty{R}}$ in GV.

  Thus, we can write:

  \begin{mathpar}
    \inferrule*
    {
      \inferrule*
      {
          \cseq{\ty{\Gamma_1}, \tm{x} : \ty{S}}{\tm{\child L}}{\ty{\tychild}} \\
          \cseq{\ty{\Gamma_2}, \tm{y} : \ty{\co{S}}}{\tm{\config{D}}}{\ty{R}}
      }
      { \cseq
          { \ty{\Gamma_1}, \ty{\Gamma_2}, \tm{\lock{x}{y}} : \ty{\lockty{S}} }
          { \tm{\child L \parallel \config{D}} }
          { \ty{R} }
      }
    }
    {
        \cseq
        { \ty{\Gamma_1}, \ty{\Gamma_2} }
            { \tm{(\nu x y)(\child L \parallel \config{D})} }
            { \ty{R} }
    }
  \end{mathpar}

  as required.
\end{proof}

 }
%  \clearpage
  {\section{Omitted Proofs for \cref{sec:relation-to-cp}: Relation between HGV and CP}\label{appendix:hcp-to-hgv}

\subsection{Full definition of \fgHGV}
{
  \usingnamespace{hgv}

\paragraph{Syntax}
\[
\small
  \usingnamespace{hgv}
  \begin{array}{lrcl}
    \text{Terms}
    & \tm{L}, \tm{M}, \tm{N}
    & \bnfdef & \tm{V}
      \sep      \tm{\letbind{x}{M}{N}}
      \sep      \tm{V\;W} \\
    &&\sep    & \tm{\letunit{V}{M}}
      \sep      \tm{\letpair{x}{y}{V}{M}} \\
    &&\sep    & \tm{\absurd{V}}
      \sep      \tm{\casesum{V}{x}{M}{y}{N}}
    \\
    \text{Values}
    & \tm{V}, \tm{W}
    & \bnfdef & \tm{x}
      \sep      \tm{K}
      \sep      \tm{\lambda x.M}
      \sep      \tm{\unit}
      \sep      \tm{\pair{V}{W}}
      \sep      \tm{\inl{V}}
      \sep      \tm{\inr{V}} \\
\text{Evaluation contexts}
    & \tm{E}
    & \bnfdef & \tm{\hole}
      \sep      \tm{\letbind{x}{E}{M}} \\
    \text{Thread contexts}
     & \tm{F}
     & \Coloneqq & \tm{\phi\;E}
  \end{array}
\]

\headersig{Typing rules for values}{$\tseq{\ty{\Gamma}}{V}{T}$}
{\small
\begin{mathpar}
  \inferrule*[lab=TV*-Var]{
    }{\tseq{\tmty{x}{T}}{x}{T}}

    \inferrule*[lab=TV*-Const]{
    }{\tseq{\emptyenv}{K}{T}}

    \inferrule*[lab=TV*-Lam]{
      \tseq{\ty{\Gamma},\tmty{x}{T}}{M}{U}
    }{\tseq{\ty{\Gamma}}{\lambda x.M}{\tylolli{T}{U}}}

    \inferrule*[lab=TV*-Unit]{
    }{\tseq{\emptyenv}{\unit}{\tyunit}}

    \inferrule*[lab=TV*-Pair]{
      \tseq{\ty{\Gamma}}{V}{T}
      \\
      \tseq{\ty{\Delta}}{W}{U}
    }{\tseq{\ty{\Gamma},\ty{\Delta}}{\pair{V}{W}}{\typrod{T}{U}}}

    \inferrule*[lab=TV*-Absurd]{
      \tseq{\ty{\Gamma}}{V}{\tyvoid}
    }{\tseq{\ty{\Gamma}}{\absurd{V}}{T}}

    \inferrule*[lab=TV*-Inl]{
      \tseq{\ty{\Gamma}}{V}{T}
    }{\tseq{\ty{\Gamma}}{\inl{V}}{\tysum{T}{U}}}

    \inferrule*[lab=TV*-Inr]{
      \tseq{\ty{\Gamma}}{V}{U}
    }{\tseq{\ty{\Gamma}}{\inr{V}}{\tysum{T}{U}}}
\end{mathpar}
}

  \headersig{Typing rules for terms}{$\tseq{\ty{\Gamma}}{M}{T}$}
  {\small
  \begin{mathpar}
      \inferrule*[lab=TM*-App]{
      \tseq{\ty{\Gamma}}{V}{\tylolli{T}{U}}
      \\
      \tseq{\ty{\Delta}}{W}{T}
    }{\tseq{\ty{\Gamma},\ty{\Delta}}{V\;W}{U}}

    \inferrule*[lab=TM*-Let]
    {
        \tseq{\ty{\Gamma}}{M}{T} \\
        \tseq{\ty{\Delta, x : \ty{T}}}{N}{U}
    }
    { \tseq{\ty{\Gamma}, \ty{\Delta}}{\letbind{x}{M}{N}}{U} }

    \inferrule*[lab=TM*-LetUnit]{
      \tseq{\ty{\Gamma}}{V}{\tyunit}
      \\
      \tseq{\ty{\Delta}}{M}{T}
    }{\tseq{\ty{\Gamma},\ty{\Delta}}{\letunit{V}{M}}{T}}

    \inferrule*[lab=TM*-LetPair]{
      \tseq{\ty{\Gamma}}{V}{\typrod{T}{T'}}
      \\
      \tseq{\ty{\Delta},\tmty{x}{T},\tmty{y}{T'}}{M}{U}
    }{\tseq{\ty{\Gamma},\ty{\Delta}}{\letpair{x}{y}{V}{M}}{U}}

    \inferrule*[lab=TM*-CaseSum]{
      \tseq{\ty{\Gamma}}{V}{\tysum{T}{T'}}
      \\\\
      \tseq{\ty{\Delta},\tmty{x}{T}}{M}{U}
      \\
      \tseq{\ty{\Delta},\tmty{y}{T'}}{N}{U}
    }{\tseq{\ty{\Gamma},\ty{\Delta}}{\casesum{V}{x}{M}{y}{N}}{U}}
\end{mathpar}
}

The typing of constants is the same as for HGV.

\paragraph{Operational Semantics}
The operational semantics for \fgHGV is the same as for HGV
(Figure~\ref{fig:hgv-reduction}), with the addition of the following explicit
rule for $\tm{\calcwd{let}}$:

\[
    \begin{array}{lrcl}
        \textsc{E-Let} & \tm{\letbind{x}{V}{M}} & \cred & \tm{\subst{M}{V}{x}}
    \end{array}
\]
Similarly, \fgHGV directly inherits HGV's runtime typing.
}

\subsection{Translating \fgHGV to HCP}

The translation is guaranteed to have only internal (i.e., $\alpha$ or
$\beta$) transitions and transitions on the dedicated output channel.
More specifically:
\begin{lem}\hfill
  \label{lem:hgv-to-hcp-terms}
  \begin{itemize}
  \item If $\hcp{\tm{\gvcpcnf{C}{r}}\lto{\ell}}$, then $\tm{\ell} \in
    \{ \tm{\alpha}, \tm{\beta} \}$ or $\ell = \ell_r$.
  \item If $\hcp{\tm{\gvcpcom{M}{r}}\lto{\ell}}$ and $\hgv{\tm{M}}$ is
    a non-value, then $\tm{\ell} \in \{ \tm{\alpha}, \tm{\beta} \}$.
  \item If $\hgv{\tm{V}}$ is a value, then
    $\hcp{\tm{\gvcpcom{V}{r}\lto{\ell_r}}}$.
  \item If $\hcp{\tm{\gvcpval{V}{r}}\lto{\ell}}$ then $\tm{\ell} \in
    \{ \tm{\alpha}, \tm{\beta} \}$.
  \end{itemize}
\end{lem}

\begin{proof}
  By induction on the structure of $\hgv{\tm{M}}$.
\end{proof}
We do not use the above lemma directly, but it is a useful sanity
check.

\begin{defi}[process contexts]
  A~process context $\tm{\plug{P}{\;}}$ is a process with a single hole, denoted $\tm{\hole}$. We extend the typing rules, LTS and typing rules to process contexts. We write $\hcp{\seq{\plug{P}{\;}}{\ty{\hyper{G}}/\ty{\hyper{H}}}}$ to mean that $\hcp{\tm{\plug{P}{\;}}}$ is typed under hyper-environment $\hcp{\ty{\hyper{H}}}$ expecting a process typed under $\hcp{\ty{\hyper{G}}}$, \ie if $\hcp{\seq{Q}{\ty{\hyper{G}}}}$ then $\hcp{\seq{\plug{P}{Q}}{\ty{\hyper{H}}}}$.
\end{defi}

\begin{defi}
  A~process $\tm{P}$ is \emph{blocked on} $\tm{x}$ if it only has transitions $\hcp{\tm{P}\lto{\ell_x}}$.
\end{defi}

\begin{lem}
  \label{lem:hgv-to-hcp-substitution}
\begin{sloppypar}
  If $\tm{\plug{P}{\;}}$ is a process context with $\tm{z},\tm{w},\tm{w'}\not\in\cn(\tm{\plug{P}{\;}})$, and $\tm{Q}$ is a process blocked on $\tm{w'}$, then $\hcp{\tm{\res{w}{w'}{(\ppar{\plug{P}{\link{z}{w}}}{Q})}}\bis_\ta\tm{\plug{P}{\subst{Q}{z}{w'}}}}$.
\end{sloppypar}
\end{lem}
\begin{proof}
  \usingnamespace{hcp}
  By induction on the process context $\tm{\plug{P}{\;}}$.
  \begin{case}{$\tm{\hole}$}
    \[
      \begin{array}{lrll}
        \multicolumn{3}{l}{\tm{\res{w}{w'}{(\ppar{\link{z}{w}}{Q})}}}
        \\
        & \ato
        & \tm{\subst{Q}{z}{w'}}
        \\
        & \sbis
        & \tm{\subst{Q}{z}{w'}}
        & \text{(by reflexivity)}
      \end{array}
    \]
  \end{case}
  \begin{case}{$\tm{\res{x}{y}{\plug{P}{\;}}}$}
    \[
      \begin{array}{lrll}
        \multicolumn{3}{l}{\tm{\res{w}{w'}{(\ppar{\res{x}{y}{(\plug{P}{\link{z}{w}})}}{Q})}}}
        \\
        & \sbis
        & \tm{\res{x}{y}{\res{w}{w'}{(\ppar{\plug{P}{\link{z}{w}}}{Q})}}}
        & \text{(by Lemma~\ref{cor:hcp-sbis-to-bis})}
        \\
        & \bis_\ta
        & \tm{\res{x}{y}{(\plug{P}{\subst{Q}{z}{w'}})}}
        & \text{(by Lemma~\ref{cor:hcp-sbis-to-bis} and IH)}
      \end{array}
    \]
  \end{case}
  \begin{case}{$\tm{\ppar{\plug{P}{\;}}{R}}$}
    \[
      \begin{array}{lrll}
        \multicolumn{3}{l}{\tm{\res{w}{w'}{(\ppar{\ppar{\plug{P}{\link{z}{w}}}{R}}{Q})}}}
        \\
        & \sbis
        & \tm{\ppar{\res{w}{w'}{(\ppar{\plug{P}{\link{z}{w}}}{Q})}}{R}}
        & \text{(by Lemma~\ref{cor:hcp-sbis-to-bis})}
        \\
        & \bis_\ta
        & \tm{\ppar{\plug{P}{\subst{Q}{z}{w'}}}{R}}
        & \text{(by Lemma~\ref{cor:hcp-sbis-to-bis} and IH)}
      \end{array}
    \]
  \end{case}
  \begin{case}{$\tm{\ppar{R}{\plug{P}{\;}}}$}
    \[
      \begin{array}{lrll}
        \multicolumn{3}{l}{\tm{\res{w}{w'}{(\ppar{\ppar{R}{\plug{P}{\link{z}{w}}}}{Q})}}}
        \\
        & \sbis
        & \tm{\ppar{R}{\res{w}{w'}{(\ppar{\plug{P}{\link{z}{w}}}{Q})}}}
        & \text{(by Lemma~\ref{cor:hcp-sbis-to-bis})}
        \\
        & \bis_\ta
        & \tm{\ppar{R}{\plug{P}{\subst{Q}{z}{w'}}}}
        & \text{(by Lemma~\ref{cor:hcp-sbis-to-bis} and IH)}
      \end{array}
    \]
  \end{case}
  \begin{case}{$\tm{\pi.{\plug{P}{\;}}}$}
    Since $\tm{Q}$ is blocked on $\tm{w'}$, the process $\tm{\res{w}{w'}{(\ppar{\pi.{\plug{P}{\link{z}{w}}}}{Q})}}$ has only one transition,
    \[
      \tm{\res{w}{w'}{(\ppar{\pi.{\plug{P}{\link{z}{w}}}}{Q})}}
      \lto{\pi}
      \tm{\res{w}{w'}{(\ppar{{\plug{P}{\link{z}{w}}}}{Q})}}.
    \]
    The process $\tm{\pi.{\plug{P}{\subst{Q}{z}{w'}}}}$ has only one transition, also with label $\tm{\pi}$,
    \[
      \tm{\pi.{\plug{P}{\subst{Q}{z}{w'}}}}
      \lto{\pi}
      \tm{{\plug{P}{\subst{Q}{z}{w'}}}}.
    \]
    The resulting processes are bisimilar by the induction hypothesis.
  \end{case}
  \begin{case}{$\tm{\offer{x}{\plug{P}{\;}}{\plug{P'}{\;}}}$}
    Since $\tm{Q}$ is blocked on $\tm{w'}$, the process\linebreak[4]$\tm{\res{w}{w'}{(\ppar{\offer{x}{\plug{P}{\link{z}{w}}}{\plug{P'}{\link{z}{w}}}}{Q})}}$ has only two transitions,
    \[
      \tm{\res{w}{w'}{(\ppar{\offer{x}{\plug{P}{\link{z}{w}}}{\plug{P'}{\link{z}{w}}}}{Q})}}
      \lto{\laboffinl{x}}
      \tm{\res{w}{w'}{(\ppar{{\plug{P}{\link{z}{w}}}}{Q})}}
    \]
    and
    \[
      \tm{\res{w}{w'}{(\ppar{\offer{x}{\plug{P}{\link{z}{w}}}{\plug{P'}{\link{z}{w}}}}{Q})}}
      \lto{\laboffinr{x}}
      \tm{\res{w}{w'}{(\ppar{{\plug{P'}{\link{z}{w}}}}{Q})}}.
    \]
    The process $\tm{\offer{x}{\plug{P}{\subst{Q}{z}{w'}}}{\plug{P'}{\subst{Q}{z}{w'}}}}$ has only two transitions, also with labels $\tm{\laboffinl{x}}$ and $\tm{\laboffinr{x}}$,
    \[
      \tm{\offer{x}{\plug{P}{\subst{Q}{z}{w'}}}{\plug{P'}{\subst{Q}{z}{w'}}}}
      \lto{\laboffinl{x}}
      \tm{{\plug{P}{\subst{Q}{z}{w'}}}}
    \]
    and
    \[
      \tm{\offer{x}{\plug{P}{\subst{Q}{z}{w'}}}{\plug{P'}{\subst{Q}{z}{w'}}}}
      \lto{\laboffinr{x}}
      \tm{{\plug{P'}{\subst{Q}{z}{w'}}}}.
    \]
    The resulting processes are bisimilar by the induction hypothesis. \qedhere
  \end{case}
\end{proof}

\begin{sloppypar}
\corhgvtohcpsubstitution*
\end{sloppypar}
\begin{proof}
  Immediately from Lemma~\ref{lem:hgv-to-hcp-substitution}.
\end{proof}

\begin{sloppypar}
  \lemhgvtohcptyping*
\end{sloppypar}

\begin{proof}
\emph{Part 1.}
\begin{itemize}
\item
    \begin{case}{$\tm{x}$}
    We have
    \(
      \hgv{
        \inferrule*{
        }{\tseq{\tmty{x}{T}}{x}{T}}
      }
    \)
    and
    \(
      \hgv{\tm{\gvcpval{x}{r}}}= \hcp{\tm{{\link{r}{x}}}}
    \). We can derive:
    \[
      \hcp{
        \inferrule*{
        }{\seq{\link{x}{r}}{\tmty x {\gvcpdown A}, \tmty r {\co{\gvcpdown A}}}}
      }
    \]
  \end{case}
\item
    \begin{case}{$\tm{K}$}
    We have one case for each communication primitive.
    \begin{itemize}
        \item Subcase $\hgv{\tm{\link}}$.  We have
      \(
      \hgv{\tmty{\link}{\tylolli{\typrod{S}{\co{S}}}{\tyends}}},
      \)
      where
      \[\hcp{
        \begin{aligned}
            \ty{\co{\gvcpdown{\hgv{\tylolli{\typrod{S}{\co{S}}}\tyends}}}}
          &= \ty{\gvcpup{\hgv{\tylolli {\typrod S {\co S}} \tyends}}} \\
          &= \ty{\typarr {\co {\gvcpup {\hgv{\typrod S {\co S}}}}} {(\tytens \tyone {\gvcpup {\hgv\tyends}})}} \\
          &= \ty{\typarr {(\typarr {\co {\gvcpup S}} {\gvcpup S})} {(\tytens \tyone \tybot)}}
      \end{aligned}
      }\]
      and
      \(
      \hgv{\tm{\gvcpval{\link}{r}}} = \hcp{\tm{\recv{r}{y}{\recv{y}{x}{\ping{r}{\wait{r}{\link{x}{y}}}}}}}.
      \)
      We can derive:
      \begin{mathpar}
        \hcp{
          \inferrule*{
            \inferrule*{
              \inferrule*{
                \inferrule*{
                  \inferrule*{
                  }{\seq {\link x y} {\tmty x {\co {\gvcpup S}}, \tmty y {\gvcpup S}}}
                }{\seq {\tm{\wait{r}{\link{x}{y}}}} {\tmty x {\co {\gvcpup S}}, \tmty y {\gvcpup S}, \tmty r \tybot}}
              }{\seq {\tm{\ping{r}{\wait{r}{\link{x}{y}}}}} {\tmty x {\co {\gvcpup S}}, \tmty y {\gvcpup S}, \tmty r {\tytens \tyone \tybot}}}
            }{\seq {\tm{\recv{y}{x}{\ping{r}{\wait{r}{\link{x}{y}}}}}} {\tmty y {\typarr {\co {\gvcpup S}} {\gvcpup S}}, \tmty r {\tytens \tyone \tybot}}}
          }{\seq {\tm{\recv{r}{y}{\recv{y}{x}{\ping{r}{\wait{r}{\link{x}{y}}}}}}} {\tmty r {\typarr {(\typarr {\co {\gvcpup S}} {\gvcpup S})} {(\tytens \tyone \tybot)}}}}}
      \end{mathpar}
      \item Subcase $\hgv{\tm{\fork}}$:
      We have
      \(
      \hgv{\tmty{\fork}{\tylolli{(\tylolli{S}{\tyends})}{\co{S}}}}
      \)
      where
      \[
        \hcp{
          \begin{aligned}
              \ty{\co{\gvcpdown{\hgv{\tylolli{(\tylolli{S}{\tyends})}{\co{S}}}}}}
            &= \ty{\gvcpup{\hgv{\tylolli{(\tylolli{S}{\tyends})}{\co{S}}}}} \\
            &= \ty{\typarr {\co{\gvcpup{\hgv{\tylolli S \tyends}}}} (\tytens
            \tyone {\gvcpup{\co S}})} \\
            &= \ty{\typarr {\co {(\typarr {\co {\gvcpup S}} {(\tytens \tyone
            {\gvcpup{\hgv\tyends}})})}} {(\tytens \tyone {\gvcpdown S})}} \\
            &= \ty{\typarr {(\tytens {\co {\gvcpdown S}} (\typarr \tybot \tyone))} {(\tytens \tyone {\gvcpdown S})}}
          \end{aligned}
        }\]
      and
      \(
      \hgv{\tm{\gvcpval{\fork}{r}}} = \hcp{\tm{\res{y}{y'}{(\recv{r}{x}{\usend{y}{x}{\ping{r}{\link{r}{y}}}} \parallel  \recv{y'}{x}{\usend{x}{y'}{\pong{x}{\close{x}{\halt}}}}})}}
      \).
      We derive:
      \begin{mathpar}
        \small
        \hcp{
          \inferrule*{
            \inferrule*{
              \inferrule*{
                \inferrule*{
                  \inferrule*{
                    \inferrule*{
                    }{\seq {\tm{\link r y}} {\tmty y {\co{\gvcpdown S}}, \tmty r {\gvcpdown S}}}
                  }{\seq {\tm{\ping{r}{\link{r}{y}}}} {\tmty y {\co{\gvcpdown S}}, \tmty r {\tytens \tyone {\gvcpdown S}}}}
                }{\seq {\tm{\usend{y}{x}{\ping{r}{\link{r}{y}}}}} {\tmty y {\tytens {(\typarr {\gvcpdown S} {(\tytens \tyone \tybot)})} {\co{\gvcpdown S}}}, \tmty x {(\tytens {\co {\gvcpdown S}} {(\typarr \tybot \tyone)})}, \tmty r {(\tytens \tyone {\gvcpdown S})}}}
              }{\seq {\tm{(\recv{r}{x}{\usend{y}{x}{\ping{r}{\link{r}{y}}}})}} {\tmty y {\tytens {(\typarr {\gvcpdown S} {(\tytens \tyone \tybot)})} {\co{\gvcpdown S}}}, \tmty r {\typarr {(\tytens {\co {\gvcpdown S}} {(\typarr \tybot \tyone)})} {(\tytens \tyone {\gvcpdown S})}}}}
              \\
              \mathcal{D}
            }{\seq {\tm{(\recv{r}{x}{\usend{y}{x}{\ping{r}{\link{r}{y}}}} \parallel  \recv{y'}{x}{\usend{x}{y'}{\pong{x}{\close{x}{\halt}}}}})} {\tmty y T, \tmty r {\typarr {(\tytens {\co {\gvcpdown S}} {(\typarr \tybot \tyone)})} {(\tytens \tyone {\gvcpdown S})}} \hypersep \tmty {y'} {\co T}}}
          }{\seq {{\tm{\res{y}{y'}{(\recv{r}{x}{\usend{y}{x}{\ping{r}{\link{r}{y}}}} \parallel  \recv{y'}{x}{\usend{x}{y'}{\pong{x}{\close{x}{\halt}}}}})}}} {\tmty r {\typarr {(\tytens {\co {\gvcpdown S}} {(\typarr \tybot \tyone)})} {(\tytens \tyone {\gvcpdown S})}}}}
        }
      \end{mathpar}
      where
      \[\hcp{
      \begin{aligned}
          \ty{T} &= \ty{\tytens {(\typarr {\gvcpdown S} {(\tytens \tyone \tybot)})} {\co{\gvcpdown S}}} \\
          \ty{\co T} &= \ty{\typarr {(\tytens {\co{\gvcpdown S}} {(\typarr \tybot \tyone)})} {{\gvcpdown S}}}
      \end{aligned}}
      \]
      \noindent
      and $\mathcal D$ is the derivation
      \begin{mathpar}
        %\tiny
        \hcp{
          \inferrule*{
            \inferrule*{
              \inferrule*{
                \inferrule*{
                  \inferrule*{
                  }{\seq {\tm\halt} \emptyhyperenv}
                }{\seq {\tm{\close{x}{\halt}}} {\tmty x \tyone}}
              }{\seq {\tm{\pong{x}{\close{x}{\halt}}}} {\tmty x {(\typarr \tybot \tyone)}}}
            }{\seq {\tm{\usend{x}{y'}{\pong{x}{\close{x}{\halt}}}}} {\tmty x {(\tytens {\co{\gvcpdown S}} {(\typarr \tybot \tyone)})}, \tmty {y'} {{{\gvcpdown S}}}}}
          }{\seq {\tm{\recv{y'}{x}{\usend{x}{y'}{\pong{x}{\close{x}{\halt}}}}}} {\tmty {y'} {\typarr {(\tytens {\co{\gvcpdown S}} {(\typarr \tybot \tyone)})} {{\gvcpdown S}}}}}
        }
      \end{mathpar}
  \item Subcase $\hgv{\ty{\send}}$:
      We have
      \(
      \hgv{\tmty{\send}{\tylolli{\typrod{T}{\tysend{T}{S}}}{S}}}
      \)
      where
      \[
        \hcp{
        \begin{aligned}
          \co{\gvcpdown{\hgv{\tylolli{\typrod{T}{\tysend{T}{S}}}{S}}}}
          &= \ty{\gvcpup{\hgv{\tylolli{\typrod{T}{\tysend{T}{S}}}{S}}}} \\
          &= \ty{\typarr {\co{\gvcpup{\hgv{\typrod T {\tysend T S}}}}} {(\tytens 1 {\gvcpup S})}} \\
          &= \ty{\typarr {(\typarr {\co{\gvcpup T}} {\gvcpdown {\hgv {\tysend T S}}})} {(\tytens 1 {\gvcpup S})}} \\
          &= \ty{\typarr {(\typarr {\co{\gvcpup T}} {(\tytens {\gvcpup T} {\gvcpdown S})})} {(\tytens 1 {\co{\gvcpdown S}})}}
        \end{aligned}
      }\]
      and
      \(
      \hgv{\tm{\gvcpval{\send}{r}}} = \hcp{\tm{\recv{r}{y}{\recv{y}{x}{\usend{y}{x}{\ping{r}{\link{r}{y}}}}}}}
      \).
      We derive:
      \begin{mathpar}
        \hcp{
          \inferrule*{
            \inferrule*{
              \inferrule*{
                \inferrule*{
                  \inferrule*{
                  }{\seq {\tm {\link r y}} {\tmty y {\gvcpdown S}, \tmty r {\co{\gvcpdown S}}}}
                }{\seq {\tm{\ping{r}{\link{r}{y}}}} {\tmty y {\gvcpdown S}, \tmty r {(\tytens 1 {\co{\gvcpdown S}})}}}
              }{\seq {\tm{\usend{y}{x}{\ping{r}{\link{r}{y}}}}} {\tmty x {\co{\gvcpup T}}, \tmty y {(\tytens {\gvcpup T} {\gvcpdown S})}, \tmty r {(\tytens 1 {\co{\gvcpdown S}})}}}
            }{\seq {\tm{\recv{y}{x}{\usend{y}{x}{\ping{r}{\link{r}{y}}}}}} {\tmty y {(\typarr {\co{\gvcpup T}} {(\tytens {\gvcpup T} {\gvcpdown S})})}, \tmty r {(\tytens 1 {\co{\gvcpdown S}})}}}
          }{\seq {\tm{\recv{r}{y}{\recv{y}{x}{\usend{y}{x}{\ping{r}{\link{r}{y}}}}}}} {\tmty r {\typarr {(\typarr {\co{\gvcpup T}} {(\tytens {\gvcpup T} {\gvcpdown S})})} {(\tytens 1 {\co{\gvcpdown S}})}}}}
        }
      \end{mathpar}
      \item Subcase $\hgv{\ty{\recv}}$:
      We have
      \(
      \hgv{\tmty{\recv}{\tylolli{\tyrecv{T}{S}}{\typrod{T}{S}}}}
      \)
      where
      \[
        \hcp{
          \begin{aligned}
              \ty{\co{\gvcpdown{\hgv{\tylolli{\tyrecv{T}{S}}{\typrod{T}{S}}}}}}
            &= \ty{\gvcpup{\hgv{\tylolli{\tyrecv{T}{S}}{\typrod{T}{S}}}}} \\
            &= \ty{\typarr {\co{\gvcpup{\hgv{\tyrecv T S}}}} {(\tytens 1 {\gvcpup{\hgv{\typrod T S}}})}} \\
            &= \ty{\typarr {(\typarr {\co{\gvcpup T}} {\co{\gvcpup S}})} {(\tytens 1 {(\tytens {\gvcpup T} {\gvcpup S})})}}
          \end{aligned}}
      \]
      and
      \(
      \hgv{\tm{\gvcpval{\recv}{r}}} =\hcp{\tm{\recv{r}{x}{\recv{x}{y}{\ping{r}{\usend{r}{y}{\link{r}{x}}}}}}}
      \).
      We derive:
      \begin{mathpar}
        \hcp{
          \inferrule*{
            \inferrule*{
              \inferrule*{
                \inferrule*{
                  \inferrule*{
                  }{\seq {\tm{\link r x}} {\tmty x {\co{\gvcpup S}}, \tmty r {\gvcpup S}}}
                }{\seq {\tm{\usend{r}{y}{\link{r}{x}}}} {\tmty y {\co{\gvcpup T}}, \tmty x {\co{\gvcpup S}}, \tmty r {\tytens {\gvcpup T} {\gvcpup S}}}}
              }{\seq {\tm{\ping{r}{\usend{r}{y}{\link{r}{x}}}}} {\tmty y {\co{\gvcpup T}}, \tmty x {\co{\gvcpup S}}, \tmty r {(\tytens 1 {(\tytens {\gvcpup T} {\gvcpup S})})}}}
            }{\seq {\tm{\recv{x}{y}{\ping{r}{\usend{r}{y}{\link{r}{x}}}}}} {\tmty x {(\typarr {\co{\gvcpup T}} {\co{\gvcpup S}})}, \tmty r {(\tytens 1 {(\tytens {\gvcpup T} {\gvcpup S})})}}}
          }{\seq {\tm{\recv{r}{x}{\recv{x}{y}{\ping{r}{\usend{r}{y}{\link{r}{x}}}}}}} {\tmty r {\typarr {(\typarr {\co{\gvcpup T}} {\co{\gvcpup S}})} {(\tytens 1 {(\tytens {\gvcpup T} {\gvcpup S})})}}}}
        }
      \end{mathpar}
  \item Subcase $\hgv{\tm{\wait}}$:
      We have
      \(
      \hgv{\tmty{\wait}{\tylolli{\tyendr}{\tyunit}}}
      \)
      where
      \[
        \hcp{
          \begin{aligned}
              \ty{\co{\gvcpdown{\hgv{\tylolli{\tyendr}{\tyunit}}}}}
            &= \ty{\gvcpup{\hgv{\tylolli\tyendr\tyunit}}} \\
            &= \ty{\typarr {\co{\gvcpup{\hgv\tyendr}}} {(\tytens \tyone {\gvcpup{\hgv\tyunit}})}} \\
            &= \ty{\typarr \tybot {(\tytens \tyone \tyone)}}
          \end{aligned}
        }
      \]
      and
      \(
      \hgv{\tm{\gvcpval{\wait}{r}}} = \hcp{\tm{\recv{r}{x}{\wait{x}{\ping{r}{\close{r}{\halt}}}}}}
      \).
      We derive
      \begin{mathpar}
        \hcp{
          \inferrule*{
            \inferrule*{
              \inferrule*{
                \inferrule*{
                  \inferrule*{
                  }{\seq {\tm\halt} \emptyhyperenv}
                }{\seq {\tm{\close{r}{\halt}}} {\tmty r \tyone}}
              }{\seq {\tm{\ping{r}{\close{r}{\halt}}}} {\tmty r {\tytens \tyone \tyone}}}
            }{\seq {\tm{\wait{x}{\ping{r}{\close{r}{\halt}}}}} {\tmty x \tybot, \tmty r {\tytens \tyone \tyone}}}
          }{\seq {\tm{\recv{r}{x}{\wait{x}{\ping{r}{\close{r}{\halt}}}}}} {\tmty r {\typarr \tybot {(\tytens \tyone \tyone)}}}}
        }
      \end{mathpar}
    \end{itemize}
  \end{case}
\item
    \begin{case}{$\tm{\lambda x.M}$}
    We assume $\tmty{\hgv{\tm{\gvcpcom M r}}}{\hcp{\gvcpdown\Gamma, \tmty x {\gvcpdown T}, \tmty r {\tytens \tyone {\co{\gvcpdown U}}}}}$ and derive
    \begin{mathpar}
      \hcp{
        \inferrule*{
          \seq {\tm{\gvcpcom M r}} {\gvcpdown \Gamma, \tmty x {\gvcpdown T}, \tmty r {\tytens \tyone {\gvcpup U}}}
        }{\seq {\tm{\recv r x {\gvcpcom M r}}} {\gvcpdown\Gamma, \tmty r {\typarr {\gvcpdown T} {(\tytens \tyone {\gvcpup U})}}}}
      }
    \end{mathpar}
  \end{case}
\item
    \begin{case}{$\tm{()}$}
    We derive:
    \begin{mathpar}
      \hcp{
        \inferrule*{
          \seq \halt \emptyhyperenv
        }{\seq {\close x \halt} {\tmty x \tyone}}}
    \end{mathpar}
  \end{case}
\item
    \begin{case}{$\hgv{\tm{\inl W}}$}
    We assume $\tmty{\hgv{\tm{\gvcpval W r}}}{\hcp{\gvcpdown \Gamma, \tmty r {\co{\gvcpdown T}}}}$ and derive
    \begin{mathpar}
      \hcp{
        \inferrule*{
          \seq {\tm {\gvcpval W r}} {\gvcpdown\Gamma, \tmty r {\co{\gvcpdown T}}}
        }{\seq {\tm{\inl r {\gvcpval W r}}} {\gvcpdown\Gamma, \tmty r {\typlus {\gvcpup T} {\gvcpup U}}}}
      }
    \end{mathpar}
  \end{case}
\item
    \begin{case}{$\tm{(V, W)}$}
    We assume $\tmty{\tm{\gvcpval V x}}{\hcp{\gvcpdown\Gamma, \tmty x {\co{\gvcpdown T}}}}$, $\tmty{\tm{\gvcpval W r}}{\cp{\gvcpdown\Delta, \tmty r {\co{\gvcpdown U}}}}$, and derive
    \begin{mathpar}
      \hcp{
        \inferrule*{
          \inferrule*{
            \seq {\tm {\gvcpval V x}} {\ty{\gvcpdown\Gamma}, \tmty x {\gvcpup T}}
            \\
            \seq {\tm {\gvcpval W r}} {\ty{\gvcpdown\Delta}, \tmty r {\gvcpup U}}
          }{\seq {\tm {\ppar {\gvcpval V x} {\gvcpval W r}}} {\ty{\gvcpdown\Gamma}, \tmty x {\gvcpup T} \hypersep \ty{\gvcpdown\Delta}, \tmty r {\gvcpup U}}}
        }{\seq {\tm {\send r x {(\ppar {\gvcpval V x} {\gvcpval W r})}}} {\ty{\gvcpdown\Gamma}, \ty{\gvcpdown\Delta}, \tmty r {\tytens {\gvcpup T} {\gvcpup U}}}}
      }
    \end{mathpar}
  \end{case}
\end{itemize}
\emph{Part 2.}
\begin{itemize}
\item
    \begin{case}{$\tm{V \; W}$}
    We assume $\tmty {\tm {\gvcpval V {y'}}} {\hcp{\ty{\gvcpdown\Gamma}, \tmty {y'} {\typarr {\co {\gvcpup T}} {(\tytens \tyone {\gvcpup U})}}}}$ and $\tmty {\tm {\gvcpval W {x'}}} {\hcp{\ty{\gvcpdown\Delta}, \tmty {x'} {\gvcpup T}}}$ and derive
    \begin{mathpar}
    \hspace*{-10pt}
      \tiny
      \hcp{
        \inferrule*{
          \inferrule*{
            \inferrule*{
              \inferrule*{
                \inferrule*{
                }{\seq {\tm {\link r y}} {\tmty y {\tytens {\gvcpup T} {(\typarr \tybot {\co{\gvcpup U}})}}, \tmty r {\tytens \tyone {\gvcpup U}}}}
              }{\seq {\tm {\usend y x {\link r y}}} {\tmty x {\co{\gvcpup T}}, \tmty y {\tytens {\gvcpup T} {(\typarr \tybot {\co{\gvcpup U}})}}, \tmty r {\tytens \tyone {\gvcpup U}}}}
              \\
              \mathcal{D}
            }{\seq {\tm {\ppar {\usend y x {\link r y}} {\ppar {\gvcpval V {y'}} {\gvcpval W {x'}}}}} {\tmty x {\co{\gvcpup T}}, \tmty y {\tytens {\gvcpup T} {(\typarr \tybot {\co{\gvcpup U}})}}, \tmty r {\tytens 1 {\gvcpup U}} \hypersep \ty{\gvcpdown\Gamma}, \tmty {y'} {\typarr {\co {\gvcpup T}} {(\tytens \tyone {\gvcpup U})}} \hypersep \ty{\gvcpdown\Delta}, \tmty {x'} {\gvcpup T}}}
          }{\seq {\tm {\res y {y'} {(\ppar {\usend y x {\link r y}} {\ppar {\gvcpval V {y'}} {\gvcpval W {x'}}})}}} {\tmty x {\co{\gvcpup T}}, \tmty r {\tytens \tyone {\gvcpup U}} \hypersep \ty{\gvcpdown\Gamma}, \ty{\gvcpdown\Delta}, \tmty {x'} {\gvcpup T}}}
        }{\seq {\tm {\res x {x'} {\res y {y'} {(\ppar {\usend y x {\link r y}} {\ppar {\gvcpval V {y'}} {\gvcpval W {x'}}})}}}} {\ty{\gvcpdown\Gamma}, \ty{\gvcpdown\Delta}, \tmty r {\tytens \tyone {\gvcpup U}}}}
      }
    \end{mathpar}
    where $\mathcal D$ is the derivation
    \begin{mathpar}
      \hcp{
        \inferrule*{
          \seq {\tm {\gvcpval V {y'}}} {\hcp{\ty{\gvcpdown\Gamma}, \tmty {y'} {\typarr {\co {\gvcpup T}} {(\tytens \tyone {\gvcpup U})}}}}
          \\
          \seq {\tm {\gvcpval W {x'}}} {\hcp{\ty{\gvcpdown\Delta}, \tmty {x'} {\gvcpup T}}}
        }{\seq {\tm {\ppar {\gvcpval V {y'}} {\gvcpval W {x'}}}} {\ty{\gvcpdown\Gamma}, \tmty {y'} {\typarr {\co {\gvcpup T}} {(\tytens \tyone {\gvcpup U})}} \hypersep \ty{\gvcpdown\Delta}, \tmty {x'} {\gvcpup T}}}
      }
    \end{mathpar}
  \end{case}
\item
    \begin{case}{$\hgv{\tm{\letpair x y V M}}$}\\
    We assume $\hcp{\tmty {\gvcpval V {y'}} {\ty{\gvcpdown\Gamma}, \tmty {y'} {\tytens {\gvcpup T} {\gvcpup {T'}}}}}$ and $\hcp{\tmty {\gvcpcom M r} {\ty{\gvcpdown\Delta}, \tmty x {\co{\gvcpup T}}, \tmty y {\co{\gvcpup {T'}}}, \tmty r {\tytens \tyone {\co{\gvcpdown U}}}}}$ and derive
    \begin{mathpar}
      \hcp{
        \inferrule*{
          \inferrule*{
            \inferrule*{
              \seq {\gvcpcom M r} {\ty{\gvcpdown\Delta}, \tmty x {\co{\gvcpup T}}, \tmty y {\co{\gvcpup {T'}}}, \tmty r {\tytens \tyone {\co{\gvcpdown U}}}}
            }{\seq {\recv y x {\gvcpcom M r}} {\ty{\gvcpdown\Delta}, \tmty y {\typarr {\co{\gvcpup T}} {\co{\gvcpup {T'}}}}, \tmty r {\tytens \tyone {\co{\gvcpdown U}}}}
            }
            \\
            \seq {\gvcpval V {y'}} {\ty{\gvcpdown\Gamma}, \tmty {y'} {\tytens {\gvcpup T} {\gvcpup {T'}}}}
          }{\seq {\ppar {\recv y x {\gvcpcom M r}} {\gvcpval V {y'}}} {\ty{\gvcpdown\Delta}, \tmty y {\typarr {\co{\gvcpup T}} {\co{\gvcpup {T'}}}}, \tmty r {\tytens \tyone {\co{\gvcpdown U}}}} \hypersep \ty{\gvcpdown\Gamma}, \tmty {y'} {\tytens {\gvcpup T} {\gvcpup {T'}}}}
        }{\seq {\res y {y'} {(\ppar {\recv y x {\gvcpcom M r}} {\gvcpval V {y'}})}} {\ty{\gvcpdown\Gamma}, \ty{\gvcpdown\Delta}, \tmty r {\tytens \tyone {\co{\gvcpdown U}}}}}
      }
    \end{mathpar}
  \end{case}
\item
    \begin{case}{$\hgv{\tm{\absurd V}}$}
    We assume $\tmty {\gvcpval V {x'}} {\hcp{\ty{\gvcpdown\Gamma}, \tmty {x'} \tynil}}$, and derive:
    \begin{mathpar}
      \hcp{
        \inferrule*{
          \inferrule*{
            \inferrule*{
            }{\seq {\absurd x} {\tmty r {\tytens \tyone {\co{\gvcpdown T}}}, \tmty {x} \tytop}}
            \\
            \seq {\gvcpval V {x'}} {\ty{\gvcpdown\Gamma}, \tmty {x'} \tynil}
          }{\seq {\ppar {\absurd x} {\gvcpval V {x'}}} {\tmty r {\tytens \tyone {\co{\gvcpdown T}}}, \tmty {x} \tytop \hypersep \ty{\gvcpdown{\Gamma}}, \tmty {x'} \tynil}}
        }{\seq {\res x {x'} {(\ppar {\absurd x} {\gvcpval V {x'}})}} {\ty{\gvcpdown\Gamma}, \tmty r {\tytens \tyone {\co{\gvcpdown T}}}}}
      }
    \end{mathpar}
  \end{case}
\item
    \begin{case}{$\hgv{\tm{\letbind x M N}}$}\\
    We assume $\tmty {\gvcpcom M {x'}} {\hcp{\ty{\gvcpdown\Gamma}, \tmty {x'} {\tytens \tyone {\gvcpup T}}}}$ and $\tmty {\gvcpcom N r} {\hcp{\ty{\gvcpdown\Delta}, \tmty x {\co{\gvcpup T}}, \tmty r {\gvcpup U}}}$ and derive
    \begin{mathpar}
      \hcp{
        \inferrule*{
          \inferrule*{
            \inferrule*{
              \seq {\gvcpcom N r} {\ty{\gvcpdown\Delta}, \tmty x {{\co{\gvcpup T}}}, \tmty r {\gvcpup U}}
            }{\seq {\pong x {\gvcpcom N r}} {\ty{\gvcpdown\Delta}, \tmty x {\typarr \tybot {\co{\gvcpup T}}}, \tmty r {\gvcpup U}}}
            \\
            {\seq {\gvcpcom M {x'}} {\ty{\gvcpdown\Gamma}, \tmty {x'} {\tytens \tyone {\gvcpup T}}}}
          }{\seq {\ppar {\pong x {\gvcpcom N r}} {\gvcpcom M {x'}}} {\ty{\gvcpdown\Delta}, \tmty x {\typarr \tybot {\co{\gvcpup T}}}, \tmty r {\gvcpup U} \hypersep \ty{\gvcpdown\Gamma}, \tmty {x'} {\tytens \tyone {\gvcpup T}}}}
        }{\seq {\tm {\res x {x'} {(\ppar {\pong x {\gvcpcom N r}} {\gvcpcom M {x'}})}}} {\ty{\gvcpdown\Gamma}, \ty{\gvcpdown\Delta}, \tmty r {\gvcpup U}}}
      }
    \end{mathpar}
  \end{case}
\item
    \begin{case}{$\hgv{\tm{\ret V}}$}
    We assume $\tmty {\gvcpval V r} {\ty{\gvcpdown\Gamma}, \tmty r {\gvcpup T}}$ and derive
    \begin{mathpar}
      \hcp{
        \inferrule*{
          \seq {\tm{\gvcpval V r}} {\ty{\gvcpdown\Gamma}, \tmty r {\gvcpup T}}
        }{\seq {\tm{\ping r {\gvcpval V r}}} {\ty{\gvcpdown\Gamma}, \tmty r {\tytens \tyone {\gvcpup T}}}}
      }
    \end{mathpar}
  \end{case}
\end{itemize}
\emph{Part 3.}
The cases are all by immediate induction.
\end{proof}

\begin{lem}
  \label{lem:hgv-to-hcp-evaluation-context}
  Let $\tm{F}$ be an \fgHGV evaluation context and $\tm{r}$ a result
  endpoint. Then there exists a process context $\tm{\gvcpevc{F}{r}}$ and a
  result endpoint $\tm{v} = \hr(\tm{F}, \tm{r})$ for the hole such that for all $\tm{M}$ we
  have that $\tm{\gvcpcnf{\plug{F}{M}}{r}} =
  \tm{\plug{\gvcpevc{F}{r}}{\gvcpcom{M}{v}}}$.
\end{lem}
\begin{proof}
    By induction on the structure of $\tm{F}$.
\end{proof}
In the above lemma, if $\tm{F}$ is the empty context then $\tm{v} =
\tm{r}$. Otherwise $\tm{v}$ is a variable bound by the process context
$\tm{\gvcpevc{F}{r}}$.

\begin{lem}[Operational Correspondence, Terms]
  \label{thm:hgv-to-hcp-oc-terms}
  If $\hgv{\tm{M}}$ is a well-typed term:
  \begin{enumerate}
  \item
    If $\hgv{\tm{M}\tred\tm{M'}}$, then there exists a $P$ such that
    $\hcp{\tm{\gvcpcom{M}{r}} \slto[\alpha]{\beta^+} P}$ and $\hcp{P \bis_{\ta} \tm{\gvcpcom{M'}{r}}}$; and
  \item
    if $\hcp{\tm{\gvcpcom{M}{r}}\bto\tm{P}}$, then there exists an
    $\hgv{\tm{M'}}$ and a $\hcp{P'}$ such that
    $\hgv{\tm{M}\tred\tm{M'}}$ and $\hcp{P \slto[\alpha]{\beta^*} P'}$
    and $\hcp{\tm{P'}\bis_{\ta}\tm{\gvcpcom{M'}{r}}}$.
  \end{enumerate}
\end{lem}
\begin{proof}
  \hspace*{1ex}\\\vspace*{-1\baselineskip}\begin{enumerate}
  \item By induction on the reduction $\hgv{\tm{M}\tred\tm{M'}}$.
    \begin{case}{\LabTirName{E-Lam}}
      \begin{mathpar}
        \begin{tikzcd}[cramped, column sep=huge]
          \hgv{\tm{(\lambda x.M)\;V}}
          \arrow[r, "\hgv{\tred}"]
          \arrow[d, "\gvcpcom{\cdot}{r}"]
          &
          \hgv{\tm{\subst{M}{V}{x}}}
          \arrow[ddd, "\gvcpcom{\cdot}{r}"]
          \\
          \hcp{\tm{\res{x}{x'}{\res{y}{y'}{(\ppar
                  {\usend{y}{x}{{\link{r}{y}}}}
                  {\ppar
                    {\recv{y'}{x}{\gvcpcom{M}{y'}}}
                    {\gvcpval{V}{x'}}
                  }
                  )}}}}
          \arrow[d, "\hcp{\bto\ato}"]
          \\
          \hcp{\tm{\res{x}{x'}{\res{y}{y'}{(\ppar
                  {{\link{r}{y}}}
                  {\ppar
                    {\gvcpcom{M}{y'}}
                    {\gvcpval{V}{x'}}
                  }
                  )}}}}
          \arrow[d, "\hcp{\ato}"]
          \\
          \hcp{\tm{{{\res{x}{x'}{(\ppar{\gvcpcom{M}{r}}{\gvcpval{V}{x'}})}}}}}
          \arrow[r, dash, "\hcp{\bis_\alpha}\;\text{(by Lemma~\ref{cor:hgv-to-hcp-substitution})}"]
          &
          \hgv{\tm{\gvcpcom{\subst{M}{V}{x}}{r}}}
      \end{tikzcd}
    \end{mathpar}
    \end{case}
    \begin{case}{\LabTirName{E-Unit}}
      \begin{mathpar}
        \begin{tikzcd}[cramped, column sep=10em]
          \hgv{\tm{\letunit{\unit}{M}}}
          \arrow[r, "\hgv{\tred}"]
          \arrow[d, "\gvcp{\cdot}{r}"]
          &
          \hgv{\tm{M}}
          \arrow[dd, "\gvcp{\cdot}{r}"]
          \\
          \hcp{\tm{\res{x}{x'}{(\ppar
                {{\wait{x}{\gvcpcom{M}{r}}}}
                {\close{x'}{\halt}}
                )}}}
          \arrow[d, "\hcp{\bto}"]
          \\
          \hcp{\tm{{\ppar{{{\gvcp{M}{r}}}}{{{\halt}}}}}}
          \arrow[r, dash, "\hcp{\sim}"]
          &
          \hgv{\tm{\gvcpcom{M}{r}}}
        \end{tikzcd}
      \end{mathpar}
    \end{case}
    \begin{case}{\LabTirName{E-Pair}}
      \begin{mathpar}
        \begin{tikzcd}[cramped, column sep=huge]
          \hgv{\tm{\letpair{x}{y}{\pair{V}{W}}{M}}}
          \arrow[r, "\hgv{\tred}"]
          \arrow[d, "\gvcp{\cdot}{r}"]
          &
          \hgv{\tm{\subst{\subst{M}{V}{x}}{W}{y}}}
          \arrow[dd, "\gvcp{\cdot}{r}"]
          \\
          \hcp{\tm{\res{y}{y'}{(\ppar
                {{\recv{y}{x}{\gvcpcom{M}{r}}}}
                {\send{y'}{x'}{(\ppar{\gvcpval{V}{x'}}{\gvcpval{W}{y'}})}}
                )}}}
          \arrow[d, "\hcp{\bto}"]
          \\
          \hcp{\tm{\res{y}{y'}{\res{x}{x'}{(\ppar
                  {{{\gvcp{M}{r}}}}
                  {\ppar{{\gvcpval{V}{x'}}}{{{\gvcpval{W}{y'}}}}})}}}}
          \arrow[r, dash, "\hcp{\bis_\alpha\;\text{(by Lemma~\ref{cor:hgv-to-hcp-substitution})}}"]
          &
          \hgv{\tm{\gvcpcom{\subst{\subst{M}{V}{x}}{W}{y}}{r}}}
        \end{tikzcd}
      \end{mathpar}
    \end{case}
    \begin{case}{\LabTirName{E-Inl}}
      \begin{mathpar}
        \begin{tikzcd}[cramped, column sep=huge]
          \hgv{\tm{\casesum{\inl{V}}{x}{M}{y}{N}}}
          \arrow[r, "\hgv{\tred}"]
          \arrow[d, "\gvcp{\cdot}{r}"]
          &
          \hgv{\tm{\subst{M}{V}{x}}}
          \arrow[dd, "\gvcp{\cdot}{r}"]
          \\
          \hcp{\tm{\res{x}{x'}{(\ppar
                {{\offer{x}{\gvcpcom{M}{r}}{\gvcpcom{\subst{N}{x}{y}}{r}}}}
                {\inl{x'}{\gvcpval{V}{x'}}}
                )}}}
          \arrow[d, "\hcp{\bto}"]
          \\
          \hcp{\tm{\res{x}{x'}{(\ppar{\gvcp{M}{r}}{{\gvcpval{V}{x'}}})}}}
          \arrow[r, dash, "\hcp{\bis_\alpha\;\text{(by Lemma~\ref{cor:hgv-to-hcp-substitution})}}"]
          &
          \hgv{\tm{\gvcpcom{\subst{M}{V}{x}}{r}}}
        \end{tikzcd}
      \end{mathpar}
    \end{case}
    \begin{case}{\LabTirName{E-Inr}}
      As \LabTirName{E-Inl}.
    \end{case}
    \begin{case}{\LabTirName{E-Let}}
      \begin{mathpar}
        \begin{tikzcd}[cramped, column sep=huge]
          \hgv{\tm{\letbind{x}{\ret{V}}{M}}}
          \arrow[r, "\hgv{\tred}"]
          \arrow[d, "\gvcp{\cdot}{r}"]
          &
          \hgv{\tm{\subst{M}{V}{x}}}
          \arrow[dd, "\gvcp{\cdot}{r}"]
          \\
          \hcp{\tm{\res{x}{x'}{(\ppar
                {\pong{x}{\gvcpcom{M}{r}}}
                {\ping{x'}{\gvcpval{V}{x'}}}
                )}}}
          \arrow[d, "\hcp{\bto\bto}"]
          \\
          \hcp{\tm{\res{x}{x'}{(\ppar{\gvcp{M}{r}}{{\gvcpval{V}{x'}}})}}}
          \arrow[r, dash, "\hcp{\bis_\alpha\;\text{(by Lemma~\ref{cor:hgv-to-hcp-substitution})}}"]
          &
          \hgv{\tm{\gvcpcom{\subst{M}{V}{x}}{r}}}
        \end{tikzcd}
      \end{mathpar}
    \end{case}
    \begin{case}{\LabTirName{E-Lift}}
      The induction hypothesis yields the reasoning steps depicted by
      the first diagram, which we use, together with HGV's
      \LabTirName{E-Lift} and HCP's \LabTirName{Str-Res} and
      \LabTirName{Str-Par2}, to justify the second diagram:
      \begin{mathpar}
        \begin{tikzcd}[cramped, column sep=huge]
          \hgv{\tm{M}}
          \arrow[r, "\hgv{\tred}"]
          \arrow[d, "\gvcpcom{\cdot}{r}"]
          &
          \hgv{\tm{M'}}
          \arrow[d, "\gvcpcom{\cdot}{r}"]
          \\
          \hcp{\tm{\gvcpcom{M}{r}}}
          \arrow[r, "\hcp{\slto[\alpha]{\beta^+}\bis_\alpha}"]
          &
          \hcp{\tm{\gvcpcom{M'}{r}}}
        \end{tikzcd}

        \begin{tikzcd}[cramped, column sep=huge]
          \hgv{\tm{\letbind{x}{\plug{E}{M}}{N}}}
          \arrow[r, "\hgv{\tred}"]
          \arrow[d, "\gvcp{\cdot}{r}"]
          &
          \hgv{\tm{\letbind{x}{\plug{E}{M'}}{N}}}
          \arrow[d, "\gvcp{\cdot}{r}"]
          \\
          \hcp{\tm{\res{x}{x'}{(\ppar
                {\pong{x}{\gvcpcom{N}{r}}}
                {\gvcpcom{M}{x'}}
                )}}}
          \arrow[r, dash, "\hcp{\slto[\alpha]{\beta^+}\bis_\alpha}"]
          &
          \hcp{\tm{\res{x}{x'}{(\ppar
                {\pong{x}{\gvcpcom{N}{r}}}
                {\gvcpcom{M'}{x'}}
                )}}}
        \end{tikzcd}
      \end{mathpar}
    \end{case}
  \item By induction on $\hgv{\tm{M}}$.
\begin{case}{$\hgv{\tm{U\;V}}$}
      There are two well-typed cases for $\tm{U}$: either $\hgv{\tm{U}=\tm{z}}$ for some $\hgv{\tm{z}}$; or $\hgv{\tm{U}=\tm{\lambda{x}.M}}$ for some $\hgv{\tm{x}}$ and $\hgv{\tm{M}}$. If $\hgv{\tm{U}=\tm{z}}$, we have ${\hcp{\tm{\res{x}{x'}{\res{y}{y'}{(\ppar{\usend{y}{x}{\link{r}{y}}}{\ppar{\link{z}{y'}}{\gvcpval{V}{x'}}})}}}\centernot{\bto}}}$, which contradicts our premise. Therefore, $\hgv{\tm{U}=\tm{\lambda{x}.M}}$. The only possible $\beta$-transition is the one in the following diagram:
      \begin{mathpar}
        \begin{tikzcd}
          \hgv{\tm{(\lambda x.M)\;V}}
          \arrow[r, "\hgv{\tred}"]
          \arrow[d, "\gvcpcom{\cdot}{r}"]
          &
          \hgv{\tm{\subst{M}{V}{x}}}
          \arrow[ddd, "\gvcpcom{\cdot}{r}"]
          \\
          \hcp{\tm{\res{x}{x'}{\res{y}{y'}{(\ppar
                  {\usend{y}{x}{{\link{r}{y}}}}
                  {\ppar
                    {\recv{y'}{x}{\gvcpcom{M}{y'}}}
                    {\gvcpval{V}{x'}}
                  }
                  )}}}}
          \arrow[d, "\hcp{\bto\ato}"]
          \\
          \hcp{\tm{\res{x}{x'}{\res{y}{y'}{(\ppar
                  {{\link{r}{y}}}
                  {\ppar
                    {\gvcpcom{M}{y'}}
                    {\gvcpval{V}{x'}}
                  }
                  )}}}}
          \arrow[d, "\hcp{\ato}"]
          \\
          \hcp{\tm{{{\res{x}{x'}{(\ppar{\gvcpcom{M}{r}}{\gvcpval{V}{x'}})}}}}}
          \arrow[r, dash, "\hcp{\bis_\alpha}\;\text{(by Lemma~\ref{cor:hgv-to-hcp-substitution})}"]
          &
          \hgv{\tm{\gvcpcom{\subst{M}{V}{x}}{r}}}
        \end{tikzcd}
      \end{mathpar}
      Hence, $\hgv{\tm{M'}=\tm{\subst{M}{V}{x}}}$.
    \end{case}

\begin{case}{$\hgv{\tm{\letunit{U}{M}}}$}
      There are two well-typed cases for $\tm{U}$: either $\hgv{\tm{U}=\tm{z}}$ for some $\hgv{\tm{z}}$; or $\hgv{\tm{U}=\tm{\unit}}$. If $\hgv{\tm{U}=\tm{z}}$, we have $\hcp{\tm{\res{x}{x'}{(\ppar{\wait{x}{\gvcpcom{M}{r}}}{\link{x'}{z}})}}\centernot{\bto}}$, which contradicts our premise. Therefore, $\hgv{\tm{U}=\tm{\unit}}$. The only possible $\beta$-transition is the one in the following diagram:
      \begin{mathpar}
        \begin{tikzcd}[cramped, column sep=10em]
          \hgv{\tm{\letunit{\unit}{M}}}
          \arrow[r, "\hgv{\tred}"]
          \arrow[d, "\gvcp{\cdot}{r}"]
          &
          \hgv{\tm{M}}
          \arrow[dd, "\gvcp{\cdot}{r}"]
          \\
          \hcp{\tm{\res{x}{x'}{(\ppar
                {{\wait{x}{\gvcpcom{M}{r}}}}
                {\close{x'}{\halt}}
                )}}}
          \arrow[d, "\hcp{\bto}"]
          \\
          \hcp{\tm{{\ppar{{{\gvcp{M}{r}}}}{{{\halt}}}}}}
          \arrow[r, dash, "\hcp{\sim}"]
          &
          \hgv{\tm{\gvcpcom{M}{r}}}
        \end{tikzcd}
      \end{mathpar}
      Hence, $\hgv{\tm{M'}=\tm{M}}$.
    \end{case}

\begin{case}{$\hgv{\tm{\letpair{x}{y}{U}{M}}}$}
      There are two well-typed cases for $\tm{U}$: either $\hgv{\tm{U}=\tm{z}}$ for some $\hgv{\tm{z}}$, or $\hgv{\tm{U}=\tm{\pair{V}{W}}}$. If $\hgv{\tm{U}=\tm{z}}$, we have ${\hcp{\tm{\res{y}{y'}{(\ppar{{\recv{y}{x}{\gvcpcom{M}{r}}}}{\link{y'}{z}})}}\centernot{\bto}}}$, which contradicts our premise. Therefore, $\hgv{\tm{U}=\tm{\pair{V}{W}}}$. The only possible $\beta$-transition is the one in the following diagram:
      \begin{mathpar}
        \begin{tikzcd}[cramped, column sep=huge]
          \hgv{\tm{\letpair{x}{y}{\pair{V}{W}}{M}}}
          \arrow[r, "\hgv{\tred}"]
          \arrow[d, "\gvcp{\cdot}{r}"]
          &
          \hgv{\tm{\subst{\subst{M}{V}{x}}{W}{y}}}
          \arrow[dd, "\gvcp{\cdot}{r}"]
          \\
          \hcp{\tm{\res{y}{y'}{(\ppar
                {{\recv{y}{x}{\gvcpcom{M}{r}}}}
                {\send{y'}{x'}{(\ppar{\gvcpval{V}{x'}}{\gvcpval{W}{y'}})}}
                )}}}
          \arrow[d, "\hcp{\bto}"]
          \\
          \hcp{\tm{\res{y}{y'}{\res{x}{x'}{(\ppar
                  {{{\gvcp{M}{r}}}}
                  {\ppar{{\gvcpval{V}{x'}}}{{{\gvcpval{W}{y'}}}}})}}}}
          \arrow[r, dash, "\hcp{\bis_\alpha\;\text{(by Lemma~\ref{cor:hgv-to-hcp-substitution})}}"]
          &
          \hgv{\tm{\gvcpcom{\subst{\subst{M}{V}{x}}{W}{y}}{r}}}
        \end{tikzcd}
      \end{mathpar}
    \end{case}

\begin{case}{$\hgv{\tm{\casesum{U}{x}{M}{x}{N}}}$}
      There are two well-typed cases\linebreak[4]for $\tm{U}$: either $\hgv{\tm{U}=\tm{z}}$ for some $\hgv{\tm{z}}$; or $\hgv{\tm{U}=\tm{\inl{V}}}$. If $\hgv{\tm{U}=\tm{z}}$, we have\linebreak[4]$\hcp{\tm{\res{x}{x'}{(\ppar{{\offer{x}{\gvcpcom{M}{r}}{\gvcpcom{\subst{N}{x}{y}}{r}}}}{\link{x'}{z}})}}\centernot{\bto}}$, which contradicts our premise. Therefore, $\hgv{\tm{U}=\tm{\inl{V}}}$. The only possible $\beta$-transition is the one in the following diagram:
      \begin{mathpar}
        \begin{tikzcd}[cramped, column sep=huge]
          \hgv{\tm{\casesum{\inl{V}}{x}{M}{y}{N}}}
          \arrow[r, "\hgv{\tred}"]
          \arrow[d, "\gvcp{\cdot}{r}"]
          &
          \hgv{\tm{\subst{M}{V}{x}}}
          \arrow[dd, "\gvcp{\cdot}{r}"]
          \\
          \hcp{\tm{\res{x}{x'}{(\ppar
                {{\offer{x}{\gvcpcom{M}{r}}{\gvcpcom{\subst{N}{x}{y}}{r}}}}
                {\inl{x'}{\gvcpval{V}{x'}}}
                )}}}
          \arrow[d, "\hcp{\bto}"]
          \\
          \hcp{\tm{\res{x}{x'}{(\ppar{\gvcp{M}{r}}{{\gvcpval{V}{x'}}})}}}
          \arrow[r, dash, "\hcp{\bis_\alpha\;\text{(by Lemma~\ref{cor:hgv-to-hcp-substitution})}}"]
          &
          \hgv{\tm{\gvcpcom{\subst{M}{V}{x}}{r}}}
        \end{tikzcd}
      \end{mathpar}
    \end{case}
    \begin{case}{$\hgv{\tm{\absurd{U}}}$}
      There is only one well-typed case for $\tm{U}$: $\hgv{\tm{U}=\tm{z}}$ for some $\hgv{\tm{z}}$. However, $\hcp{\tm{\res{x}{x'}{(\ppar{{\absurd{x}}}{\link{x'}{z}})}}\centernot{\bto}}$, which contradicts our premise.
    \end{case}
    \begin{case}{$\hgv{\tm{\letbind{x}{M}{N}}}$}
      There are two possible cases: either $\hgv{\tm{M}=\tm{V}}$; or $\hcp{\tm{\gvcpcom{M}{x'}}\bto\tm{P}}$ for some $\hcp{\tm{P}}$. If $\hgv{\tm{M}}$ is a value, the only possible $\beta$-transition is the one in the following diagram:
      \begin{mathpar}
        \begin{tikzcd}[cramped, column sep=huge]
          \hgv{\tm{\letbind{x}{\ret{V}}{M}}}
          \arrow[r, "\hgv{\tred}"]
          \arrow[d, "\gvcp{\cdot}{r}"]
          &
          \hgv{\tm{\subst{M}{V}{x}}}
          \arrow[dd, "\gvcp{\cdot}{r}"]
          \\
          \hcp{\tm{\res{x}{x'}{(\ppar
                {\pong{x}{\gvcpcom{M}{r}}}
                {\ping{x'}{\gvcpval{V}{x'}}}
                )}}}
          \arrow[d, "\hcp{\bto\bto}"]
          \\
          \hcp{\tm{\res{x}{x'}{(\ppar{\gvcp{M}{r}}{{\gvcpval{V}{x'}}})}}}
          \arrow[r, dash, "\hcp{\bis_\alpha\;\text{(by Lemma~\ref{cor:hgv-to-hcp-substitution})}}"]
          &
          \hgv{\tm{\gvcpcom{\subst{M}{V}{x}}{r}}}
        \end{tikzcd}
      \end{mathpar}
      Otherwise, if $\hcp{\tm{\gvcpcom{M}{x'}}\bto\tm{P}}$ for some $\hcp{\tm{P}}$, the induction hypothesis gives us an $\hgv{\tm{M'}}$ such that $\hgv{\tm{M}\tred\tm{M'}}$ and $\hcp{\tm{P}\bis\gvcpcom{\tm{M'}}{r}}$. We apply HGV's \LabTirName{E-Lift} and HCP's \LabTirName{Str-Res} and \LabTirName{Str-Par2}.
    \end{case}
    \begin{case}{$\hgv{\tm{\ret{V}}}$}
      We have $\hcp{\tm{\ping{r}{\gvcpval{V}{r}}}\centernot{\bto}}$, which contradicts our premise. \qedhere
    \end{case}
  \end{enumerate}
\end{proof}

\thmhgvtohcpocconfs*
\begin{proof}
  \hspace*{1ex}\\\vspace*{-1\baselineskip}\begin{enumerate}
  \item By induction on the reduction
    $\hgv{\tm{\conf{C}}\cred\tm{\conf{C}'}}$.
We implicitly make use of
    Lemma~\ref{lem:hgv-to-hcp-evaluation-context} throughout the proof
    in order to recast the translation of a plugged evaluation context
    $\hcp{\tm{\gvcpcnf{\plug{F}{M}}{r}}}$ into the plugging of the translated
    evaluation context with the translation of the plugged term
    $\hcp{\tm{\plug{\gvcpevc{F}{r}}{\gvcpcom{M}{v}}}}$ where $v = \hr(F, r)$.

    \begin{case}{\rulename{E-Reify-Fork}}
{\scriptsize
      \begin{mathpar}
        \begin{tikzcd}[cramped, column sep=huge]
          \hgv{\tm{\plug{F}{\fork\;V}}}
          \arrow[r, "\hgv{\cred}"]
          \arrow[d, "\gvcpcnf{\cdot}{r}"]
          &
          \hgv{\tm{\res{x}{x'}{(\ppar{\plug{F}{x}}{\child\;V\;x'})}}}
          \arrow[dd, "\gvcpcnf{\cdot}{r}"]
          \\
          \hcp{\tm{\plug{\gvcpevc{F}{r}}{
                      \res{z}{z'}{\res{y}{y'}{\left(
                      \begin{array}{@{}l@{}}
                      {
                        \usend{y}{z}{\link{v}{y}}
                        \parallel
                      }\\
                      {
                        \res{x}{x'}{\left(
                                    \begin{array}{@{}l@{}}
                                    \recv{y'}{w}{\usend{x}{w}{\ping{y'}{\link{y'}{x}}}}
                                    \parallel \\
                                    \recv{x'}{w}{\usend{w}{x'}{\pong{w}{\close{w}{\halt}}}}
                                    \end{array}\right)}
                        \parallel
                      }\\
                      {
                        \gvcpval{V}{z'}
                      }\\
                      \end{array}
                    \right)}}}
              }}
          \mkern-250mu
          \arrow[d, "\hcp{\bto\bto\ato}"]
          \\
          \hcp{\tm{
                  \plug{\gvcpevc{F}{r}}
                       {\res{z}{z'}{\left(
                         \begin{array}{@{}l@{}}
                         \res{x}{x'}{\left(
                           \begin{array}{@{}l@{}}
                           {\ping{v}{\link{v}{x}}}
                           \parallel \\
                           {\usend{z}{x'}{\pong{z}{\close{z}{\halt}}}} \\
                           \end{array}
                         \right)}
                         \parallel \\
                         {\gvcpval{V}{z'}}
                         \end{array}
              \right)}}}}
          \arrow[r, dash, "\hcp{\bis_\ta}"]
          &
          \hcp{\tm{\res{x}{x'}{\left(
                \begin{array}{@{}l@{}}
                {\plug{\gvcpevc{F}{r}}{\ping{v}{\link{v}{x}}}}
                \parallel \\
                \res{y}{y'}{\left(
                  \begin{array}{@{}l@{}}
                    \res{w}{w'}{\res{z}{z'}{\left(
                    \begin{array}{@{}l@{}}
                    {\usend{z}{w}{\link{y}{z}}}
                    \parallel \\
                    {\gvcpval{V}{z'}}
                    \parallel \\
                    {\link{w'}{x'}} \\
                    \end{array}
                    \right)}}
                  \parallel \\
                  {\pong{y'}{\close{y'}{\halt}}} \\
                  \end{array}
                  \right)} \\
                \end{array}
                \right)}}}
          \\
        \end{tikzcd}
      \end{mathpar}
}

The endpoint $\tm{v} = \hr(F, r)$. The final two terms are bisimilar
by Lemma~\ref{lem:hgv-to-hcp-substitution}.

\end{case}

    \begin{case}{\rulename{E-Reify-Link}}
{\scriptsize
      \begin{mathpar}
        \begin{tikzcd}[cramped, column sep=tiny]
          \hgv{\tm{\child\;\plug{E}{\link\;\pair{x}{y}}}}
          \arrow[r, "\hgv{\cred}"]
          \arrow[d, "\gvcpcnf{\cdot}{r}"]
          &
          \hgv{\tm{\res{z}{z'}{(\ppar{\linkconfig{z}{x}{y}}{\child\;\plug{E}{z'}})}}}
          \arrow[d, dash, shorten >=3mm]
          \\
          \hcp{\tm{
              \res{a}{a'}{(
              \plug{\gvcpcom{E}{r}}{
                \res{z}{z'}{\res{w}{w'}{(
                    {\usend{w}{z}{{\link{v}{w}}}}
                    \parallel
                    {\recv{w'}{t}{\recv{t}{s}{\ping{w'}{\wait{w'}{\link{s}{t}}}}}}
                    \parallel
                    {\usend{z'}{x}{\link{y}{z'}}}
                    \parallel
                    \ping{a'}{\close{a'}{\halt}}
                    )}}
              )}}}}
          \mkern-300mu
          \arrow[dd, "\hcp{\bto\bto\ato\ato}"]
          &
          ~
          \arrow[dd, "\gvcpcnf{\cdot}{r}"]
          \\
          \\
          \hcp{\tm{\res{a}{a'}{(
                \plug{\gvcpcom{E}{r}}{\ping{v}{\wait{v}{\link{x}{y}}}}
                \parallel
                \ping{a'}{\close{a'}{\halt}}
                )}}}
          \arrow[r, dash, "\hcp{\bis_\ta}"]
          &
          \hcp{\tm{\res{z}{z'}{(
                \ping{z}{\wait{z}{\link{x}{y}}}
                \parallel
                \res{a}{a'}{(
                  \gvcpcom{\plug{E}{\link{v}{z'}}}{a}
                  \parallel
                  \ping{a'}{\close{a'}{\halt}}
                  )}
              )}}}
        \end{tikzcd}
      \end{mathpar}
}
The endpoint $\tm{v} = \hr(F, r)$.
    \end{case}

    \begin{case}{\rulename{E-Comm-Link}}
{\scriptsize
     \begin{mathpar}
        \begin{tikzcd}
          \hgv{\tm{\res{z}{z'}{\res{x}{x'}{(\ppar{\ppar{\linkconfig{z}{x}{y}}{\child\; z'}}{\phi\;M})}}}}
          \arrow[r, "\hgv{\cred}"]
          \arrow[d, "\gvcpcnf{\cdot}{r}"]
          &
          \hgv{\tm{\phi\;(M \{ y / x' \})}}
          \arrow[dd, "\gvcpcnf{\cdot}{r}"]
          \\
          \hcp{\tm{\res{z}{z'}{\res{x}{x'}{(
                    {\ping{z}{\wait{z}{\link{x}{y}}}}
                    \parallel
                    \res{w}{w'}{(
                      {\link{z'}{w}}
                      \parallel
                      {\pong{w'}{\close{w'}{\halt}}}
                      )}
                    \parallel
                    \gvcpcnf{\phi\;M}{r}
                    )}}}}
          \mkern-50mu
          \arrow[d, "\hcp{\ato, \bto \times 3}"]
          \\
          \hcp{\tm{\subst{\gvcpcnf{\phi\;M}{r}}{y}{x'}}}
          \arrow[r, "\hcp{\bis_\ta}"]
          &
          \hcp{\tm{\gvcpcnf{\phi\;\subst{M}{y}{x'}}{r}}}
        \end{tikzcd}
      \end{mathpar}
}
    \end{case}

    \begin{case}{\rulename{E-Comm-Send}}
{\scriptsize
      \begin{mathpar}
        \begin{tikzcd}
          \hgv{\tm{\res{x}{x'}{(\ppar{\plug{F}{\send\;{\pair{V}{x}}}}{\plug{F'}{\recv\;{x'}}})}}}
          \arrow[r, "\hgv{\cred}"]
          \arrow[d, "\gvcpcnf{\cdot}{r}"]
          &
          \hgv{\tm{\res{x}{x'}{(\ppar{\plug{F}{x}}{\plug{F'}{\pair{V}{x'}}})}}}
          \arrow[d, dash, shorten >=5mm]
          \\
          \hcp{\tm{\res{x}{x'}{\left(
                  \begin{array}{l}
                    {\plug{\gvcpevc{F}{r}}{
                    \res{y}{y'}{\res{z}{z'}{(
                    {\usend{z}{y}{{\link{u}{z}}}}
                    \parallel
                    {\recv{z'}{t}{\recv{t}{s}{\usend{t}{s}{\ping{z'}{\link{z'}{t}}}}}}
                    \parallel
                    {\send{y'}{w}{(\ppar{\gvcpval{V}{w}}{\link{x}{y'}})}}
                    )}}
                    }}
                    \parallel
                    \\
                    {\plug{\gvcpcnf{F'}{r}}{
                    \res{y}{y'}{\res{z}{z'}{(
                    {\usend{z}{y}{{\link{v}{z}}}}
                    \parallel
                    {\recv{z'}{s}{\recv{s}{t}{\ping{z'}{\usend{z'}{t}{\link{z'}{s}}}}}}
                    \parallel
                    {\link{x'}{y'}}
                    )}}
                    }}
                  \end{array}
                \right)}}}
          \mkern-300mu
          \arrow[d, "\hcp{\bto \times 5, \ato \times 2}"]
          &
          ~
          \arrow[dd, "\gvcpcnf{\cdot}{r}"]
          \\
          \hcp{\tm{\res{x}{x'}{(
                {\plug{\gvcpevc{F}{r}}{
                    {(
                      {\usend{x}{w}{\ping{u}{\link{x}{u}}}}
                      \parallel
                      {\gvcpval{V}{w}}
                      )}}}
                \parallel
                {\plug{\gvcpcnf{F'}{r}}{
                    {\recv{x'}{t}{\ping{v}{\usend{v}{t}{\link{v}{x'}}}}}
                  }}
                )}}}
          \mkern-100mu
          \arrow[dd, "\hcp{\beta}"]
          &
          ~
          \\
          ~
          &
          \mkern-150mu
          \hcp{\tm{\res{x}{x'}{(
                {\plug{\gvcpevc{F}{r}}{\ping{u}{\link{x}{u}}}}
                \parallel
                {\plug{\gvcpcnf{F'}{r}}{
                    \ping{v}{\send{v}{w}{(\ppar{\gvcpval{V}{w}}{\link{v}{x'}}})}
                  }}
                )}}}
          \arrow[d, dash]
          \\
          \hcp{\tm{\res{x}{x'}{(
                {\plug{\gvcpevc{F}{r}}{
                    {(
                      {{\ping{u}{\link{u}{x}}}}
                      \parallel
                      {\gvcpval{V}{w}}
                      )}}}
                \parallel
                {\plug{\gvcpcnf{F'}{r}}{
                    {{\ping{v}{\usend{v}{w}{\link{v}{x'}}}}}
                  }}
                )}}}
          \arrow[r, dash, "\hcp{\bis_\ta}"]
          &
          ~
        \end{tikzcd}
      \end{mathpar}
}
The endpoint $\tm{u} = \hr(F, r)$ and the endpoint $\tm{v} = \hr(F', r)$.
    \end{case}

    \begin{case}{\rulename{E-Comm-Close}}
{\scriptsize
     \begin{mathpar}
        \begin{tikzcd}
          \hgv{\tm{\res{x}{x}{(\ppar{\child\;x}{\plug{F}{\wait\;x'}})}}}
          \arrow[r, "\hgv{\cred}"]
          \arrow[d, "\gvcpcnf{\cdot}{r}"]
          &
          \hgv{\tm{{\plug{F}{\unit}}}}
          \arrow[dd, "\gvcpcnf{\cdot}{r}"]
          \\
          \hcp{\tm{\res{x}{x}{\left(
                  \begin{array}{l}
                    {\res{y}{y'}{(\ppar{\ping{y}{\link{x}{y}}}{\pong{y'}{\close{y'}{\halt}}})}}
                    \parallel
                    \\
                    {\plug{\gvcpevc{F}{r}}{
                    \res{z}{z'}{\res{w}{w'}{(
                    {\usend{w}{z}{{\link{v}{w}}}}
                    \parallel
                    \recv{w'}{s}{\wait{s}{\ping{w'}{\close{w'}{\halt}}}}
                    \parallel
                    {\link{x'}{z'}}
                    )}}
                    }}
                  \end{array}
                \right)}}}
          \mkern-50mu
          \arrow[d, "\hcp{\bto \times 3, \ato \times 3, \bto}"]
          \\
          \hcp{\tm{{\plug{\gvcpevc{F}{r}}{{\ping{v}{\close{v}{\halt}}}}}}}
          \arrow[r, "="]
          &
          \hcp{\tm{{\plug{\gvcpevc{F}{r}}{{\ping{v}{\close{v}{\halt}}}}}}}
        \end{tikzcd}
      \end{mathpar}
}
The endpoint $\tm{v} = \hr(F, r)$.
    \end{case}

    \begin{case}{\rulename{E-Res}}
      \begin{mathpar}
        \begin{tikzcd}[cramped, column sep=huge]
          \hgv{\tm{(\nu x y) \config{C}}}
          \arrow[r, "\hgv{\cred}"]
          \arrow[d, "\gvcpcnf{\cdot}{r}"]
          &
          \hgv{\tm{(\nu x y) \config{C}'}}
          \arrow[d, "\gvcpcnf{\cdot}{r}"]
          \\
          \hcp{\tm{(\nu x y)\gvcpcnf{\config{C}}{r}}}
          \arrow[r, "\hcp{\slto[\alpha]{\beta^+}\bis_\ta} (\text{IH})"]
          &
          \hcp{\tm{(\nu x y)\gvcpcnf{\config{C}'}{r}}}
        \end{tikzcd}
      \end{mathpar}
    \end{case}

    \begin{case}{\rulename{E-Par}}
      \begin{mathpar}
        \begin{tikzcd}[cramped, column sep=huge]
          \hgv{\tm{\config{C} \parallel \config{D}}}
          \arrow[r, "\hgv{\cred}"]
          \arrow[d, "\gvcpcnf{\cdot}{r}"]
          &
          \hgv{\tm{\config{C}' \parallel \config{D}}}
          \arrow[dd, "\gvcpcnf{\cdot}{r}"]
          \\
          \hcp{\tm{\gvcpcnf{\config{C}}{r} \parallel \gvcpcnf{\config{D}}{r}}}
          \arrow[d, "\hcp{\slto[\alpha]{\beta^+}\bis_\ta} (\text{IH})"]
          \\
          \hcp{\tm{\gvcpcnf{\config{C}'}{r} \parallel \gvcpcnf{\config{D}}{r}}  }
          \arrow[r, "="]
          &
          \hcp{\tm{\gvcpcnf{\config{C}' \parallel \config{D}}{r}}}
        \end{tikzcd}
      \end{mathpar}
    \end{case}

    \begin{case}{\rulename{E-Equiv}}
      \begin{mathpar}
        \begin{tikzcd}[cramped, column sep=huge]
            \tm{\config{C}}
            \arrow[r, "\equiv"]
            \arrow[d, "\gvcpcnf{\cdot}{r}"]
            &
            \tm{\config{C}'}
            \arrow[r, "\hgv\cred"]
            &
            \tm{\config{D}'}
            \arrow[r, "\equiv"]
            &
            \tm{\config{E}}
            \arrow[ddd, "\gvcpcnf{\cdot}{r}"]
            \\
            \tm{\gvcpcnf{\config{C}}{r}}
            \arrow[d, "\hcp{\bis_\ta} (\text{Lemma~\ref{cor:hcp-sbis-to-bis}})"]
            \\
            \tm{\gvcpcnf{\config{C'}}{r}}
            \arrow[d, "\hcp{\slto[\alpha]{\beta^+}\bis_\ta} (\text{IH})"]
            \\
            \tm{\gvcpcnf{\config{D'}}{r}}
            \arrow[rrr, "\hcp{\bis_\ta} (\text{Lemma~\ref{cor:hcp-sbis-to-bis}})"]
            &
            &
            &
            \tm{\gvcpcnf{\config{E}}{r}}
        \end{tikzcd}
      \end{mathpar}
    \end{case}

    \begin{case}{\rulename{E-Lift-M}}
      The cases for $\hgv{\phi = \main}$ and $\hgv{\phi} = \hgv{\child}$ are similar; here we show the case for $\hgv{\main}$.
      \begin{mathpar}
        \begin{tikzcd}[cramped, column sep=8em]
          \hgv{\tm{\main M}}
          \arrow[r, "\hgv{\cred}"]
          \arrow[d, "\gvcpcnf{\cdot}{r}"]
          &
          \hgv{\tm{\main N}}
          \arrow[d, "\gvcpcnf{\cdot}{r}"]
          \\
          \hcp{\tm{\gvcpcom{M}{r}}}
          \arrow[r, "\hcp{\slto[\alpha]{\beta^+}\bis_\ta} (\text{Lemma~\ref{thm:hgv-to-hcp-oc-terms}})"]
          &
          \hcp{\tm{\gvcpcom{N}{r}}}
        \end{tikzcd}
      \end{mathpar}
    \end{case}

\item
    Reflection of $\alpha$-transitions is trivial as
    $\alpha$-transition is included in $\alpha$-bisimulation.
Reflection of $\beta$-transitions is by induction on
    $\hgv{\tm{C}}$; as with Lemma~\ref{thm:hgv-to-hcp-oc-terms}, the
    only well-typed $\beta$-transitions that can occur for each case
    are those specified in the simulation case. \qedhere

\end{enumerate}
\end{proof}

 }
%\clearpage
}
\end{document}